\PassOptionsToPackage{unicode}{hyperref}
\PassOptionsToPackage{hyphens}{url}
\PassOptionsToPackage{dvipsnames,svgnames,x11names}{xcolor}
\documentclass[
  12pt]{article}
\usepackage{setspace}
\usepackage{amsmath,amssymb}
\usepackage{iftex}
\ifPDFTeX
  \usepackage[T1]{fontenc}
  \usepackage[utf8]{inputenc}
  \usepackage{textcomp}
\else
  \usepackage{unicode-math}
  \defaultfontfeatures{Scale=MatchLowercase}
  \defaultfontfeatures[\rmfamily]{Ligatures=TeX,Scale=1}
\fi
\usepackage{lmodern}
\ifPDFTeX\else
\fi
\IfFileExists{upquote.sty}{\usepackage{upquote}}{}
\IfFileExists{microtype.sty}{%
  \usepackage[]{microtype}
  \UseMicrotypeSet[protrusion]{basicmath}
}{}
\makeatletter
\@ifundefined{KOMAClassName}{%
  \IfFileExists{parskip.sty}{%
    \usepackage{parskip}
  }{%
    \setlength{\parindent}{0pt}
    \setlength{\parskip}{6pt plus 2pt minus 1pt}}
}{%
  \KOMAoptions{parskip=half}}
\makeatother
\usepackage{xcolor}
\makeatletter
\ifx\paragraph\undefined\else
  \let\oldparagraph\paragraph
  \renewcommand{\paragraph}{
    \@ifstar
      \xxxParagraphStar
      \xxxParagraphNoStar
  }
  \newcommand{\xxxParagraphStar}[1]{\oldparagraph*{#1}\mbox{}}
  \newcommand{\xxxParagraphNoStar}[1]{\oldparagraph{#1}\mbox{}}
\fi
\ifx\subparagraph\undefined\else
  \let\oldsubparagraph\subparagraph
  \renewcommand{\subparagraph}{
    \@ifstar
      \xxxSubParagraphStar
      \xxxSubParagraphNoStar
  }
  \newcommand{\xxxSubParagraphStar}[1]{\oldsubparagraph*{#1}\mbox{}}
  \newcommand{\xxxSubParagraphNoStar}[1]{\oldsubparagraph{#1}\mbox{}}
\fi
\makeatother

\usepackage{longtable,booktabs,array}
\usepackage{calc}
\usepackage{etoolbox}
\makeatletter
\patchcmd\longtable{\par}{\if@noskipsec\mbox{}\fi\par}{}{}
\makeatother
\IfFileExists{footnotehyper.sty}{\usepackage{footnotehyper}}{\usepackage{footnote}}
\makesavenoteenv{longtable}
\usepackage{graphicx}
\makeatletter
\def\maxwidth{\ifdim\Gin@nat@width>\linewidth\linewidth\else\Gin@nat@width\fi}
\def\maxheight{\ifdim\Gin@nat@height>\textheight\textheight\else\Gin@nat@height\fi}
\makeatother
\setkeys{Gin}{width=\maxwidth,height=\maxheight,keepaspectratio}
\makeatletter
\def\fps@figure{htbp}
\makeatother

\makeatletter
\@ifpackageloaded{caption}{}{\usepackage{caption}}
\AtBeginDocument{%
\ifdefined\contentsname
  \renewcommand*\contentsname{Table of contents}
\else
  \newcommand\contentsname{Table of contents}
\fi
\ifdefined\listfigurename
  \renewcommand*\listfigurename{List of Figures}
\else
  \newcommand\listfigurename{List of Figures}
\fi
\ifdefined\listtablename
  \renewcommand*\listtablename{List of Tables}
\else
  \newcommand\listtablename{List of Tables}
\fi
\ifdefined\figurename
  \renewcommand*\figurename{Figure}
\else
  \newcommand\figurename{Figure}
\fi
\ifdefined\tablename
  \renewcommand*\tablename{Table}
\else
  \newcommand\tablename{Table}
\fi
}
\@ifpackageloaded{float}{}{\usepackage{float}}
\floatstyle{ruled}
\@ifundefined{c@chapter}{\newfloat{codelisting}{h}{lop}}{\newfloat{codelisting}{h}{lop}[chapter]}
\floatname{codelisting}{Listing}

\makeatother
\makeatletter
\@ifpackageloaded{caption}{}{\usepackage{caption}}
\@ifpackageloaded{subcaption}{}{\usepackage{subcaption}}
\makeatother

\ifLuaTeX
  \usepackage{selnolig}
\fi
\usepackage[]{natbib}
\usepackage{bookmark}

\IfFileExists{xurl.sty}{\usepackage{xurl}}{}
\hypersetup{
  pdftitle={A Unified Framework for Inference with General Missingness Patterns and Machine Learning Imputation},
  pdfauthor={Xingran Chen; Tyler McCormick; Bhramar Mukherjee; Zhenke Wu},
  pdfkeywords={missing data, prediction, statistical inference, Z-estimation},
  colorlinks=true,
  linkcolor={blue},
  filecolor={Maroon},
  citecolor={Blue},
  urlcolor={Blue}}

\usepackage{amsmath,amsthm,amssymb, graphicx, multicol, array, bm, bbm}
\usepackage{graphicx}
\usepackage{caption}
\usepackage{subcaption}
\usepackage{natbib}
\usepackage{algorithm}
\usepackage[utf8]{inputenc}
\usepackage[margin=1in]{geometry}
\usepackage{booktabs}
\usepackage{array}
\usepackage{amsmath, amsfonts, amssymb}
\usepackage{color}
\usepackage{bbm}
\usepackage[noend]{algpseudocode}
\usepackage{subcaption}
\usepackage{graphbox}
\usepackage{accents}
\usepackage{algorithm}
\usepackage{pdflscape}
\usepackage{lscape}
\usepackage{multirow}
\usepackage{cleveref}
\usepackage{pifont}

\newcommand{\cmark}{\textcolor{black}{\ding{51}}}
\newcommand{\xmark}{\textcolor{darkred}{\ding{55}}}

\newcommand{\vDelta}{\mathbf{\Delta}}
\newcommand{\vzeta}{{\bm{\zeta}}}
\newcommand{\vnu}{{\bm{\nu}}}
\newcommand{\vYimp}{\widetilde{\vY}}

\newcommand{\obs}{\text{obs}}
\newcommand{\mis}{\text{mis}}
\newcommand{\cc}{\text{CC}}
\newcommand{\ic}{\text{IC}}
\newcommand{\opt}{\text{opt}}
\newcommand{\psppi}{\text{PS-PPI}}
\newcommand{\mpsppi}{\text{MPS-PPI}}
\newcommand{\mi}{\text{MI}}
\newcommand{\wcc}{\text{WCC}}
\newcommand{\full}{\text{Full}}

\definecolor{darkred}{RGB}{180,0,0}

\usepackage{amsmath,amsfonts,bm}

\def\vzero{{\bm{0}}}
\def\vtheta{{\bm{\theta}}}
\def\vomega{{\bm{\omega}}}
\def\vpsi{{\bm{\psi}}}
\def\vphi{{\bm{\phi}}}
\def\vomega{{\bm{\omega}}}
\def\vgamma{{\bm{\gamma}}}
\def\vc{{\bm{c}}}
\def\vq{{\bm{q}}}
\def\vz{{\bm{z}}}

\def\mI{{\bm{I}}}

\newcommand{\E}{\mathbb{E}}
\newcommand{\R}{\mathbb{R}}
\newcommand{\Var}{\mathrm{Var}}
\newcommand{\Cov}{\mathrm{Cov}}

\renewcommand{\cal}{\mathcal}
\newcommand{\cC}{{\cal C}}
\newcommand{\cY}{{\cal Y}}
\newcommand{\cD}{{\cal D}}
\newcommand\cH{{\mathcal H}}
\newcommand{\cN}{{\cal N}}
\newcommand{\cS}{{\mathcal S}}

\renewcommand{\b}{\mathbb}

\renewcommand{\leq}{\leqslant}
\renewcommand{\geq}{\geqslant}
\newcommand{\wt}{\widetilde}

\newtheorem{theorem}{Theorem}
\newtheorem{proposition}{Proposition}

\newtheorem{corollary}{Corollary}[section]
\newtheorem{remark}{Remark}[section]
\newtheorem{assumption}{Assumption}

\newtheorem{example}{Example}

\let\Pr\relax
\DeclareMathOperator*{\Pr}{\mathbb{P}}

\renewcommand{\b}{\mathbb}

\def\vzero{\boldsymbol{0}}
\def\vY{\boldsymbol{Y}}
\def\vX{\boldsymbol{X}}
\def\vZ{\mathbf{Z}}

\newcommand{\Ab}{\mathbf{A}}
\newcommand{\Bb}{\mathbf{B}}
\newcommand{\Ib}{\mathbf{I}}
\newcommand{\Jb}{\mathbf{J}}
\newcommand{\Kb}{\mathbf{K}}
\newcommand{\Mb}{\mathbf{M}}
\newcommand{\Tb}{\mathbf{T}}
\newcommand{\Wb}{\mathbf{W}}
\newcommand{\Xb}{\mathbf{X}}

\newcommand{\bD}{\bm{D}}
\newcommand{\bQ}{\bm{Q}}
\newcommand{\bR}{\bm{R}}
\newcommand{\bV}{\bm{V}}
\newcommand{\bX}{\bm{X}}

\newcommand{\bpi}{\bm{\pi}}

\begin{document}

\addtocontents{toc}{\protect\setcounter{tocdepth}{-1}}

\def\spacingset#1{\renewcommand{\baselinestretch}%
{#1}\small\normalsize} \spacingset{1}

\title{\bf A Unified Framework for Inference with General Missingness Patterns and Machine Learning Imputation}
\author{Xingran Chen\\
  Department of Biostatistics, University of Michigan\\
  and \\
  Tyler McCormick \\
  Department of Statistics, University of Washington\\
  and \\
  Bhramar Mukherjee \\
  Department of Biostatistics, Yale University\\
  and \\
  Zhenke Wu \\
  Department of Biostatistics, University of Michigan}

\date{}
\maketitle

\bigskip
\begin{abstract}
Predictions from machine learning (ML) models are increasingly used to impute missing data, but their naive use risks biased inference. Existing methods provide valid inference with ML imputations regardless of prediction quality to enhance efficiency, while they are limited to either simple missing data structures or a missing-completely-at-random assumption. We develop a novel framework for valid statistical inference in Z-estimation problems using ML imputations under a missing-at-random assumption and for general missingness patterns. The method stratifies data by distinct missingness patterns and constructs an estimator by aggregating pattern-specific observed and prediction information. Theoretically, we establish the asymptotic properties and globally optimal aggregation that guarantees efficiency dominance over weighted complete-case analyses (WCCA). Practically, the method affords simple implementations by leveraging existing WCCA software. Extensive simulations validate theoretical results and show clear competitiveness in various scenarios. An analysis of \textit{All of Us} data further demonstrates the utility of the method. The paper concludes with a brief discussion on practical implications and potential future directions.
\end{abstract}

\noindent%
\vfill

\newpage
\spacingset{1.2}

\section{Introduction}\label{sec-intro}

With the rapid development of machine learning (ML) and artificial intelligence (AI), models trained on large-scale datasets have become increasingly capable of making highly accurate predictions. These predictions are often used as imputations for missing outcomes ($Y$) or covariates ($\bX$) due to nonresponse, cost, or other factors. Because ML-imputed datasets can be seamlessly applied in downstream statistical analyses, such impute-then-analyze procedures have become popular across disciplines such as finance~\citep{bryzgalova2025missing}, social science~\citep{egami2023using}, psychology~\citep{little2024missing}, biomedicine~\citep{barrabes2024advances}, and public health~\citep{perkins2018principled}. However, direct imputations may distort conditional or marginal distributions. Proceeding with analysis without careful consideration may result in biased estimation and invalid inference.

Recent research has explored prediction-based (PB) inference methods, which incorporate ML predictions for single imputation while guaranteeing valid statistical inference on datasets with partially missing outcomes (see, for example,~\citet{wang2020methods, angelopoulos2023prediction, egami2023using, salerno2025moment}, or \citet{salerno2025ipd}). Specifically, prediction-powered inference~\citep[PPI,][]{angelopoulos2023prediction} is a framework where inference is conducted using ML-imputed outcomes, with model parameter estimates calibrated, or ``rectified'', using a complete-case subset. Since this approach does not impose assumptions on the imputation model itself, it allows for a broad selection of ML-based imputation methods, the quality of which impacts the efficiency but not the validity of inference (e.g., consistency, achieving nominal coverage rate). Several extensions of this framework have adapted the idea of PPI to important areas, including causal inference~\citep{demirel2024prediction}, Bayesian inference~\citep{hofer2024bayesian, rovckova2025ai}, and genome-wide association studies~\citep[GWAS,][]{miao2024valid}. Furthermore, recognizing that the PPI approach does not guarantee a smaller variance relative to complete-case analysis (CCA), a few authors have proposed weighting the rectifier to achieve better efficiency than CCA~\citep{angelopoulos2023PPI++, gan2024prediction, gronsbell2024another, ji2025predictions, miao2025assumption}.

Despite the tremendous progress, two key challenges remain unaddressed. First, existing methods assume data are missing completely at random (MCAR), neglecting the more realistic missing-at-random (MAR) setting~\citep{rubin1976inference, little2019statistical}. Second, current literature focuses on missing outcomes or missing multiple covariates, overlooking general missingness patterns where both $Y$ and $\bX$ may be missing in arbitrary combinations. For example, in regression analysis, missing values are common in both outcomes and covariates, and we may observe different patterns of missingness in different subjects, e.g., missing $Y$ and $X_1$ for one subject and missing $Y$ and $X_2$ for the other. Although recent literature has proposed solutions for one of the challenges~\citep{mccaw2024synthetic, xu2025blockwise, zhao2025imputation, kluger2025prediction}, no work resolves them simultaneously.

\paragraph{Main Contributions} Methodologically, to use possibly unreliable predictions for reliable statistical inference, we propose a unified framework called Pattern-Stratified PPI (PS-PPI) to account for both general missingness patterns and the missing-at-random mechanism for general Z-estimation problems. Theoretically, we establish asymptotic properties of the proposed PS-PPI estimator. In addition, we provide globally optimal tuning parameter values that maximize the asymptotic efficiency, which is guaranteed to dominate the weighted complete-case analysis (WCCA) regardless of prediction quality. We also theoretically  connect PS-PPI to the existing PB inference methods,  semiparametric theory, and multiple imputation (MI). Practically, the framework affords simple implementation using existing statistical software capable of solving weighted complete-case Z-estimation problems. Table~\ref{tab:summary_literature} shows the broader scope and advantages of our work by comparing with popular alternatives.

\begin{table}[!tbp]
\centering
\caption{Comparison of PS-PPI with existing prediction-based (PB), multiple imputation (MI), and complete-case (CC)-based methods. For each method, the table indicates its assumed missing data mechanism, whether pre-trained machine learning predictions are used, the missingness patterns accommodated, and whether it has a theoretical efficiency improvement guarantee relative to the (weighted) complete-case counterpart (denoted by ``-'' as reference). Texts in {\color{darkred}red} indicate more restrictive assumptions or scenarios compared to our work.}
\label{tab:summary_literature}
\resizebox{\linewidth}{!}{%
\begin{tabular}{cccccccc}
\toprule
\multirow{3}{*}{\shortstack{Missing data\\mechanism}}
& \multirow{3}{*}{\shortstack{Method\\Category}}
& \multirow{3}{*}{Method}
& \multirow{3}{*}{\shortstack{Externally trained \\ ML models for imputation}}
& \multicolumn{3}{c}{Accommodated missingness}
& \multirow{3}{*}{\shortstack{Efficiency\\guarantee}} \\
\cmidrule(lr){5-7}
& & & & \shortstack{Outcome \\ missingness } & \shortstack{Missing covariates \\ (1 pattern)} & \shortstack{Missing covariates \\ ($\ge$ 1 patterns)} & \\
\midrule
\multirow{7}{*}[-0.6ex]{{\color{darkred}MCAR}}
& CC
& Complete-Case Analysis (CCA)     & \xmark & \cmark & \cmark & \cmark & - \\
\cmidrule(l){2-8}
& \multirow{6}{*}{PB}
& PPI~\citep{angelopoulos2023prediction}    & \cmark & \cmark & \xmark & \xmark & \xmark \\
& & PPI\texttt{++}~\citep{angelopoulos2023PPI++}   & \cmark & \cmark & \xmark & \xmark & \cmark \\
& & PDC~\citep{gan2024prediction}     & \cmark & \cmark & \xmark & \xmark & \cmark \\
& & PSPA~\citep{miao2025assumption}    & \cmark & \cmark & \xmark & \xmark & \cmark \\
& & IPI~\citep{zhao2025imputation}     & \cmark & \cmark & \cmark & \cmark & \cmark \\
& & IBM~\citep{xu2025blockwise}    & \cmark & \cmark & \cmark & \cmark & \cmark \\
\midrule
\multirow{5}{*}[-1.1ex]{MAR}
& CC
&  Weighted Complete-Case Analysis (WCCA)    & \xmark & \cmark & \cmark & \cmark & - \\
\cmidrule(l){2-8}
& MI
& MI~\citep{rubin1987multiple} & \xmark & \cmark & \cmark & \cmark & \xmark \\
\cmidrule(l){2-8}
& \multirow{3}{*}{PB}
& SynSurr~\citep{mccaw2024synthetic} & \cmark & \cmark & \xmark & \xmark & \xmark \\
& & PTD~\citep{kluger2025prediction}     & \cmark & \cmark & \cmark & \xmark & \cmark \\
& & \textbf{PS-PPI (ours) }
& \cmark & \cmark & \cmark & \cmark & \cmark \\
\bottomrule
\end{tabular}}
\end{table}

\paragraph{Paper Organization} The remainder of this paper is organized as follows. Section~\ref{sec:preliminary} introduces the notation, assumptions, and inferential targets. Section~\ref{sec:rationale_ppi} reviews the core idea of the PB inference methods and their limitations. Section~\ref{sec:methodology} presents PS-PPI, its theoretical properties, and connections to existing methods, semiparametric theory, and MI. In Section~\ref{sec:experiments}, a simulation study is presented to verify the theoretical validity of PS-PPI, followed by a real-data application using the {\it All of Us} data in Section~\ref{sec:data_example}. In Section~\ref{sec:discussion}, we conclude the paper with both its contribution and the limitations of the proposed method.
\section{Notation}
\label{sec:preliminary}
In this section, we introduce the notation used throughout this work. A glossary of the key notation can also be found in Section~\ref{sec:notation_glossary} of the Supplementary Materials.

Let $\mathcal{D} = \{\boldsymbol{Y}_{1}, \ldots, \boldsymbol{Y}_{N}\}$ denote the full data for $N$ observations, subject to potential missingness, where $\boldsymbol{Y}_{i} = (Y_{i1}, \ldots, Y_{ip})^\top \in \cY = \cY_1 \times \cdots \times \cY_p$ represents the $p$-dimensional data for subject $i$, $\cY_j\subseteq \b{R}$ is the marginal sample space of the $j$-th component in $\vY_i$. Each $\boldsymbol{Y}_{i}$ is assumed to be independently and identically sampled from a population distribution $\b{P}_{\vY}$. In practice, three types of variables may be included in the vector $\vY_i$: (a) a variable treated as the scalar outcome in a specific scientific inquiry; (b)  $q - 1$ variables treated as covariates (e.g., main exposure and potential confounders); (c) $p-q$ fully-observed auxiliary variables not included in the scientific analysis but potentially useful to predict variables in type (a) and (b). In a real data analysis, it is possible that one or more variables of type (b) are also not subject to missingness, and therefore can be used along with type (c) variables as inputs to ML prediction models for imputations. We partition $\vY_i$ into three types of variables for ease of presentation, and it does not impact the validity of our theoretical results. For readers who require this distinction in their scientific inquiries (e.g., via a regression analysis), one may assume that $Y_{i1}$ and \{$Y_{i2}, \ldots, Y_{iq}$\} represent the outcome and the covariates, respectively.  We assume $\cD$ to be partially observed. Let $\bR_i = (R_{i1}, \ldots, R_{ip})^\top \in  \{0,1\}^p$ denote the missingness indicators, where $R_{ij}=1$ if $Y_{ij}$ is observed and $0$ otherwise. For the $i$-th subject, we define the observed and missing components as $\vY_{\obs, i} = \{Y_{ij} | R_{ij} = 1, j = 1, \ldots, p \}$ and $\vY_{\mis, i} = \{Y_{ij} | R_{ij} = 0, j = 1, \ldots, p \}$, respectively.
\paragraph{Missingness patterns and mechanisms} In the existing literature on PB inference, typically only the outcome variable is subject to potential missingness \citep[e.g.,][]{angelopoulos2023prediction, angelopoulos2023PPI++, gan2024prediction, gronsbell2024another, miao2025assumption}, with the exceptions of~\cite{kluger2025prediction}, which focused on the missing-multiple-covariate pattern, and the concurrent works of~\citet{xu2025blockwise} and~\citet{zhao2025imputation} that focus on general missingness patterns under MCAR. However, in practice, missingness can occur anywhere within the dataset, e.g., for both outcomes and covariates under the regression setting. To address a more general missingness structure, we use the concept of missingness patterns and formalize the general missingness pattern problem by treating each missingness pattern as a special case of \textit{coarsening}~\citep{heitjan1991ignorability}, following~\cite{tsiatis2006semiparametric}. A missingness pattern is defined as an element in the sample space $\{0, 1\}^p$ of the missing indicator vector $\bR_i$. Excluding the case where all variables are missing, there are at most $2^p - 1$ different missingness patterns appearing in the data. In the context of PB inference, it is generally assumed that there exist fully observed, low-cost variables that serve as inputs for predicting the missing values (i.e., type (c) variables). Under this assumption, the number of distinct missingness patterns is reduced to at most $2^q$. Each of the missingness patterns can be indexed with $k \in \{1, ..., K\}$; we define $k = \infty$ to indicate that the subject is fully observed. The $k$-th missingness pattern is denoted as $\bR^{(k)}=(R^{(k)}_1, \ldots, R^{(k)}_p) \in \{0, 1\}^p$. We say that the $i$-th subject $\vY_i$ is in the $k$-th stratum of missingness pattern if $\bR_i = \bR^{(k)}$. Following~\cite{tsiatis2006semiparametric}, for the $k$-th missingness pattern, we introduce a deterministic and many-to-one coarsening function $G_k: \cY \mapsto \prod_{\{j: R^{(k)}_j = 1\}} \cY_j$, such that $G_k(\vY_i) = \vY_{\obs, i}$ if $\vY_i$ is under the $k$-th missingness pattern. For fully observed data, $G_\infty(\vY_i)$ is an identity mapping. Note that the coarsening notation is borrowed only to emphasize the constructive nature of our proposed method; the PB inference for general coarsened data where $G_k$ is broader than entry subset selection is beyond the scope of this work. Accordingly, we define the observed dataset as $\cD_{\obs} = \{(G_{\cC_i}(\boldsymbol{Y}_i), \cC_i)\}^N_{i=1}$, where $\cC_i \in \{1, \dots, K, \infty\}$ is the missingness index for subject $i$.

Missing data mechanisms play an important role in dictating how researchers should deal with the missingness in data analyses~\citep{rubin1976inference}. Here we provide a brief overview for completeness and refer readers to \cite{little2019statistical} for a comprehensive, book-length treatment of this topic. Based on how the conditional distribution of the missingness index $\cC_i$ given $\boldsymbol{Y}_i$ can be simplified, a few general categories with distinct analytic implications can be defined \citep{tsiatis2006semiparametric}.
In the simplest case with the most restrictive assumption of missing completely at random, the probability distribution of the missingness index $\cC_i$ does not depend on any observed or unobserved variables, namely, $P(\cC_i | \boldsymbol{Y}_i; \vomega) = P(\cC_i; \vomega)$
where $\vomega$ denotes the parameters of the conditional distribution.  A more common and realistic case compared to MCAR is that the probability of a missingness pattern only depends on observed variables: $P(\cC_i | \boldsymbol{Y}_i; \vomega) = P(\cC_i | \boldsymbol{Y}_{\obs, i}; \vomega)$, referred to as missing at random. If the missingness depends on the missing values themselves: $P(\cC_i | \boldsymbol{Y}_i; \vomega) = P(\cC_i | \boldsymbol{Y}_{\obs, i}, \boldsymbol{Y}_{\mis, i}; \vomega),$
it is referred to as missing not at random (MNAR). In practice, the MAR assumption is more commonly considered because it does not suffer from the identifiability problems that exist in the MNAR assumption, yet it is weaker than the MCAR assumption~\citep{tsiatis2006semiparametric}. In this paper, we assume that the missing data mechanism is MAR and develop our method under this assumption, which is formally stated in Assumption~\ref{assumption:missing-data}.
\begin{assumption}
    \label{assumption:missing-data}
    The following assumption is imposed on the missing data:
    \begin{enumerate}
        \item \textbf{Missingness pattern:} Any components of $\vY_i$ may be missing with arbitrary missingness patterns. The missingness index $\cC_i \in \{1, \dots, K, \infty\}$ denotes the pattern of the $i$-th subject, where $\cC_i = \infty$ indicates that the subject is fully observed.
        \item \textbf{Missing data mechanism:} The data are missing at random, i.e., $\Pr(\cC_i = k \mid \vY_i) = \Pr\{\cC_i = k \mid G_k(\vY_i)\}$ for every $k \in \{1, \dots, K, \infty\}$. In addition, we impose a positivity assumption such that, for every $k \in \{1, \dots, K, \infty\}$, $\Pr(\cC_i = k \mid G_k(\vY_i))$ is bounded away from zero.
    \end{enumerate}
\end{assumption}
Section~\ref{sec:mcar-vs-mar} of the Supplementary Materials provides a concrete example that distinguishes MCAR from MAR under both outcome-missing and general missingness settings.
\paragraph{Prediction of missing or masked values} For each missingness pattern $k=1,\ldots, K$, suppose we have access to a prediction function $h_k$ trained on data external to $\cD_{\text{obs}}$ where $h_k: \prod_{\{j:\bR^{(k)}_j=1\}} \cY_j \mapsto \prod_{\{j:\bR^{(k)}_j=0\}} \cY_j$. A key ingredient in our proposed method is the following ``mask-and-impute'' operation: in addition to filling missing values for subjects with the $k$-th missingness pattern using $h_k$, we can also mask fully observed subjects according to pattern $k$ (i.e., $\{Y_{ij}: j\in \{j': R_{j'}^{(k)}=0\}\}$ for subject $i$ with $\cC_i =\infty$), and then apply the same $h_k$ to impute these ``pseudo'' missing entries.

In fact, to study PB inference, we will show in our theoretical analyses that it is central to generalize this notion of data masking-and-imputation by missingness pattern $k$ from the fully observed subjects ($\cC_i=\infty$) to the entire population. For each missingness pattern $k=1, \ldots, K$, we define the imputed random vector of the $i$-th subject $\vYimp_{k, i} = \left \{G_{k}(\vY_i), h_{k}\left (G_{k}(\vY_i) \right ) \right \}$ to be comprised of observed components $G_{k}(\vY_i)$ and masked-and-imputed components $h_{k}\left (G_{k}(\vY_i)\right)$, where $\vY_i$ follows the population \textit{full-data} distribution $\b{P}_{\vY}$. We assume that $\vYimp_{k, i}$ is rearranged into the original coordinate order and is distributed according to a population-level distribution $\b{P}_{\vYimp_k}$. We note two key points. First, $h_k$ can be deterministic or stochastic as long as there exists a fixed population distribution $\b{P}_{\vYimp_k}$ which we will assume in theoretical analyses (Section \ref{sec:methodology}). This assumption is less tenable if the prediction mappings $\{h_k\}$ shift over time, e.g., due to model updates that are continually pushed into online repositories to progressively impute different subsets of data.
Second, although the number of predictive functions $h_k$ can grow exponentially with the number of variables subject to missingness, the number of models that must be obtained or trained need not grow at the same rate. When always observed variables are available, they can serve as shared predictors in base predictive models, one for each variable subject to missingness. For the $k$-th missingness pattern, $h_k$ is then assembled from the base models corresponding to the variables missing under that pattern. A base model can be reused across patterns only if all its inputs are observed under each of those patterns. Section~\ref{sec:base-prediction-model} of the Supplementary Materials formalizes this construction and provides an example. This setting is realistic in electronic health record (EHR)-linked biobanks, where expensive biomarkers are incompletely observed but sociodemographic variables, inexpensive biomarkers, and genotype data are more widely available. In our analysis of the \textit{All of Us} data in Section~\ref{sec:data_example}, we follow this strategy, fitting a base predictive model for each variable subject to missingness on an independent model-building subset. Section~\ref{sec:multiple_predictive_functions} in the Supplementary Materials discusses extensions to when multiple choices are available for each $h_k$.
\paragraph{Inferential target} We focus on valid inference for a $d$-dimensional parameter $\vtheta$ in the general Z-estimation framework under general missingness patterns and the MAR assumption. Formally, the population parameters of interest $\vtheta^\ast$ are defined as the unique solution of
$\E_{{\vY}} [\vpsi(\boldsymbol{Y}_i; \vtheta)] = \boldsymbol{0}$,
where $\vpsi$ are known vector-valued maps to a typically $d$-dimensional space, as defined in~\cite{van2000asymptotic}. For independent and identically distributed, fully observed data $\cD$ of size $N$, the Z-estimator $\widehat{\vtheta}_{\text{Full}}$ solves $\sum^N_{i=1} \vpsi(\vY_i; \vtheta) = \vzero$.
Under certain regularity conditions introduced in~\citet{van2000asymptotic} and detailed in Section~\ref{appendix:proof:thrm1} in the Supplementary Materials, $\widehat{\vtheta}_{\text{Full}}$ is asymptotically normal with mean $\vtheta^\ast$ and variance $\dot{\vpsi}^{-1}_{\textrm{Full}} \E[\vpsi^\ast\vpsi^{\ast\top}]\left(\dot{\vpsi}^{-1}_{\textrm{Full}}\right)^{\top}$, where $\dot{\vpsi}_{\text{Full}} = \E \left [ \partial \vpsi / \partial \vtheta \Bigr|_{\vtheta=\vtheta^\ast} \right ]$ is an invertible $d\times d$ matrix and $\vpsi^\ast = \vpsi(\boldsymbol{Y}_i; \vtheta^\ast)$ with $\E[\vpsi^\ast\vpsi^{\ast\top}]<\infty$.

\section{Rationale of Prediction-Based Inference}
\label{sec:rationale_ppi}

In the presence of missing values, we will propose a general approach and provide the theoretical guarantee of valid inference in Section \ref{sec:methodology}. The proposed approach couples the insights from the original outcome-based PPI framework \citep{angelopoulos2023prediction,angelopoulos2023PPI++} and the classical augmented
inverse probability weighted (AIPW) estimators, of which outcome-based PPI is a special case \citep{gronsbell2024another}. It is useful to review the rationale behind the original outcome-based PPI approach, which is designed to incorporate additional information from observations with missing outcomes and the predicted outcomes. Technically, such estimators can be written as a difference between a base estimate (CCA estimate) $\widehat{\vtheta}_{\text{CC}}$ and an estimate of zero $\widehat{\vDelta}$:\footnote{The notation of $\widehat{\vDelta}$ in this paper differs from the one used in~\cite{angelopoulos2023prediction} which denotes a rectifier to correct the bias introduced by ML predictions. The current form, in our view, is more helpful for motivating the construction of more efficient estimators and for theoretical analyses.} $\widehat{\vtheta}_{\text{PP}} = \widehat{\vtheta}_{\text{CC}} - \widehat{\vDelta}$,
which will be more efficient than $\widehat{\vtheta}_{\text{CC}}$ while preserving consistency when $2\Cov(\widehat{\vtheta}_{\text{CC}},\widehat{\vDelta})-\Var(\widehat{\vDelta})$ is positive definite.

\begin{example}[PPI for mean outcome estimation] \label{example:mean}
\end{example}

We illustrate with a simple example based on population mean estimation of a scalar outcome $Y$ in the presence of outcome missingness introduced in \citet{angelopoulos2023prediction}. Suppose the data consist of two partitions of sizes $m$ and $M$ and are independently and identically distributed according to a population distribution with marginal mean outcome $\mu=\E[Y]$. For the first partition, we observe the outcome $Y^{\tt lab}_i$ and its corresponding ML prediction $\widehat{Y}^{\tt lab}_i$: $\cD_{\tt lab} = \{(Y^{\tt lab}_i, \widehat{Y}^{\tt lab}_i)\}^m_{i=1}$, which is referred to as the ``labeled subset''. For the second, often larger, partition of data, we do not observe the outcome but can obtain its ML prediction, i.e., $\cD_{\tt unlab} = \{\widehat{Y}^{\tt unlab}_i\}^M_{i=1}$, referred to as the ``unlabeled subset''. We assume the outcome is MCAR, i.e., the full data distribution of $Y$ between the two partitions is identical. For both partitions, the predictions are obtained from other observed covariates and prediction models that have been externally trained to predict $Y$ using these covariates. As such, the pair $(Y,\widehat{Y})$ has the same distribution between the labeled and unlabeled partitions, although $Y$ is missing for the unlabeled subset.

In PPI, the CCA estimator for $\mu$ is the empirical mean based on the complete subset of data: $\widehat{\theta}_{\text{CC}} = \bar{Y}^{\tt lab}$, and the corresponding estimate of zero is computed as $\widehat{\Delta} = \bar{\widehat{Y}}^{\tt lab} - \bar{\widehat{Y}}^{\tt unlab}$ because the predicted outcomes in the labeled and unlabeled partitions share the same distribution under MCAR. By rewriting $\widehat{\theta}_{\text{PP}} = \widehat{\theta}_{\text{CC}}-\widehat{\Delta}=\bar{\widehat{Y}}^{\tt unlab}-(\bar{\widehat{Y}}^{\tt lab}-\bar{Y}^{\tt lab})$, we have
$\Var(\widehat{\theta}_{\text{PP}}) =  \Var(\widehat{Y}^{\tt unlab}_i) / M + \Var(Y^{\tt lab}_i - \widehat{Y}^{\tt lab}_i) / m$,
where the denominator of the first term is determined by the sample size of the unlabeled subset $M$, while the numerator of the second term is determined by the accuracy of the outcome prediction. $\widehat{\theta}_{\text{PP}}$ will have a small variance when $M \gg m$ to let the second term dominate, and the machine learning predictions are precise enough such that the variance in the second term is small. To see this, when $M\gg m$, we have $\Var(\widehat{\theta}_{\text{PP}})\approx\Var(Y^{\tt lab}_i - \widehat{Y}^{\tt lab}_i) / m<\Var(\widehat{\theta}_{\text{CC}})=\Var(Y_i^{\tt lab})/m$ when the predictions are of good quality, or more precisely when the correlation satisfies $\text{Corr}(Y,\widehat{Y})>\frac{1}{2}\sqrt{\Var(\widehat{Y})/\Var(Y)}$. Section~\ref{sec:mean-example-derivation} provides the detailed derivation.

However, it is critical to understand when $\widehat{\vtheta}_{\text{PP}}$ is more efficient than the complete-case analysis in settings other than the mean estimation problem.~\citet{gronsbell2024another} analyzed the efficiency of the $\widehat{\vtheta}_{\text{PP}}$ by deriving its influence function under the linear regression setting, and connecting it to the class of AIPW estimators. Specifically,~\citet{gronsbell2024another} show that $\widehat{\vtheta}_{\text{PP}}$ is an AIPW estimator and is not the most efficient in its class. Although the optimal estimator is difficult to obtain, a weighted version of the PPI estimator was proposed to guarantee the efficiency improvement compared to the complete-case analysis~\citep{gronsbell2024another}. The estimator is of the form $\widehat{\vtheta}_{\text{Chen}} = \widehat{\vtheta}_{\text{CC}} - \widehat{\Wb}_{\text{Chen}}\widehat{\vDelta}$, where $\widehat{\Wb}_{\text{Chen}}$ is a tuning parameter estimated from the data, originally proposed in~\citet{chen2000unified} under a different double-sampling application context but the same setup. Furthermore, recent modifications to the $\widehat{\vtheta}_{\text{PP}}$ consider multiplying $\widehat{\vDelta}$ with a data-driven weight to seek efficiency improvement over
$\widehat{\vtheta}_{\text{CC}}$~\citep{angelopoulos2023PPI++, gan2024prediction, miao2025assumption}.

\subsection{Limitations of existing prediction-based inference methods}
We now discuss the limitations of the existing methods, which we will address. First, existing PB inference methods do not account for the missing-at-random mechanism. In Example 1, suppose $Y$ is the annual income response in a questionnaire. An auxiliary variable, the level of education, affects the missing probability of $Y$, such that respondents with a lower level of education tend to skip the response to the income question, leaving their $Y$ missing. In such a case, the CCA estimator $\widehat{\theta}_{\text{CC}}$ itself is biased, because it can only reflect the income levels of the population with higher education levels. Second, existing methods only consider missing outcomes or missing multiple covariates, although in practice, general missingness patterns commonly appear. For example, in responding to a questionnaire, respondents may randomly skip multiple questions, leading to general missingness in the resulting dataset; in a longitudinal study, participants may skip a visit and return in the next visit~\citep{sun2018inverse}. In such cases, existing PB inference methods can only be carried out by removing the incomplete covariates, or truncating a longitudinal trajectory up to when the first missingness occurred. Some recent works have addressed one of these challenges.
For example, SynSurr~\citep{mccaw2024synthetic} assumes the outcome to be MAR when the random noise of true and predicted outcomes follows a bivariate normal distribution. \citet{xu2025blockwise} and \citet{zhao2025imputation} are concurrent works that focus on general missingness patterns under MCAR. However, none of these works addresses both challenges simultaneously.

\section{Method}
\label{sec:methodology}

\subsection{Pattern-stratified prediction-powered inference}
\label{sec:psppi-method}
In this paper, we propose valid inference for a $d$-dimensional parameter $\vtheta$ in a given Z-estimation framework when external ML models are available to impute data with general missingness patterns. As discussed in Section~\ref{sec:rationale_ppi}, existing PB inference augments a base estimator with an estimate of zero constructed using the predicted values. We extend this idea to general missingness by augmenting the base estimator with one or more estimates of zero built from ML imputations. Below, we describe: 1) how the base estimator is obtained under the Z-estimation framework with the MAR assumption; 2) how to construct the estimates of zero with general missingness. The full PS-PPI procedure is summarized in Figure~\ref{fig:general_missing_pattern}.
\paragraph{Missing data model} We define the true missing data model as a conditional multinomial distribution $\bpi(\cdot; \vomega)=\{\pi_k(\cdot; \vomega),k=1,\ldots, K, \infty\} \in \R^{K + 1}$ parameterized by $\vomega \in \R^{d_{\pi}}$, where $\pi_k(\cdot; \vomega) = P(\cC_i = k | \cdot; \vomega)$ is the probability of observing the $k$-th missingness pattern and $\mathbf{1}^\top \bpi=1$. Denote the true parameter by $\vomega^\ast$. In the data generating process, $\cC_i$ is sampled from $\bpi(\vY_i, \vomega^\ast)$. Under MAR, the $k$-th element of $\bpi$ takes as input only $G_k(\vY_i)$, the variables observed under the $k$-th missingness pattern. Modeling missing data under MAR has been widely studied.~\citet{robins1994estimation} and~\citet{tsiatis2006semiparametric} proposed discrete hazard models for monotone missingness, and~\citet{sun2018inverse} model each missingness pattern with a submodel under general missingness. Identifying and estimating such a missing data model remains challenging, because under MAR different missingness patterns may depend on different subsets of variables. A full investigation is beyond the scope of this paper. Following \citet{sun2018inverse}, we assume that the missing data model is correctly specified when it is unknown, and that $\vomega^\ast$ can be estimated from the observed dataset $\cD_\obs$ by Z-estimation, i.e., $\widehat{\vomega}$ solves $\sum^N_{i=1}\vphi_i(\vomega) = \sum^N_{i=1}\vphi(\vY_{\obs, i}; \vomega) = \vzero$, where $\vphi$ has the same dimension as $\vomega$ and satisfies the regularity conditions in Section~\ref{appendix:proof:thrm1} of the Supplementary Materials.  Examples of missing data models that can be formulated as a Z-estimation problem are a class of missing data models for general missingness proposed in~\citet{sun2018inverse} and the randomized monotone missingness (RMM) model proposed by~\citet{robins1997non}, in both of which $\widehat{\vomega}$ can be obtained by solving the likelihood score equations. For the $i$-th subject, $\pi_k(G_{\cC_i}(\vY_i); \widehat{\vomega})$, which the base estimator and the estimates of zero require, can be computed from the observed data, although the probabilities of the other patterns may depend on missing variables.

\paragraph{Base estimator} When the missing data model is correctly specified and estimated, our method considers the inverse-probability-weighted complete-case (WCC) estimator $\widehat{\vtheta}_{\text{WCC}}$ solved from the below estimating equations as the base estimator:
    \begin{equation}
    \label{eq:general_missing_pattern_ee}
        \sum_{i = 1}^{N} \vpsi_{\text{WCC}, i}(\vtheta, \widehat{\vomega}) = \sum_{i = 1}^{N} \frac{I(\cC_i = \infty)}{\pi_\infty(\boldsymbol{Y}_i; \widehat{\vomega})} \vpsi \left (\boldsymbol{Y}_i; \vtheta \right ) = \vzero.
    \end{equation}
Equation (\ref{eq:general_missing_pattern_ee}) uses only the complete cases weighted by the inverse of their estimated probability of being fully observed. With certain standard regularity conditions and the consistency assumption of $\widehat{\vomega}$, we show in Section~\ref{sec:proof_proposition_wcc} of the Supplementary Materials that $\widehat{\vtheta}_{\text{WCC}}$ is consistent and asymptotically normal with mean $\vtheta^\ast$ and variance-covariance matrix $\Sigma_\vtheta = \E\left[ \varphi_{\infty,i}(\vtheta^\ast,\vomega^\ast) \varphi_{\infty,i}^{\top}(\vtheta^\ast,\vomega^\ast) \right]$, where
$\varphi_{\infty,i}(\vtheta^\ast,\vomega^\ast) = \widetilde{\varphi}_{\infty,i}(\vtheta^\ast,\vomega^\ast) - \E \left [\nabla_\vomega \widetilde{\varphi}_{\infty, i}(\vtheta^\ast,\vomega^\ast) \right ] \dot{\vphi}^{-1}(\vomega^\ast) \vphi_i(\vomega^\ast)$
is the $i$-th influence function of $\widehat{\vtheta}_{\text{WCC}}$, $\widetilde{\varphi}_{\infty,i}(\vtheta,\vomega) = - \dot{\vpsi}_\infty^{-1}(\vtheta,\vomega) \vpsi_{\mathrm{WCC},i}(\vtheta,\vomega)$, $\dot{\vpsi}_\infty(\vtheta, \vomega) = \E \left [\nabla_\vtheta\vpsi_{\text{WCC}, i}(\vtheta, \vomega) \right ]$ is the expectation of $\vpsi_{\text{WCC}, i}$ after taking the derivative with respect to $\vtheta$, and $\dot{\vphi}(\vomega) = \E \left [\nabla_\vomega\vphi(\vomega) \right ]$.
\begin{proposition}
    \label{proposition:wcc}
    Suppose $\widehat{\pi}(\cdot)$ is correctly specified and estimated under the MAR assumption. Under Assumption~\ref{assumption:main} in the Supplementary Materials, $\widehat{\vtheta}_{\text{WCC}}$ is asymptotically normal with mean $\vtheta^\ast$ and variance-covariance matrix $\Sigma_\vtheta = \E\left[ \varphi_{\infty,i}(\vtheta^\ast,\vomega^\ast) \varphi_{\infty,i}^{\top}(\vtheta^\ast,\vomega^\ast) \right]$.
\end{proposition}
\begin{figure}[!t]
    \centering
    \includegraphics[width=1.0\linewidth]{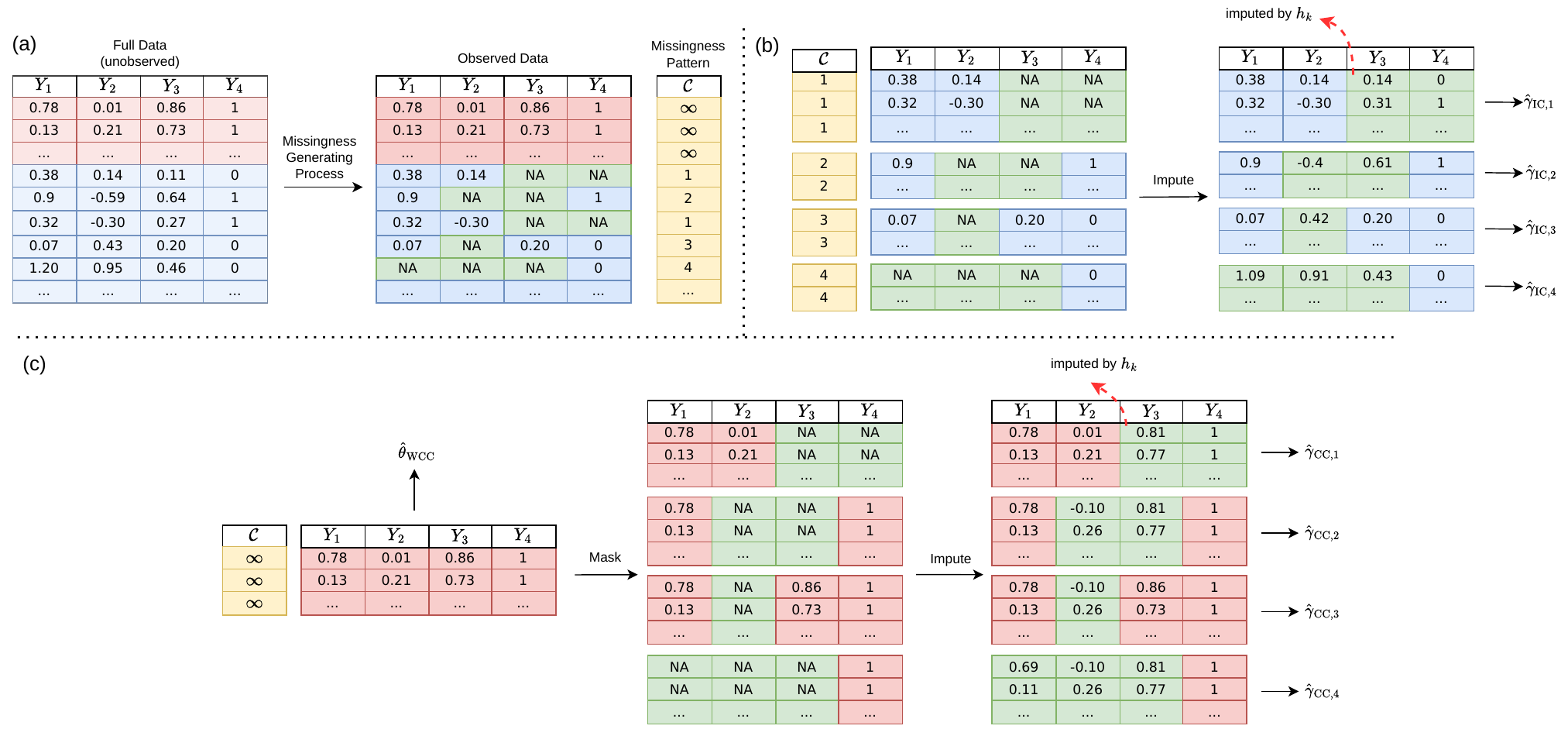}
    \caption{The PS-PPI procedure. (a) Samples are grouped and indexed according to their missingness patterns. (b) For each pattern $k$, samples with $\cC_i=k$ are imputed using $h_k$ to obtain $\widehat{\vgamma}_{\ic,k}$. (c) For fully observed samples, $\widehat{\vtheta}_{\text{WCC}}$ is computed. These samples are also masked by each pattern, respectively, and are imputed via $h_k$ to obtain $\widehat{\vgamma}_{\cc,k}$.}
    \label{fig:general_missing_pattern}
\end{figure}
\paragraph{Estimates of zero} Our approach stratifies the partially observed data by missingness pattern and, with ML imputation, constructs an estimate of zero within each pattern: $\widehat{\vDelta}_k = \widehat{\vgamma}_{\cc, k} - \widehat{\vgamma}_{\ic, k} \in \R^d$, the difference of two estimators computed from the fully observed subjects ($\cC_i = \infty$) and from the subjects with the $k$-th missingness pattern ($\cC_i = k$). To compute $\widehat{\vgamma}_{\cc, k}$, we mask the fully observed data with the $k$-th missingness pattern to create ``pseudo'' missingness and impute the masked entries with $h_k$, giving the imputed dataset $\{\vYimp_{k, i}: \cC_i = \infty \}$, where $\vYimp_{k, i} = \left \{G_{k}(\vY_i), h_{k}\left (G_{k}(\vY_i) \right ) \right \}$ is the imputed $i$-th subject under the $k$-th missingness pattern. Then $\widehat{\vgamma}_{\cc,k}$ solves the weighted estimating equations
\begin{equation}
    \label{eq:general_missing_pattern_ee_gamma_1}
    \sum_{i = 1}^{N} \vpsi^{(k)}_{\cc, i}(\vgamma_{\cc, k}, \widehat{\vomega}) = \sum^{N}_{i= 1}  \frac{I(\cC_i = \infty)}{\pi_\infty(\vY_{\obs, i}; \widehat{\vomega})}  \vpsi \left ( \vYimp_{k, i}; \vgamma_{\cc, k} \right ) = \vzero.
\end{equation}
As in Equation (\ref{eq:general_missing_pattern_ee}), Equation (\ref{eq:general_missing_pattern_ee_gamma_1}) weights by the inverse probability of being fully observed. For the subjects with $\cC_i = k$, the components missing under the $k$-th missingness pattern are genuinely missing, so we impute them directly with $h_k$, obtain the imputed dataset $\{\vYimp_{k, i}: \cC_i = k \}$, and solve $\widehat{\vgamma}_{\ic,k}$ from the weighted estimating equations
\begin{equation}
    \label{eq:general_missing_pattern_ee_gamma_2}
    \sum_{i = 1}^{N} \vpsi^{(k)}_{\ic,i}(\vgamma_{\ic, k}, \widehat{\vomega}) = \sum^{N}_{i=1} \frac{I(\cC_i = k)}{\pi_k(G_{k}(\boldsymbol{Y}_i); \widehat{\vomega})} \vpsi \left ( \vYimp_{k, i}; \vgamma_{\ic, k} \right ) = \vzero.
\end{equation}
Before establishing the asymptotic properties of $\widehat{\vDelta}_k$, we introduce Assumption~\ref{assumption:gamma_existence}.
\begin{assumption}
\label{assumption:gamma_existence}
For each $k=1,\ldots,K$, there exists a unique $\vgamma_k^\ast \in \R^d$ satisfying
$
\E_{\vY}\left[
\vpsi\left(\vYimp_{k,i};\vgamma_k^\ast\right)
\right]
=
\vzero,
$
where the expectation is taken over the distribution of the underlying full-data random vector $\vY_i$.
\end{assumption}
Specifically, the uniqueness statement in Assumption~\ref{assumption:gamma_existence} is proposed to guarantee 1) the existence of the population-level parameters $\vgamma^\ast_k$ and 2) the convergence of $\widehat{\vgamma}_{\cc, k}$ and $\widehat{\vgamma}_{\ic, k}$ to the same population-level root; it can be relaxed only under an additional assumption ensuring this common limit. The assumption is commonly imposed in the PB inference literature~\citep{gronsbell2024another, miao2025assumption, kluger2025prediction} and holds for many classical models, including generalized linear models with canonical links and the linear regression model used in our simulations (Section~\ref{sec:experiments}) and \textit{All of Us} application (Section~\ref{sec:data_example}). With Assumption~\ref{assumption:gamma_existence}, the $i$-th influence functions of $\widehat{\vgamma}_{\cc, k}$ and $\widehat{\vgamma}_{\ic, k}$ are
$$ \varphi_{\cc,i}^{(k)}(\vgamma^\ast_k,\vomega^\ast) = \widetilde{\varphi}_{\cc,i}^{(k)}(\vgamma^\ast_k,\vomega^\ast) - \E\left[ \nabla_{\vomega} \widetilde{\varphi}_{\cc,i}^{(k)}(\vgamma^\ast_k,\vomega^\ast) \right] \dot{\vphi}^{-1}(\vomega^\ast) \vphi_i(\vomega^\ast),$$
$$\varphi_{\ic,i}^{(k)}(\vgamma^\ast_k,\vomega^\ast) = \widetilde{\varphi}_{\ic,i}^{(k)}(\vgamma^\ast_k,\vomega^\ast) - \E\left[ \nabla_{\vomega} \widetilde{\varphi}_{\ic,i}^{(k)}(\vgamma^\ast_k,\vomega^\ast) \right] \dot{\vphi}^{-1}(\vomega^\ast) \vphi_i(\vomega^\ast),$$
respectively. Note that $\widetilde{\varphi}_{\cc,i}^{(k)}(\vgamma_k,\vomega) = -\dot{\vpsi}^{-1}_k(\vgamma_k,\vomega)\vpsi_{\cc,i}^{(k)}(\vgamma_k,\vomega)$ and $\widetilde{\varphi}_{\ic,i}^{(k)}(\vgamma_k,\vomega) = -\dot{\vpsi}^{-1}_k(\vgamma_k,\vomega)\vpsi_{\ic,i}^{(k)}(\vgamma_k,\vomega)$, and $\dot{\vpsi}_{k}(\vgamma_k,\vomega) = \E \left [\nabla_{\vgamma_k}\vpsi^{(k)}_{\cc, i}(\vgamma_k, \vomega) \right ]$. As detailed in the Supplementary Materials, the two bread matrices $ \E \left [\nabla_{\vgamma_k}\vpsi^{(k)}_{\cc, i}(\vgamma_k, \vomega) \right ]$ and $\E \left [\nabla_{\vgamma_k}\vpsi^{(k)}_{\ic, i}(\vgamma_k, \vomega) \right ]$ are equal at $\vomega = \vomega^\ast$. As a result, $\widehat{\vDelta}_k$ is consistent and asymptotically normal with mean $\vzero$ and asymptotic variance-covariance matrix $\Sigma_{\vDelta_k} = \E \left [ \varphi_{k, i}(\vgamma_k^\ast,\vomega^\ast) \varphi^\top_{k, i}(\vgamma_k^\ast,\vomega^\ast) \right ] $, where $\varphi_{k, i}(\vgamma_k^\ast,\vomega^\ast) = \varphi_{\cc,i}^{(k)}(\vgamma^\ast_k,\vomega^\ast) - \varphi_{\ic,i}^{(k)}(\vgamma^\ast_k,\vomega^\ast)$. We formally present this result in Proposition~\ref{proposition:estimate_of_zero}, with proof presented in Section~\ref{sec:proof_estimate_of_zero} of the Supplementary Materials.
\begin{proposition}[Estimate of zero]
\label{proposition:estimate_of_zero}
Suppose the missing data model $\bpi(\cdot;\vomega)$ is correctly specified and estimated under the MAR assumption. Under Assumption~\ref{assumption:gamma_existence} and Assumption~\ref{assumption:main} in the Supplementary Materials, for each $k=1,\ldots,K$, $\widehat{\vDelta}_k=\widehat{\vgamma}_{\cc,k}-\widehat{\vgamma}_{\ic,k}$ is asymptotically normal with mean $\vzero_d$ and variance-covariance matrix $\Sigma_{\vDelta_k} = \E\left[\varphi_{k,i}(\vgamma_k^\ast,\vomega^\ast)\varphi_{k,i}^{\top}(\vgamma_k^\ast,\vomega^\ast)\right]$.
\end{proposition}

Intuitively, this is because after weighting by the inverse of missingness probabilities, both subsets $\{\vYimp_{k, i}: \cC_i = \infty \}$ and $\{\vYimp_{k, i}: \cC_i = k \}$ follow the same distribution $\b{P}_{\vYimp_k}$. As a result, $\widehat{\vgamma}_{\cc,k}$ and $\widehat{\vgamma}_{\ic,k}$  are consistent for the same value $\vgamma_k^\ast$. This fact becomes useful when constructing estimates of zero, even though the specific value of $\vgamma_k^\ast$ may not be of scientific interest.

The PS-PPI estimator is defined by subtracting a weighted sum of all the estimates of zero $\widehat{\vDelta} = (\widehat{\vDelta}_1^\top, \ldots, \widehat{\vDelta}_K^\top)^\top$ from the base estimator $\widehat{\vtheta}_{\text{WCC}}$:
\begin{equation}
    \label{eq:general_missing_pattern_estimator}
    \widehat{\vtheta}_{\text{PS-PPI}} = \widehat{\vtheta}_{\text{WCC}} - \widehat{\Wb}\widehat{\vDelta} = \widehat{\vtheta}_{\text{WCC}} - \sum^K_{k=1} \widehat{\Wb}_k\widehat{\vDelta}_k.
\end{equation}
Here, $\widehat{\Wb} = (\widehat{\Wb}_1, \widehat{\Wb}_2, \dots, \widehat{\Wb}_K) \in \R^{d \times Kd}$. Each $\widehat{\Wb}_k \in \R^{d \times d}$ is specific to pattern $k$ and controls the contribution of $\widehat{\vDelta}_k$ to $\widehat{\vtheta}_{\text{PS-PPI}}$. Section~\ref{sec:theoretical_analysis} derives an optimal data-dependent choice $\widehat{\Wb}_{\opt}$ that maximizes the asymptotic efficiency of $\widehat{\vtheta}_{\text{PS-PPI}}$.
\subsection{Theoretical analysis}
\label{sec:theoretical_analysis}
In this subsection, we show that the proposed $\widehat{\vtheta}_{\text{PS-PPI}}$ is capable of providing valid statistical inference under general missingness patterns and the MAR assumption.

Let $\vgamma^\ast = (\vgamma^{\ast\top}_1, \dots, \vgamma^{\ast\top}_K)^\top \in \R^{Kd}$ be a vector of all the $K$ different unique solutions of the population-level estimating equations defined in Equations (\ref{eq:general_missing_pattern_ee_gamma_1}) and (\ref{eq:general_missing_pattern_ee_gamma_2}). We show in Theorem~\ref{thrm:general_missing_patterns} that $\widehat{\vtheta}_{\text{PS-PPI}}$ is asymptotically normal with mean equal to the true parameters of interest $\vtheta^\ast$ under certain regularity conditions detailed in Assumption~\ref{assumption:main} of the Supplementary Materials.
\begin{theorem}[Asymptotic normality of $\widehat{\vtheta}_{\text{PS-PPI}}$]
\label{thrm:general_missing_patterns}
     Suppose $\widehat{\pi}(\cdot)$ is correctly specified and estimated under the MAR assumption. For each $\widehat{\Wb}_k$ in $\widehat{\Wb}$, suppose there exists a constant matrix $\Wb_k$ such that $\widehat{\Wb}_k \overset{p}{\rightarrow} \Wb_k$, where $\overset{p}{\rightarrow}$ denotes convergence in probability. Define $\Wb = (\Wb_1, \dots, \Wb_K) \in \R^{d \times Kd}$. Under Assumptions~\ref{assumption:gamma_existence} and~\ref{assumption:main} in the Supplementary Materials, $\widehat{\vtheta}_{\text{PS-PPI}}$ is asymptotically normal with mean $\vtheta^\ast$ and variance-covariance matrix:
    \begin{equation}
        \label{eq:full_asymptotic_variance}
        \Sigma_{\text{PS-PPI}}(\vtheta^\ast, \vgamma^\ast, \vomega^\ast, \Wb) = \E \left [ \varphi_i(\vtheta^\ast, \vgamma^\ast, \vomega^\ast, \Wb) \varphi^\top_i(\vtheta^\ast, \vgamma^\ast, \vomega^\ast, \Wb) \right ],
    \end{equation}
    where $\varphi_i(\vtheta, \vgamma, \vomega, \Wb) = \varphi_{\infty, i}(\vtheta, \vomega) - \Wb \varphi_{\vDelta, i}(\vgamma, \vomega)$ is the $i$-th influence function of $\widehat{\vtheta}_{\text{PS-PPI}}$. Here, $\varphi_{\vDelta, i}(\vgamma, \vomega) = (\varphi_{1, i}^\top(\vgamma_1, \vomega), \dots, \varphi_{K, i}^\top(\vgamma_K, \vomega))^\top$.
\end{theorem}

The proof is detailed in Section~\ref{appendix:proof:thrm1} in the Supplementary Materials.
By further expanding Equation~(\ref{eq:full_asymptotic_variance}), the asymptotic
variance-covariance matrix can be written as
\begin{equation}
\label{eq:full_asymptotic_variance_expanded}
\begin{aligned}
\Sigma_{\text{PS-PPI}}(\vtheta^\ast, \vgamma^\ast, \vomega^\ast, \Wb)
&=
\Sigma_{\vtheta}
-
\Sigma_{\vtheta,\vDelta}\Wb^\top
-
\Wb\Sigma_{\vtheta,\vDelta}^\top
+
\Wb\Sigma_{\vDelta}\Wb^\top ,
\end{aligned}
\end{equation}
\begin{equation}
\label{eq:sigma_components}
\text{where}\quad
\Sigma_{\vtheta,\vDelta}
:=
\E\left[
\varphi_{\infty,i}(\vtheta^\ast,\vomega^\ast)
\varphi_{\vDelta,i}^{\top}(\vgamma^\ast,\vomega^\ast)
\right],
\;
\Sigma_{\vDelta}
:=
\E\left[
\varphi_{\vDelta,i}(\vgamma^\ast,\vomega^\ast)
\varphi_{\vDelta,i}^{\top}(\vgamma^\ast,\vomega^\ast)
\right].
\end{equation}
Here, $\Sigma_{\vtheta}$ is defined in Proposition~\ref{proposition:wcc};
\(\Sigma_{\vtheta,\vDelta}\) is the
asymptotic covariance matrix between \(\widehat{\vtheta}_{\mathrm{WCC}}\) and
\(\widehat{\vDelta}\), and \(\Sigma_{\vDelta}\) is the asymptotic variance-covariance matrix of \(\widehat{\vDelta}\).

\paragraph{Optimal tuning parameters} Note that in Theorem~\ref{thrm:general_missing_patterns}, $\Sigma_{\text{PS-PPI}}$ is also a function of the tuning parameter $\Wb$. It is natural to consider whether there exists a tuning parameter denoted as $\Wb_\opt$ such that, when $\vtheta^\ast$, $\vgamma^\ast$, and $\vomega^\ast$ are fixed, $\Sigma_{\text{PS-PPI}}(\vtheta^\ast, \vgamma^\ast, \vomega^\ast, \Wb_\opt) \preceq \Sigma_{\text{PS-PPI}}(\vtheta^\ast, \vgamma^\ast, \vomega^\ast, \Wb)$ for any $\Wb \in \R^{d \times Kd}$, where, for two square matrices $\Ab$ and $\Bb$, $\Ab\preceq \Bb$ means $\Bb-\Ab$ is positive semidefinite. We have the following corollary in terms of $\Wb_\opt$:
\begin{corollary}
    \label{corollary:optimal_w}
    There exists a $\Wb_\opt$ that leads to an optimal asymptotic variance-covariance matrix of $\widehat{\vtheta}_{\text{PS-PPI}}$, which is $\Wb_\opt = \Sigma_{\vtheta, \vDelta} \Sigma_{\vDelta}^{-1}$.
    As such, the asymptotic variance-covariance matrix of $\widehat{\vtheta}_{\text{PS-PPI}}$ with the optimal $\Wb_\opt$ is
    \begin{equation}
        \label{eq:optimal_asymptotic_variance}
        \Sigma_{\text{PS-PPI}}(\vtheta^\ast, \vgamma^\ast, \vomega^\ast, \Wb_\opt) = \Sigma_\vtheta - \Sigma_{\vtheta, \vDelta}\Sigma^{-1}_{\vDelta}\Sigma_{\vtheta, \vDelta}^{\top} \preceq \Sigma_{\vtheta},
    \end{equation}
    where $\Sigma_{\vtheta}$ is the asymptotic variance-covariance matrix of $\widehat{\vtheta}_{\text{WCC}}$.
\end{corollary}
In practice, a consistent estimator $\widehat{\Wb}_\opt$ can be obtained by replacing $\Sigma_{\vtheta, \vDelta}$ and $\Sigma_{\vDelta}$ with their consistent estimators. Section~\ref{sec:variance-estimation} describes variance estimation for PS-PPI, and Section~\ref{sec:incorporate_software} explains how existing weighted complete-case analysis software can be used for this purpose.

It is worth noting that with the introduction of $\Wb_\opt$, discarding any missingness patterns from $\widehat{\vtheta}_{\text{PS-PPI}}$ (i.e., let $\Wb_k = \vzero_{d \times d}$) will not further improve the asymptotic efficiency. We leave the detailed discussion to Section~\ref{sec:optimal_pattern_selection} of the Supplementary Materials. This indicates that including the missingness patterns, or more generally the estimates of zero, as many as possible is the optimal strategy for the analysts.

Applying the same principle of optimal aggregation, we provide an extension to accommodate multiple sets of predictions for each missingness pattern (see the formulation and results in Section \ref{sec:multiple_predictive_functions} in Supplementary Materials). In the following, we focus theoretical and empirical studies on a single set of predictions for each missingness pattern.

\subsection{Connecting PS-PPI to semiparametric theory}
\label{sec:psppi_semiparametric_efficiency}

In semiparametric theory, one may seek the most efficient estimator within the class of augmented inverse probability weighted complete-case (AIPWCC) estimators under specific settings~\citep{robins1994estimation, robins1995analysis, chaudhuri2016gmm}.
In particular,~\citet{xu2025unified} formulate the missing-outcome setting in a semi-supervised learning framework, derive the semiparametric efficiency lower bound for a given inferential problem, and propose both an \textit{efficient} estimator that attains this bound and a \textit{safe} estimator that is much easier to compute in practice. We prove in Section~\ref{sec:psppi_safe_xu} of the Supplementary Materials that the \textit{safe} estimator is a special case of PS-PPI, namely under the ordinary semi-supervised missing-outcome setting specified by~\citet{xu2025unified}, the optimally weighted PS-PPI estimator is asymptotically equivalent to their \textit{safe} estimator.
Under general missingness patterns and MAR, we prove in Section~\ref{sec:proof_psppi_aipwcc} of the Supplementary Materials that $\widehat{\vtheta}_{\text{PS-PPI}}$ belongs to the class of AIPWCC estimators defined in~\citet{tsiatis2006semiparametric}. Given predictive functions $h_k$ and the estimating equations $\vpsi$, PS-PPI considers the augmentation space
$$\Lambda_{\mathrm{PS\text{-}PPI}}
=
\left\{
\sum_{k=1}^K \Wb_k \dot{\vpsi}_k^{-1}
\left[
\frac{I(\cC=\infty)}{\pi_\infty(\vY; \vomega^\ast)}
\vpsi \left ( \vYimp_k;\vgamma_k^\ast \right )
-
\frac{I(\cC=k)}{\pi_k\{G_k(\vY);  \vomega^\ast\}}
\vpsi \left ( \vYimp_k;\vgamma_k^\ast \right )
\right]:
\Wb_k\in\mathbb R^{d\times d}
\right\},$$
where $\vYimp_{k} = \left \{G_{k}(\vY), h_{k}\left (G_{k}(\vY) \right ) \right \}$ is the imputed random vector. Then $\Lambda_{\mathrm{PS\text{-}PPI}} \subseteq \Lambda_2$, where $\Lambda_2$ denotes the augmentation space of the AIPWCC estimator defined in~\citet{tsiatis2006semiparametric}. This inclusion should be interpreted as a restriction on the allowable AIPWCC augmentation functions. In the full AIPWCC class, the pattern-specific augmentation $L_{2k}\{G_k(\vY)\}$ can be any square-integrable function of the variables observed under pattern $k$. In PS-PPI, once the predictive functions $h_k$ and the inferential estimating equations $\vpsi$ are fixed, the only free quantities are the constant matrices $\Wb_k$. Hence,
$\Wb_\opt$ yields the most efficient PS-PPI estimator within this restricted subspace, but generally not the semiparametric efficient AIPWCC estimator over the full space $\Lambda_2$ unless the efficient augmentation functions happen to lie in $\Lambda_{\mathrm{PS\text{-}PPI}}$. The main advantage of $\widehat{\vtheta}_{\text{PS-PPI}}$, however, lies in its ability to incorporate predictions from external predictive models. With the rapid development of deep learning and other machine learning techniques, highly precise ML predictions are becoming increasingly available. By working in an augmentation subspace $\Lambda_{\text{PS-PPI}}$ that explicitly incorporates the predictive functions $h_k$, PS-PPI enables analysts to combine these ML predictions with their statistical models in the presence of missing data. The space $\Lambda_\text{PS-PPI}$ is finite-dimensional. In contrast, semiparametric efficient AIPWCC estimators in $\Lambda_2$ are difficult to implement because they generally do not have a closed form and require iterative computation. We refer readers to~\citet{tsiatis2006semiparametric} and \citet{chaudhuri2016gmm} for a more in-depth discussion.

When the missing data model $\pi$ is unknown, another desirable theoretical property of the semiparametric efficient AIPWCC estimator within the full augmentation space $\Lambda_2$ is double robustness. As illustrated empirically in Section~\ref{sec:misspecification_robustness} of the Supplementary Materials, PS-PPI is not doubly robust. We kindly refer readers to Section~\ref{sec:psppi-doubly-robust} for a detailed discussion on constructing a doubly robust variant of PS-PPI.

\subsection{Theoretical comparison of PS-PPI and multiple imputation}
\label{sec:theoretical-comparison-psppi-mi}

Multiple imputation based on Rubin's rule is a well-established methodology that has been widely applied to account for the uncertainty introduced by imputation models to guarantee valid downstream statistical inference~\citep{rubin1976inference, rubin1987multiple, rubin1996multiple, little1992regression}. In Section~\ref{sec:theoretical-comparison-mi-psppi} of the Supplementary Materials, we compare the asymptotic variances of MI and PS-PPI under a correctly specified parametric model and proper imputation, and characterize conditions favoring PS-PPI. The results show that PS-PPI tends to be asymptotically more efficient than MI when the observed components of incomplete subjects provide little information for estimating the parameters of interest and when the ML predictions are of high quality. A mean estimation example based on a bivariate normal distribution is provided to further illustrate the comparison; we refer the reader to Section~\ref{sec:theoretical-comparison-mi-psppi} for a detailed discussion.

\subsection{Connections to existing prediction-based inference work}

\paragraph{Connections to the setting when only the outcome is missing.} We show that the PS-PPI procedure subsumes the outcome-based PB inference framework in the existing literature. Specifically, assume that in addition to the fully observed data subset, there is only one missingness pattern where the outcome, say $Y_{i1}$, is missing. Furthermore, assume that the underlying missing mechanism is MCAR, i.e., the missing data model $\bpi(\cdot)$ is a constant 2-dimensional vector $(\pi_1, \pi_\infty)$ across all the samples, where $\pi_1$ is the probability of missing outcomes, and $\pi_\infty$ is the probability of being fully observed. The PS-PPI estimator under this setting becomes $\widehat{\vtheta}_{\text{PS-PPI}} = \widehat{\vtheta}_{\text{CC}} - \widehat{\Wb}(\widehat{\vgamma}_{\cc, 1} - \widehat{\vgamma}_{\ic, 1})$. Note that here we have the CCA estimate $\widehat{\vtheta}_{\text{CC}}$ instead of $\widehat{\vtheta}_{\text{WCC}}$, because the weights are the same across all samples. In~\citet{gronsbell2024another} and \citet{chen2000unified}, their proposed estimator is $\widehat{\vtheta}_{\text{Chen}} = \widehat{\vtheta}_{\text{CC}} - \widehat{\Wb}_{\text{Chen}}(\widehat{\vgamma}_{\text{lab}} - \widehat{\vgamma}_{\text{all}})$. Here, $\widehat{\vtheta}_{\text{CC}}$ follows the same definition as in the present paper and $\widehat{\vgamma}_{\text{lab}}=\widehat{\vgamma}_{\cc,1}$ with $K=1$ missingness pattern (missing outcomes only); $\widehat{\vgamma}_{\text{all}}$, which will be discussed in detail later, is a variant of $\widehat{\vgamma}_{\ic,1}$; $\widehat{\Wb}_{\text{Chen}}$ is a tuning parameter matrix estimated from the data. We show the asymptotic equivalence between $\widehat{\vtheta}_{\text{Chen}}$ and $\widehat{\vtheta}_{\text{PS-PPI}}$:
\begin{proposition}
    \label{proposition:equivalence_gronsbell}
Assuming that only the outcome, say $Y_{i1}$, is possibly missing, we have that $\widehat{\vtheta}_{\text{PS-PPI}}$ and $\widehat{\vtheta}_{\text{Chen}}$  are asymptotically equivalent in the linear regression (squared error loss) and MCAR setting considered in~\citet{gronsbell2024another}.
\end{proposition}
Proposition~\ref{proposition:equivalence_gronsbell} is proved by showing that their influence functions are the same (see Section~\ref{sec:proof_equivalence_gronsbell} of the Supplementary Materials). Note that in $\widehat{\vtheta}_{\text{Chen}}$, $\widehat{\vgamma}_{\text{all}}$ is estimated using both partially and fully-observed data. This is different from $\widehat{\vgamma}_{\ic, 1}$ in $\widehat{\vtheta}_{\text{PS-PPI}}$, which is estimated using only the missing-outcome subset of data. \citet{gronsbell2024another} argue that this modification guarantees better efficiency than the original PPI procedure~\citep{angelopoulos2023prediction}. However, when an optimal $\Wb_\opt$ is selected, Proposition~\ref{proposition:equivalence_gronsbell} indicates that such efficiency improvement is absorbed into the $\Wb_\opt$, and thereby the two estimators become asymptotically equivalent. We leave a detailed discussion to Section~\ref{sec:asymptotic_equivalence_chen} of the Supplementary Materials.

$\widehat{\vtheta}_{\text{PS-PPI}}$ also shares core insights with other existing outcome-based PB inference literature~\citep{angelopoulos2023PPI++, gan2024prediction, miao2025assumption}, although they may not be asymptotically equivalent. In particular,~\citet{angelopoulos2023PPI++, gan2024prediction, miao2025assumption} introduce tunable weights within the estimating equations, and derive similar estimators with guarantees of being more efficient than CCA. All of these methods are based on the outcome-missing setting. It is worth noting that $\widehat{\vtheta}_{\text{PS-PPI}}$ and other literature mentioned above are formulated differently, and as a result they may not be asymptotically equivalent to each other. Specifically, $\widehat{\vtheta}_{\text{PS-PPI}}$ augments WCCA at an estimator level, while others augment WCCA at an estimating-equation level. We refer readers to~\citet{chen2026prediction} for a detailed discussion of their differences.

\paragraph{Connections to the missing-multiple-covariate setting.} In addition, the PS-PPI procedure also subsumes the missing covariates scenario studied in~\citet{kluger2025prediction}, which considered two missingness patterns: 1) complete cases (all $p$ variables in $\vY_i$); 2) $s$ different variables in $\vY_i$ are missing simultaneously. Furthermore, the missing mechanism was assumed to be MAR with known probabilities for each missingness pattern. In the above missing covariates setting, we can establish the following asymptotic equivalence between $\widehat{\vtheta}_{\text{PTD}}$ proposed in~\citet{kluger2025prediction} and $\widehat{\vtheta}_{\text{PS-PPI}}$, with the proof provided in Section~\ref{sec:proof_equivalence_kluger}:
\begin{proposition}
    \label{proposition:equivalence_kluger}
    $\widehat{\vtheta}_{\text{PTD}}$ and $\widehat{\vtheta}_{\text{PS-PPI}}$ are asymptotically equivalent.
\end{proposition}

\section{Numerical Experiments}
\label{sec:experiments}

We conduct simulation studies to show the operating characteristics of PS-PPI to corroborate the theoretical predictions of estimation consistency, inferential validity, and efficiency gains.

\subsection{Experimental setup}
\label{sec:experimental-setup}

\paragraph{Data generating process} To demonstrate the operating characteristics of the proposed method under various ML prediction qualities, here we focus on an analytic model specified by a multiple linear regression $Y_i = \beta_0 + \beta_1X_{i1} + \beta_2X_{i2} + \epsilon_i$ with $(\beta_0, \beta_1, \beta_2) = (-1, 0.5, 1.2)$,
where the covariates are generated from $X_{i1} = 0.1\exp(Z_{i1}) + \tau_i$, $X_{i2} = \sin(Z_{i2}) + \nu_i$, with $(Z_{i1}, Z_{i2})$ independently drawn from a bivariate normal distribution with mean zero and a homogeneous marginal variance $\sigma_z^2$ and Pearson correlation $\rho$. The error terms are independent, with $\epsilon_i \sim \cN(0, \sigma^2)$, $\tau_i \sim \cN(0, \sigma_\tau^2)$, and $\nu_i$ following an exponential distribution with mean $\lambda$. To connect with the notation $\vY_i$ introduced in Section~\ref{sec:preliminary}, we have $\vY_i = (Y_i, X_{i1}, X_{i2}, Z_{i1}, Z_{i2})$ in this simulation setting. $Z_{i1}$ and $Z_{i2}$ are assumed to be fully observed and not used in the main analysis model. They can be used as inputs to produce ML predictions of possibly missing $(Y_i, X_{i1}, X_{i2})$ in the analysis model of scientific interest. We set $\cC_i = \infty$ to indicate that all $Y_{i}, X_{i1}, X_{i2}$ are observed. We introduce three distinct missingness patterns in addition to the fully observed data as follows: $\cC_i = 1$ indicates that $Y_i$ is missing; $\cC_i = 2$ indicates that only $X_{i2}$ is missing; $\cC_i = 3$ indicates that both $Y_i$ and $X_{i1}$ are missing. We define a missing data model in the form proposed by~\citet{sun2018inverse} as follows:
\begin{equation} \label{eq:ps_model_1}
\textrm{logit}\{\pi(\cC_i|\vY_i)\} =
\begin{cases}
  -1 + 0.1X_{i2} + 0.1Z_{i1} + 0.1X_{i1}X_{i2}, & \cC_i = 1 \\
  -1.8 - 0.2Y_i + 0.1X_{i1} + 0.1Z_{i1} + 0.3X_{i1}Y_i, & \cC_i = 2 \\
  -1.0 + 0.1X_{i2} + 0.2Z_{i1}, & \cC_i = 3 \\
\end{cases},
\end{equation}
\begin{equation} \label{eq:ps_model_2}
    \pi(\cC_i = \infty| \vY_i) = 1 - \sum^3_{k=1} \pi(\cC_i = k|\vY_i).
\end{equation}
The above missing data model satisfies the MAR assumption, i.e., the probability of $\cC_i = k$ depends only on variables observed in that pattern. To generate simulation data, we first sample $(Z_{i1}, Z_{i2})$, then generate $(Y_i, X_{i1}, X_{i2})$ accordingly to obtain the full data. For each subject, we then sample $\cC_i$ from the missing data model and mask the variables based on the missingness pattern assigned to construct the observed dataset.

\paragraph{ML predictions design} To enable nimble control of bias and variance of the ML predictions for revealing the performance of the different methods under distinct qualities of ML predictions, instead of an explicit ML prediction, we design the ML predictions by adding additive noise and bias at a pre-specified level onto the unknown real values. More specifically, $\widehat{X}_{i1} = X_{i1} + b_{i1} + \epsilon_{i1}, \widehat{X}_{i2} = X_{i2} + b_{i2} + \epsilon_{i2}, \widehat{Y}_i = Y_i + b_{i3} + \epsilon_{i3}$. Here $\epsilon_{i1}$, $\epsilon_{i2}$, and $\epsilon_{i3}$ follow the Normal distribution $\cN(0, \sigma^2_{\text{pred}})$, while $b_{i1}$, $b_{i2}$, and $b_{i3}$ follow exponential distributions with mean $\lambda_\text{pred}$. The two types of error terms represent the noise and bias imposed on the ML predictions, respectively. When the noise and bias are set to 0, perfect ML predictions result. In addition, we consider predictions generated by ML models fitted to external data; Section~\ref{sec:supp-pred-ml} of the Supplementary Materials describes the model fitting in detail.
\paragraph{Hyperparameter settings} In our simulation, we fix $\sigma = 0.5$, $\sigma_\tau = 0.3$, $\lambda = 0.02$, $\sigma_z = 0.2$, and $\rho = 0.4$. In our main simulation, we set the noise and bias level $\sigma_{\text{pred}} = \lambda_{\text{pred}} = 0.1$. Additionally, to investigate the impact of the quality of the ML predictions on the performance of our proposed methods, we iterate either $\sigma_{\text{pred}}$ or $\lambda_{\text{pred}}$ over $\{0.05, 0.1, 0.2, 0.5, 1.0, 1.5, 2.0\}$ while fixing the other at 0. We assume the missing-data model is unknown but correctly specified in functional form and covariates, and estimate it by maximizing the likelihood. We set the sample size $N = 5{,}000$ and repeat the simulation $500$ times. The average number of subjects under each missingness pattern in the 500 simulations is $1{,}469.9$ for complete cases, and $1{,}350.7$, $829.5$, $1{,}349.9$ for $\cC = 1,2,3$, respectively.
\paragraph{Included methods} We include the following methods in the simulation. First, we include two complete-case-based approaches, CCA and WCCA, which use only complete cases during inference. Second, we include PPI\texttt{++}, which is a PB inference baseline that uses only outcome-missing data and assumes the missing data mechanism to be MCAR. For PPI\texttt{++}, fully observed cases and the cases with missing outcomes are used for parameter estimation; the cases with other missingness patterns will be dropped. Third, we include MI, which accommodates general missingness patterns under MAR. In the main simulation, we use Bayesian linear regression for imputation, implemented with \texttt{MICE 3.16.0}~\citep{van2011mice} in \texttt{R}. We further evaluate its sensitivity by using predictive mean matching (PMM), classification and regression trees (CART), and random forests as imputation models instead (Figure~\ref{fig:mi_method_comparison}). Finally, we include our proposed PS-PPI.
\paragraph{Evaluation metrics} Suppose that $M$ replicates are generated in data simulation, and for the $j$-th replicate, let $\widehat{\beta}^{(j)}$ be the point estimate for a scalar parameter, and $(L^{(j)}, U^{(j)})$ its 95\% CI. We compute the estimation bias by
$ \texttt{Bias} = \frac{1}{M} \sum^M_{j=1} \left ( \widehat{\beta}^{(j)} - \beta_{\text{true}} \right );$
we compute the coverage rate by
$ \texttt{Coverage Rate} = \frac{1}{M} \sum^M_{j=1} I \left (\beta_{\text{true}} \in (L^{(j)}, U^{(j)}) \right );$
finally, we compute the average width of the 95\% CIs by $ \texttt{Avg. 95\% CI Width} = \frac{1}{M} \sum^M_{j=1} \left ( U^{(j)} - L^{(j)} \right ).$

\subsection{Results}
\label{sec:experimental-results}

\begin{table}[!tbp]
\centering
\caption{Simulation results with prediction noise and bias $\sigma_\text{pred} = \lambda_\text{pred} = 0.1$. Three evaluation metrics are included: 1) coverage rate; 2) widths of 95\% confidence intervals (CIs); 3) estimation bias.}
\label{tab:main-simulation}
\small
\resizebox{1.0\textwidth}{!}{%
\begin{tabular}{lccccccccc}
\toprule
Method & \multicolumn{3}{c}{$\beta_0$} & \multicolumn{3}{c}{$\beta_1$} & \multicolumn{3}{c}{$\beta_2$} \\
\cmidrule(lr){2-4}\cmidrule(lr){5-7}\cmidrule(lr){8-10}
 & Coverage rate & 95\% CI width & Bias & Coverage rate & 95\% CI width & Bias & Coverage rate & 95\% CI width & Bias \\
\midrule
CCA & 0.558 & 0.054 & 0.025 & 0.836 & 0.171 & -0.044 & 0.946 & 0.258 & -0.006 \\
WCCA & 0.984 & 0.055 & 0.000 & 0.980 & 0.174 & 0.000 & 0.950 & 0.261 & -0.005 \\
PPI++ & 0.690 & 0.041 & 0.015 & 0.904 & 0.128 & -0.023 & 0.964 & 0.195 & -0.001 \\
MI & 0.944 & 0.051 & -0.001 & 0.940 & 0.157 & -0.004 & 0.542 & 0.266 & 0.120 \\
PS-PPI & 0.944 & 0.031 & 0.000 & 0.950 & 0.102 & 0.001 & 0.948 & 0.156 & -0.001 \\
\bottomrule
\end{tabular}%
}
\end{table}
We present the simulation results in Table~\ref{tab:main-simulation}, and the results regarding the impact of ML prediction quality on the performance of PS-PPI are presented in Figures~\ref{fig:ml_imputation_quality_simulation_noise_sweep} and~\ref{fig:ml_imputation_quality_simulation_bias_sweep}.
\paragraph{Coverage rate} In Table~\ref{tab:main-simulation}, CCA achieves coverage rates of 55.8\% and 83.6\% for $\beta_0$ and $\beta_1$, below the nominal level of 95\%. This is expected because CCA does not use estimated missing data weights to adjust for the missing-at-random mechanism that renders the complete cases non-representative of the entire population. WCCA achieves conservative coverage rates of 98.4\% and 98.0\% for $\beta_0$ and $\beta_1$, likely because its robust sandwich variance does not account for the uncertainty from the estimated missing data model. MI achieves a coverage rate of 54.2\% for $\beta_2$, falling short of the nominal 95\% level. We hypothesize that this is because the imputation model is not correctly specified. A supplementary simulation (Figure~\ref{fig:mi_method_comparison}) shows that MI fails to achieve the nominal 95\% coverage rate across a variety of imputation models. The PB inference baseline PPI\texttt{++} achieves coverage rates of 69.0\% and 90.4\% for $\beta_0$ and $\beta_1$ for the same reason as CCA. It is worth noting that from Figures~\ref{fig:ml_imputation_quality_simulation_noise_sweep} and~\ref{fig:ml_imputation_quality_simulation_bias_sweep}, when the quality of the ML predictions degrades, there is a decreasing trend of PPI\texttt{++}'s coverage rates for $\beta_0$ and $\beta_1$, which finally converges to levels similar to the CCA approach. In contrast, PS-PPI achieves coverage rates around 95\% for all three coefficients, demonstrating valid inference when the missing data model is correctly specified and estimated.

\paragraph{Widths} Among the methods with valid coverage rates, PS-PPI produces much narrower 95\% CIs than WCCA (0.031, 0.102, and 0.156 versus 0.055, 0.174, and 0.261 for the three coefficients). Although the widths of PS-PPI are strongly impacted by the quality of the ML predictions, they remain narrower than those of WCCA, which is consistent with Corollary~\ref{corollary:optimal_w}.

\paragraph{Bias} For $\beta_0$ and $\beta_1$, WCCA and PS-PPI produce negligible biases smaller than 0.002, while CCA, PPI\texttt{++}, and MI exhibit greater biases when the coverage rate collapses. MI shows a bias of 0.120 for $\beta_2$, and CCA shows biases of 0.025 and -0.044 for $\beta_0$ and $\beta_1$. PPI\texttt{++} is less biased than CCA but more biased than WCCA, and its biases for $\beta_0$ and $\beta_1$ increase as the quality of the imputations degrades (Figures~\ref{fig:ml_imputation_quality_simulation_noise_sweep} and~\ref{fig:ml_imputation_quality_simulation_bias_sweep}).

With predictions from ML models trained on external data, PS-PPI shows performance patterns similar to those of the above simulation results (Figure~\ref{fig:prediction_model_z_exp_sin_noise_grouped_bar_N_5000}). We refer the readers to Section~\ref{sec:supp-pred-ml} of the Supplementary Materials for more details.

Section~\ref{sec:misspecification_robustness} of the Supplementary Materials discusses the robustness of the proposed method against missing data model misspecification. Under a misspecified missing data model, the coverage rates of PS-PPI collapse when lower-quality predictions are used to impute the missing entries. Because CCA, WCCA, and MI do not use ML predictions, their coverage rates, CI widths, and biases do not vary with the quality of the ML predictions.

In addition to the simulation study above, the Supplementary Materials provide further supplementary numerical simulation. In Section~\ref{sec:finite-sample-psppi}, we vary the sample size to evaluate the finite-sample behavior of PS-PPI. In Section~\ref{sec:supp-ps-ml}, we evaluate PS-PPI with missing data weights estimated from externally trained ML models.

\section{Application to the \textit{All of Us} Data}
\label{sec:data_example}
\paragraph{Background} Education attainment is a key social determinant of cardiovascular health, shaping lifestyle behaviors, access to healthcare, and long-term disease outcomes \citep{wang2024association}. Among the biological pathways linking social conditions to cardiovascular risk, systemic inflammation (often measured through C-reactive protein (CRP)) plays a central role. CRP has been widely associated with the development of cardiovascular disease (CVD)~\citep{ridker2003clinical}. In this section, we use data from the \textit{All of Us} (AoU) program~\citep{all2019all} to examine how education levels relate to CRP concentrations, accounting for demographic and clinical covariates, to explore whether disparities in inflammation help explain the socioeconomic gradient in CVD risk \citep{magnani2024educational}. AoU was launched in 2018 to establish a large-scale database of EHRs and biospecimens for hypothesis-free population research. In this study, we analyze data from 314,625 participants who self-reported as male or female and have both EHR and genotype data available in the Curated Data Repository (CDR) version 8.
\paragraph{Analysis model with ML predictions} To investigate the association, we fit the following regression model with the primary covariate (education), adjusting for age, sex, race, smoking status, and the first 5 genetic principal components (PCs):
$$\text{CRP}  \sim \text{Intercept} + \text{Education} + \text{Age} + \text{Sex} + \text{Race}  + \text{Smoking}  +\text{PC1} + \cdots + \text{PC5}.$$
To address genetic relatedness, we keep one individual from each kinship group for analysis (294,347 participants) and use the remaining 20,278 to train ML models~\citep{mccaw2024synthetic}. In the analytic sample, 47,317 subjects are fully observed, and missingness occurs in race, smoking status, education, and CRP, yielding 15 patterns with at least one missing entry. Besides the complete cases, we keep only the patterns with more than 10,000 samples: 1) missing CRP only (172,685); 2) missing race and CRP (50,040); 3) missing race only (10,934). We drop the remaining 13,371 samples (4.54\% of the analytic sample) and reduce the number of patterns from 16 to 4. Summary statistics are given in Section~\ref{sec:descriptive_aou} of the Supplementary Materials. Each missing variable is imputed by gradient boosting decision trees \citep[XGBoost,][]{chen2016xgboost} or random forest, whichever performs better on a validation set. The predictors include age, sex, genetically inferred ancestry, PCs, and diagnostic events coded by Systematized Nomenclature of Medicine – Clinical Terms~\citep{stearns2001snomed}. We fit missing data models of the form proposed by~\cite{sun2018inverse} (similar to Equations~(\ref{eq:ps_model_1}) and~(\ref{eq:ps_model_2})) with all the available variables in each pattern. We compare PS-PPI against WCCA and against PS-PPI using only the outcome-missing pattern. We conduct sensitivity analysis of PS-PPI without inverse probability weighting (IPW) and compare against CCA and its outcome-based variant. For each method, we report coefficient estimates, standard error, 95\% CIs, and variance ratio (relative to WCCA or CCA).

\begin{table}[!tbp]
\centering
\caption{\textit{All of Us} based coefficient estimates and 95\% CIs for education level. Variance ratios are relative to CCA without IPW and WCCA with IPW. Smaller variance ratios represent larger efficiency gains.}
\label{tab:data_example_results}
\resizebox{1.0\linewidth}{!}{%
\begin{tabular}{c|lcccc}
\toprule
 & Method & $\widehat{\beta}_{\textrm{edu}}$ & Standard Error & 95\% CIs & Variance Ratio \\
\midrule
\multirow{3}{*}{without IPW}
 & CCA (baseline)     & -1.187 & 0.048 & $[-1.281, -1.093]$ & 1 \\
 & PS-PPI (outcome missing only)     & -0.953 & 0.045 & $[-1.041, -0.864]$ & 0.889 \\
 & PS-PPI             & -1.040 & 0.042 & $[-1.122, -0.958]$ & \textbf{0.760} \\
\midrule
\multirow{3}{*}{with IPW}
 & WCCA (baseline)    & -1.162 & 0.051 & $[-1.262, -1.063]$ & 1 \\
 & PS-PPI (outcome missing only) & -0.927 & 0.048 & $[-1.021, -0.833]$ & 0.898 \\
 & PS-PPI             & -1.004 & 0.045 & $[-1.091, -0.916]$ & \textbf{0.774} \\
\bottomrule
\end{tabular}}
\end{table}

\paragraph{Results} From Table~\ref{tab:data_example_results}, all methods show that a higher education level is significantly associated with a lower CRP level, adjusting for other covariates. Specifically, the WCCA estimate for ${\beta}_{\textrm{edu}}$  is $-1.162$ (95\% CI: [-1.262, -1.063]). After integrating data from the three retained missingness patterns with ML imputations, PS-PPI achieves a 22.6\% variance reduction relative to WCCA, and an estimate of $-1.004$ (95\% CI: [-1.091, -0.916]). We observe a smaller magnitude of variance reduction (10.2\%) when integrating only the CRP-missing pattern. As a sensitivity analysis, without IPW, the CCA estimate is $-1.187$ (95\% CI: [-1.281, -1.093]), similar to the WCCA estimate, indicating a small degree of selection bias of having CRP test results. The percent variance reductions of PS-PPI with only complete plus missing CRP data and PS-PPI with data from the three retained missingness patterns are 11.1\% and 24.0\%, respectively, favoring PS-PPI, which integrates more data and ML predictions.

\section{Conclusion and Discussion}
\label{sec:discussion}

This paper introduces a unified framework for conducting valid statistical inference using machine learning imputations in datasets characterized by general, multiple different missing data patterns under the MAR assumption. By leveraging the concept of missingness patterns and building upon the theoretical foundations of AIPW estimators, our work significantly extends the applicability of PB inference beyond the more restrictive MCAR assumption and simple missing data structures, such as missing outcomes or missing multiple covariates. A key practical advantage of our proposed methods is their adaptability to any Z-estimation framework, allowing researchers to utilize existing statistical software for weighted complete-case analysis with minimal additional implementation effort. The method introduced here can integrate individuals across diverse missingness patterns, hence maximizing data utility and inclusiveness while ensuring accuracy. When multiple prediction models are available for each missingness pattern, we have also provided an extension (see Section \ref{sec:multiple_predictive_functions} in the Supplementary Materials).

Despite these important contributions, our work has several limitations that inform avenues for future research. First, a central assumption of the correct specification of the missing data model, which governs the probability of observing different missingness patterns. This assumption is the price we pay in the challenging scenario where PB inference is most needed: analysts typically have minimal or no control over the prediction models. As such, assumptions on the prediction models need to be avoided, and the correct specification of the missing data model comes at a price. Figures~\ref{fig:varying_prediction_noise_misspecified_ps} and~\ref{fig:varying_prediction_bias_misspecified_ps} in the Supplementary Materials show that missing data model misspecification will lead to biased estimation. Although we have proposed a doubly robust variant of PS-PPI in Section~\ref{sec:psppi-doubly-robust} in the Supplementary Materials, it requires a different missing data assumption and, under a misspecified missing data model, both a correctly specified working model $\vq$ that satisfies (b.1) and predictive functions $h_k$ that satisfy condition (b.2) of Theorem~\ref{thrm:psppi-dr}. Future work would focus on mitigating the effect of missing data model misspecification. Second, a practical limitation of PS-PPI is that it imposes assumptions on the minimum size of data within each missingness pattern to guarantee the learnability of the parameter of interest. For example, to fit a linear regression analysis model with $d$ parameters (including the intercept), at least $d + 1$ data points are needed. In this case, PS-PPI requires each missingness pattern within the data to contain at least $d + 1$ data points, and if this requirement is not satisfied, then data with that missingness pattern have to be dropped from the dataset. Third, the number of possible missingness patterns can grow exponentially with the number of variables. Although in many practical applications, the number of observed patterns is manageable, scenarios with a very large number of distinct patterns could pose computational challenges. Future work can improve the utility of the framework by developing adaptive procedures to use data from a small number of the most informative missingness patterns. Fourth, while our framework is versatile within Z-estimation problems, its extension to other inferential paradigms or specific complex data types, such as longitudinal or spatial data where the temporal and spatial dependency of missingness is critical, or high-dimensional settings where variable selection and inference are performed simultaneously, presents fertile ground for future research. Finally, the current work also focuses on the MAR assumption. Extending these prediction-powered inference ideas to handle missing not at random (MNAR) scenarios, while notoriously challenging due to identifiability issues, would be a significant theoretical and practical advancement. We leave these directions for future work.

\section*{Acknowledgment}
\label{sec:acknowledgement}
This study used data from the \textit{All of Us} Research Program’s Controlled Tier Dataset v8, available to authorized users on the \href{https://workbench.researchallofus.org/login}{Researcher Workbench}. We gratefully acknowledge \textit{All of Us} participants for their contributions, without whom this research would not have been possible. We also thank the National Institutes of Health’s \href{https://www.nih.gov/allofus}{\textit{All of Us} Research Program} for making available the participant data examined in this study. \textit{Conflict of Interest}: None.

\section*{Declaration of Generative AI Use}
Generative AI (GPT 6 Astra, Claude Opus 5.5 and Sonnet 5.5) were used to assist with the implementation of the simulation studies and the \textit{All of Us} analysis, and with language editing of the manuscript.

\singlespacing
\bibliographystyle{apalike}
\bibliography{bibliography.bib}

\clearpage
\pdfbookmark[0]{Supplementary Materials}{supplementary-materials}
\begin{center}
{\LARGE\bf Supplementary Materials\par}
\end{center}
\medskip
\noindent The Supplementary Materials contain a glossary of key notation used in the main paper, proofs of the main theoretical results, additional preliminaries and theoretical analyses, implementation details, additional simulations, a descriptive analysis of the \textit{All of Us} data, and additional discussions. Sections, equations, figures, tables, algorithms, and theorem-like environments in the Supplementary Materials use the prefix ``S''; cross-references without this prefix refer to the main paper.

\newpage

\addtocontents{toc}{\protect\setcounter{tocdepth}{3}}
\tableofcontents
\spacingset{1.2}

\newpage

\setcounter{section}{0}
\setcounter{equation}{0}
\setcounter{figure}{0}
\setcounter{table}{0}
\setcounter{algorithm}{0}
\setcounter{theorem}{0}
\setcounter{proposition}{0}
\setcounter{lemma}{0}
\setcounter{assumption}{0}
\setcounter{definition}{0}
\setcounter{example}{0}
\renewcommand{\thesection}{S\arabic{section}}
\numberwithin{equation}{section}
\renewcommand{\thefigure}{S\arabic{figure}}
\renewcommand{\thetable}{S\arabic{table}}
\renewcommand{\thealgorithm}{S\arabic{algorithm}}
\renewcommand{\thetheorem}{S\arabic{theorem}}
\renewcommand{\theproposition}{S\arabic{proposition}}
\renewcommand{\thelemma}{S\arabic{lemma}}
\renewcommand{\theassumption}{S\arabic{assumption}}
\renewcommand{\thedefinition}{S\arabic{definition}}
\renewcommand{\theexample}{S\arabic{example}}
\renewcommand{\theHsection}{S\arabic{section}}
\renewcommand{\theHequation}{\theHsection.\arabic{equation}}
\def\theHfigure{S\arabic{figure}}
\def\theHtable{S\arabic{table}}
\def\theHalgorithm{S\arabic{algorithm}}
\def\theHtheorem{S\arabic{theorem}}
\def\theHproposition{S\arabic{proposition}}
\def\theHlemma{S\arabic{lemma}}
\def\theHassumption{S\arabic{assumption}}
\def\theHdefinition{S\arabic{definition}}
\def\theHexample{S\arabic{example}}

\newpage

\section{Glossary: Key Notation}
\label{sec:notation_glossary}

Table~\ref{tab:notation_glossary} summarizes the key notation used in the main paper.

\begingroup
\small
\begin{longtable}{>{\raggedright\arraybackslash}p{0.25\textwidth} >{\raggedright\arraybackslash}p{0.69\textwidth}}
\caption{Glossary of key notation used in the main paper.}
\label{tab:notation_glossary} \\
\toprule
\textbf{Symbol} & \textbf{Description} \\
\midrule
\endfirsthead
\multicolumn{2}{l}{Table~\thetable\ (continued)} \\
\toprule
\textbf{Symbol} & \textbf{Description} \\
\midrule
\endhead
\midrule
\multicolumn{2}{r}{Continued on next page} \\
\endfoot
\bottomrule
\endlastfoot
\multicolumn{2}{l}{\textit{Data and missingness}} \\
$\vY_i$ & Full data of subject $i$. \\
$\bR_i$ & Missingness indicators; $R_{ij} = 1$ if $Y_{ij}$ is observed and $0$ otherwise. \\
$\cC_i$ & Missingness index; $\cC_i = k \in \{1, \ldots, K\}$ denotes the $k$-th missingness pattern, and $\cC_i = \infty$ denotes a fully observed subject. \\
$G_k(\cdot)$ & Coarsening function that returns the components observed under the $k$-th missingness pattern; $G_\infty$ is the identity mapping. \\
$\cD_{\obs}$ & Observed dataset $\{(G_{\cC_i}(\vY_i), \cC_i)\}_{i=1}^N$. \\
$\bpi(\cdot; \vomega)$ & Missing data model with parameter $\vomega$. \\
$\vphi_i(\vomega)$ & Estimating function of the missing data model for subject $i$. \\
$h_k$ & Predictive function trained on external data that predicts the components missing under the $k$-th pattern from $G_k(\vY_i)$. \\
$\vYimp_{k, i}$ & Imputed vector $\{G_k(\vY_i), h_k(G_k(\vY_i))\}$ of subject $i$ under the $k$-th pattern. \\
\addlinespace[0.5em]
\multicolumn{2}{l}{\textit{Estimating equations and estimators}} \\
$\vpsi(\vY_i; \vtheta)$ & Full-data estimating function for subject $i$. \\
$\vpsi_{\wcc, i}(\vtheta, \vomega)$ & Weighted complete-case (WCC) estimating function for subject $i$. \\
$\vpsi^{(k)}_{\cc, i}$, $\vpsi^{(k)}_{\ic, i}$ & Weighted estimating functions of the imputed vector $\vYimp_{k, i}$ for subject $i$ with $\cC_i = \infty$ and with $\cC_i = k$, respectively. \\
$\vtheta^\ast$ & Parameters of interest; $\E [\vpsi(\vY_i; \vtheta^\ast)] = \vzero$. \\
$\widehat{\vtheta}_{\full}$ & An estimate of $\vtheta^\ast$ computed from the full data. \\
$\widehat{\vtheta}_{\cc}$ & An unweighted complete-case (CC) estimate of $\vtheta^\ast$. \\
$\widehat{\vtheta}_{\wcc}$ & A weighted complete-case (WCC) estimate of $\vtheta^\ast$. \\
$\widehat{\vtheta}_{\text{PP}}$ & An estimator proposed by~\citet{angelopoulos2023prediction}. \\
$\widehat{\vtheta}_{\text{Chen}}$ & An estimator proposed by~\citet{chen2000unified} and~\citet{gronsbell2024another}. \\
$\widehat{\vtheta}_{\text{PTD}}$ & An estimator proposed by~\citet{kluger2025prediction}. \\
$\widehat{\vgamma}_{\cc, k}$ & The solution of Equation (\ref{eq:general_missing_pattern_ee_gamma_1}), which uses complete-case data, with all entries supposed to be missing in the $k$-th missingness pattern being substituted with predictions generated from $h_k$. \\
$\widehat{\vgamma}_{\ic, k}$  &  The solution of Equation (\ref{eq:general_missing_pattern_ee_gamma_2}), which uses data with the $k$-th missingness pattern, with all missing values imputed by predictions generated from $h_k$.\\
$\widehat{\vDelta}_k$, $\widehat{\vDelta}$ & Estimate of zero $\widehat{\vgamma}_{\cc, k} - \widehat{\vgamma}_{\ic, k}$ from the $k$-th pattern, and $\widehat{\vDelta} = (\widehat{\vDelta}_1^\top, \ldots, \widehat{\vDelta}_K^\top)^\top$. \\
$\Wb$, $\Wb_\opt$ & Weight matrix $\Wb = (\Wb_1, \ldots, \Wb_K) \in \R^{d \times Kd}$; $\Wb_\opt $ is the globally optimal choice. \\
$\widehat{\vtheta}_{\psppi}$ & PS-PPI estimator proposed in this work. \\
\addlinespace[0.5em]
\multicolumn{2}{l}{\textit{Influence functions and variance-covariance matrices}} \\
$\widetilde{\varphi}_{\infty, i}$, $\varphi_{\infty, i}$ & The influence function of $\widehat{\vtheta}_{\wcc}$; the latter further accounts for missing data model uncertainty. \\
$\varphi^{(k)}_{\cc, i}$, $\varphi^{(k)}_{\ic, i}$ & The influence functions of $\widehat{\vgamma}_{\cc, k}$ and $\widehat{\vgamma}_{\ic, k}$ that account for missing data model uncertainty. \\
$\varphi_{k, i}$, $\varphi_{\vDelta, i}$ & The influence functions $\varphi^{(k)}_{\cc, i} - \varphi^{(k)}_{\ic, i}$ of $\widehat{\vDelta}_k$; $\varphi_{\vDelta, i}$ stacks $\varphi_{1, i}, \ldots, \varphi_{K, i}$. \\
$\varphi_i$ & The influence function of $\widehat{\vtheta}_{\psppi}$. \\
$\Sigma_{\vtheta}$ & Asymptotic variance-covariance matrix of $\widehat{\vtheta}_{\wcc}$ in Proposition~\ref{proposition:wcc}. \\
$\Sigma_{\vDelta_k}$ & Asymptotic variance-covariance matrix of $\widehat{\vDelta}_k$. \\
$\Sigma_{\vtheta, \vDelta}$ & Asymptotic covariance matrix between $\widehat{\vtheta}_{\wcc}$ and $\widehat{\vDelta}$. \\
$\Sigma_{\vDelta}$ & Asymptotic variance-covariance matrix of $\widehat{\vDelta}$. \\
$\Sigma_{\psppi}$ & Asymptotic variance-covariance matrix of $\widehat{\vtheta}_{\psppi}$. \\
\end{longtable}
\endgroup

\newpage

\section{Proofs}
\subsection{Proof of Proposition~\ref{proposition:wcc}}
\label{sec:proof_proposition_wcc}

\begin{proof}
Under the MAR assumption and correct specification of the missing data model,
\begin{align*}
\E\left[
\vpsi_{\mathrm{WCC},i}(\vtheta^\ast,\vomega^\ast)
\right]
&=
\E\left[
\frac{I(\cC_i=\infty)}
{\pi_\infty(\boldsymbol{Y}_i;\vomega^\ast)}
\vpsi(\boldsymbol{Y}_i;\vtheta^\ast)
\right] \\
&=
\E\left[
\frac{
\E\{I(\cC_i=\infty)\mid \boldsymbol{Y}_i\}
}
{\pi_\infty(\boldsymbol{Y}_i;\vomega^\ast)}
\vpsi(\boldsymbol{Y}_i;\vtheta^\ast)
\right] \\
&=
\E\left[
\vpsi(\boldsymbol{Y}_i;\vtheta^\ast)
\right]
=
\vzero.
\end{align*}
Thus, Equation (\ref{eq:general_missing_pattern_ee}) targets the same parameter
$\vtheta^\ast$ as the full-data estimating equation. Under Assumption~\ref{assumption:main},
the usual consistency theorem for 2-step $Z$-estimators detailed in~\citet{newey1994large} proves
\[
\widehat{\vtheta}_{\mathrm{WCC}}
\overset{p}{\longrightarrow}
\vtheta^\ast.
\]
Regarding the asymptotic normality, by a first-order Taylor expansion of Equation (\ref{eq:general_missing_pattern_ee}) at $(\vtheta^\ast,\vomega^\ast)$, we have
\begin{align*}
\vzero
&=
\frac{1}{\sqrt{N}}
\sum_{i=1}^N
\vpsi_{\mathrm{WCC},i}(\vtheta^\ast,\vomega^\ast)
+
\dot{\vpsi}_\infty(\vtheta^\ast,\vomega^\ast)
\sqrt{N}(\widehat{\vtheta}_{\mathrm{WCC}}-\vtheta^\ast)
\\
&\quad
+
\dot{\vpsi}^{\vomega}_\infty(\vtheta^\ast,\vomega^\ast)
\sqrt{N}(\widehat{\vomega}-\vomega^\ast)
+
o_p(1),
\end{align*}
where
$$
\dot{\vpsi}_\infty(\vtheta^\ast,\vomega^\ast)
=
\E\left[
\nabla_\vtheta
\vpsi_{\mathrm{WCC},i}(\vtheta^\ast,\vomega^\ast)
\right] = \E\left[
\frac{\partial \vpsi_{\mathrm{WCC},i}(\vtheta,\vomega)}{\partial \vtheta}
\Bigr|_{(\vtheta, \vomega) = (\vtheta^\ast, \vomega^\ast)} \right],$$
$$\dot{\vpsi}^{\vomega}_\infty(\vtheta^\ast,\vomega^\ast) =
\E\left[
\nabla_\vomega
\vpsi_{\mathrm{WCC},i}(\vtheta^\ast,\vomega^\ast)
\right] = \E\left[
\frac{\partial \vpsi_{\mathrm{WCC},i}(\vtheta,\vomega)}{\partial \vomega}
\Bigr|_{(\vtheta, \vomega) = (\vtheta^\ast, \vomega^\ast)} \right].
$$
Similarly, applying a first-order Taylor expansion to the estimating equation for the missing data model,
\[
\vzero
=
\frac{1}{\sqrt{N}}
\sum_{i=1}^N
\vphi_i(\vomega^\ast)
+
\dot{\vphi}(\vomega^\ast)
\sqrt{N}(\widehat{\vomega}-\vomega^\ast)
+
o_p(1),
\]
where
$
\dot{\vphi}(\vomega^\ast)
=
\E\left[
\nabla_\vomega
\vphi_i(\vomega^\ast)
\right]
$.
Therefore,
\[
\sqrt{N}(\widehat{\vomega}-\vomega^\ast)
=
-
\dot{\vphi}^{-1}(\vomega^\ast)
\frac{1}{\sqrt{N}}
\sum_{i=1}^N
\vphi_i(\vomega^\ast)
+
o_p(1).
\]
As such, we have
\begin{align*}
\sqrt{N}
(\widehat{\vtheta}_{\mathrm{WCC}}-\vtheta^\ast)
&=
\frac{1}{\sqrt{N}}
\sum_{i=1}^N
\Big[
-\dot{\vpsi}^{-1}_\infty(\vtheta^\ast,\vomega^\ast)
\vpsi_{\mathrm{WCC},i}(\vtheta^\ast,\vomega^\ast)
\\
&\quad
+
\dot{\vpsi}^{-1}_\infty(\vtheta^\ast,\vomega^\ast)
\dot{\vpsi}^{\vomega}_\infty(\vtheta^\ast,\vomega^\ast)
\dot{\vphi}^{-1}(\vomega^\ast)
\vphi_i(\vomega^\ast)
\Big]
+
o_p(1).
\end{align*}
Defining
$
\widetilde{\varphi}_{\infty,i}(\vtheta^\ast,\vomega^\ast)
:=
-
\dot{\vpsi}^{-1}_\infty(\vtheta^\ast,\vomega^\ast)
\vpsi_{\mathrm{WCC},i}(\vtheta^\ast,\vomega^\ast)
$, we know that
\begin{align*}
\E\left[
\nabla_\vomega
\widetilde{\varphi}_{\infty,i}(\vtheta^\ast,\vomega^\ast)
\right]
&=
\E\left[
\nabla_\vomega
\left\{
-
\dot{\vpsi}^{-1}_\infty(\vtheta^\ast,\vomega^\ast)
\vpsi_{\mathrm{WCC},i}(\vtheta^\ast,\vomega^\ast)
\right\}
\right] \\
&=
-
\dot{\vpsi}^{-1}_\infty(\vtheta^\ast,\vomega^\ast)
\E\left[
\nabla_\vomega
\vpsi_{\mathrm{WCC},i}(\vtheta^\ast,\vomega^\ast)
\right] \\
&=
-
\dot{\vpsi}^{-1}_\infty(\vtheta^\ast,\vomega^\ast)
\dot{\vpsi}^{\vomega}_\infty(\vtheta^\ast,\vomega^\ast).
\end{align*}
As a result, the expansion can be written as
\[
\sqrt{N}
(\widehat{\vtheta}_{\mathrm{WCC}}-\vtheta^\ast)
=
\frac{1}{\sqrt{N}}
\sum_{i=1}^N
\varphi_{\infty,i}(\vtheta^\ast,\vomega^\ast)
+
o_p(1),
\]
where
\[
\varphi_{\infty,i}(\vtheta^\ast,\vomega^\ast)
=
\widetilde{\varphi}_{\infty,i}(\vtheta^\ast,\vomega^\ast)
-
\E \left [
\nabla_\vomega
\widetilde{\varphi}_{\infty,i}(\vtheta^\ast,\vomega^\ast)
\right]
\dot{\vphi}^{-1}(\vomega^\ast)
\vphi_i(\vomega^\ast).
\]
Correspondingly,
\[
\sqrt{N}
(\widehat{\vtheta}_{\mathrm{WCC}}-\vtheta^\ast)
\overset{d}{\longrightarrow}
N(\vzero,\Sigma_\vtheta),
\]
where
\[
\Sigma_\vtheta
=
\E\left[
\varphi_{\infty,i}(\vtheta^\ast,\vomega^\ast)
\varphi_{\infty,i}^{\top}(\vtheta^\ast,\vomega^\ast)
\right].
\]
This proves the consistency and asymptotic normality of
$\widehat{\vtheta}_{\mathrm{WCC}}$.
\end{proof}

\subsection{Proof of Proposition~\ref{proposition:estimate_of_zero}}
\label{sec:proof_estimate_of_zero}

\begin{proof}
Fix $k\in\{1,\ldots,K\}$. We first show that the two population-level
estimating equations have the same root $\vgamma_k^\ast$. Under the MAR
assumption and correct specification of the missing data model, from
Equation~(\ref{eq:general_missing_pattern_ee_gamma_1}) and
Assumption~\ref{assumption:gamma_existence}, we have
\begin{align*}
&\quad
\E_{(\boldsymbol{Y},\cC)}
\left[
\frac{I(\cC_i=\infty)}
{\pi_\infty(\boldsymbol{Y}_i;\vomega^\ast)}
\vpsi\left(\vYimp_{k,i};\vgamma_k^\ast\right)
\right] \\
&=
\E_{\boldsymbol{Y}}
\left[
\vpsi\left(\vYimp_{k,i};\vgamma_k^\ast\right)
\E_{(\cC\mid \boldsymbol{Y})}
\left\{
\frac{I(\cC_i=\infty)}
{\pi_\infty(\boldsymbol{Y}_i;\vomega^\ast)}
\right\}
\right] \\
&=
\E_{\boldsymbol{Y}}
\left[
\vpsi\left(\vYimp_{k,i};\vgamma_k^\ast\right)
\right]
=
\vzero.
\end{align*}
Similarly, from Equation~(\ref{eq:general_missing_pattern_ee_gamma_2}), we have
\begin{align*}
&\quad
\E_{(\boldsymbol{Y},\cC)}
\left[
\frac{I(\cC_i=k)}
{\pi_k(G_k(\boldsymbol{Y}_i);\vomega^\ast)}
\vpsi\left(\vYimp_{k,i};\vgamma_k^\ast\right)
\right] \\
&=
\E_{\boldsymbol{Y}}
\left[
\vpsi\left(\vYimp_{k,i};\vgamma_k^\ast\right)
\E_{(\cC\mid \boldsymbol{Y})}
\left\{
\frac{I(\cC_i=k)}
{\pi_k(G_k(\boldsymbol{Y}_i);\vomega^\ast)}
\right\}
\right] \\
&=
\E_{\boldsymbol{Y}}
\left[
\vpsi\left(\vYimp_{k,i};\vgamma_k^\ast\right)
\right]
=
\vzero.
\end{align*}
Therefore, both weighted estimating equations target the same population root
$\vgamma_k^\ast$. Under Assumption~\ref{assumption:main}, the usual consistency
theorem for two-step $Z$-estimators detailed in~\citet{newey1994large} gives
\[
\widehat{\vgamma}_{\cc,k}
\overset{p}{\longrightarrow}
\vgamma_k^\ast,
\qquad
\widehat{\vgamma}_{\ic,k}
\overset{p}{\longrightarrow}
\vgamma_k^\ast.
\]
Hence,
\[
\widehat{\vDelta}_k
=
\widehat{\vgamma}_{\cc,k}
-
\widehat{\vgamma}_{\ic,k}
\overset{p}{\longrightarrow}
\vzero_d.
\]

We next derive the asymptotic linear expansion of
$\widehat{\vgamma}_{\cc,k}$. By a first-order Taylor expansion of
Equation~(\ref{eq:general_missing_pattern_ee_gamma_1}) at
$(\vgamma_k^\ast,\vomega^\ast)$,
\begin{align*}
\vzero
&=
\frac{1}{\sqrt{N}}
\sum_{i=1}^N
\vpsi_{\cc,i}^{(k)}(\vgamma_k^\ast,\vomega^\ast)
+
\dot{\vpsi}_k(\vgamma_k^\ast,\vomega^\ast)
\sqrt{N}
(\widehat{\vgamma}_{\cc,k}-\vgamma_k^\ast) \\
&\quad
+
\dot{\vpsi}_{\cc,k}^{\vomega}(\vgamma_k^\ast,\vomega^\ast)
\sqrt{N}
(\widehat{\vomega}-\vomega^\ast)
+
o_p(1),
\end{align*}
where
\[
\dot{\vpsi}_k(\vgamma_k^\ast,\vomega^\ast)
=
\E\left[
\nabla_{\vgamma_k}
\vpsi_{\cc,i}^{(k)}(\vgamma_k^\ast,\vomega^\ast)
\right],
\]
and
\[
\dot{\vpsi}_{\cc,k}^{\vomega}(\vgamma_k^\ast,\vomega^\ast)
=
\E\left[
\nabla_{\vomega}
\vpsi_{\cc,i}^{(k)}(\vgamma_k^\ast,\vomega^\ast)
\right].
\]
By plugging in the asymptotic expansion of $\sqrt{N}(\widehat{\vomega}-\vomega^\ast)$ from Proposition~\ref{proposition:wcc}, we have
\begin{align*}
\sqrt{N}
(\widehat{\vgamma}_{\cc,k}-\vgamma_k^\ast)
&=
\frac{1}{\sqrt{N}}
\sum_{i=1}^N
\Big[
-\dot{\vpsi}_k^{-1}(\vgamma_k^\ast,\vomega^\ast)
\vpsi_{\cc,i}^{(k)}(\vgamma_k^\ast,\vomega^\ast) \\
&\quad
+
\dot{\vpsi}_k^{-1}(\vgamma_k^\ast,\vomega^\ast)
\dot{\vpsi}_{\cc,k}^{\vomega}(\vgamma_k^\ast,\vomega^\ast)
\dot{\vphi}^{-1}(\vomega^\ast)
\vphi_i(\vomega^\ast)
\Big]
+
o_p(1).
\end{align*}
Define
\[
\widetilde{\varphi}_{\cc,i}^{(k)}(\vgamma_k^\ast,\vomega^\ast)
:=
-
\dot{\vpsi}_k^{-1}(\vgamma_k^\ast,\vomega^\ast)
\vpsi_{\cc,i}^{(k)}(\vgamma_k^\ast,\vomega^\ast).
\]
Because
\[
\E\left[
\vpsi_{\cc,i}^{(k)}(\vgamma_k^\ast,\vomega^\ast)
\right]
=
\vzero,
\]
we have
\begin{align*}
\E\left[
\nabla_{\vomega}
\widetilde{\varphi}_{\cc,i}^{(k)}(\vgamma_k^\ast,\vomega^\ast)
\right]
&=
\E\left[
\nabla_{\vomega}
\left\{
-\dot{\vpsi}_k^{-1}(\vgamma_k^\ast,\vomega^\ast)
\vpsi_{\cc,i}^{(k)}(\vgamma_k^\ast,\vomega^\ast)
\right\}
\right] \\
&=
-
\dot{\vpsi}_k^{-1}(\vgamma_k^\ast,\vomega^\ast)
\E\left[
\nabla_{\vomega}
\vpsi_{\cc,i}^{(k)}(\vgamma_k^\ast,\vomega^\ast)
\right] \\
&=
-
\dot{\vpsi}_k^{-1}(\vgamma_k^\ast,\vomega^\ast)
\dot{\vpsi}_{\cc,k}^{\vomega}(\vgamma_k^\ast,\vomega^\ast).
\end{align*}
Thus,
\[
\sqrt{N}
(\widehat{\vgamma}_{\cc,k}-\vgamma_k^\ast)
=
\frac{1}{\sqrt{N}}
\sum_{i=1}^N
\varphi_{\cc,i}^{(k)}(\vgamma_k^\ast,\vomega^\ast)
+
o_p(1),
\]
where
\[
\varphi_{\cc,i}^{(k)}(\vgamma_k^\ast,\vomega^\ast)
=
\widetilde{\varphi}_{\cc,i}^{(k)}(\vgamma_k^\ast,\vomega^\ast)
-
\E\left[
\nabla_{\vomega}
\widetilde{\varphi}_{\cc,i}^{(k)}(\vgamma_k^\ast,\vomega^\ast)
\right]
\dot{\vphi}^{-1}(\vomega^\ast)
\vphi_i(\vomega^\ast).
\]
Similarly we can have the following results for $\widehat{\vgamma}_{\ic, k}$:
\[
\sqrt{N}
(\widehat{\vgamma}_{\ic,k}-\vgamma_k^\ast)
=
\frac{1}{\sqrt{N}}
\sum_{i=1}^N
\varphi_{\ic,i}^{(k)}(\vgamma_k^\ast,\vomega^\ast)
+
o_p(1),
\]
where
\[
\varphi_{\ic,i}^{(k)}(\vgamma_k^\ast,\vomega^\ast)
=
\widetilde{\varphi}_{\ic,i}^{(k)}(\vgamma_k^\ast,\vomega^\ast)
-
\E\left[
\nabla_{\vomega}
\widetilde{\varphi}_{\ic,i}^{(k)}(\vgamma_k^\ast,\vomega^\ast)
\right]
\dot{\vphi}^{-1}(\vomega^\ast)
\vphi_i(\vomega^\ast),
\]
and
\[
\widetilde{\varphi}_{\ic,i}^{(k)}(\vgamma_k^\ast,\vomega^\ast)
:=
-
\dot{\vpsi}_k^{-1}(\vgamma_k^\ast,\vomega^\ast)
\vpsi_{\ic,i}^{(k)}(\vgamma_k^\ast,\vomega^\ast).
\]
Note that the two bread matrices are equivalent:
\begin{align*}
\dot{\vpsi}_k(\vgamma_k^\ast,\vomega^\ast) =
\E\left[
\nabla_{\vgamma_k}
\vpsi_{\cc,i}^{(k)}(\vgamma_k^\ast,\vomega^\ast)
\right]
=
\E_{\boldsymbol{Y}}
\left[
\nabla_{\vgamma_k}
\vpsi\left(\vYimp_{k,i};\vgamma_k^\ast\right)
\right] =
\E\left[
\nabla_{\vgamma_k}
\vpsi_{\ic,i}^{(k)}(\vgamma_k^\ast,\vomega^\ast)
\right].
\end{align*}
Combining the two asymptotic linear expansions, we obtain
\begin{align*}
\sqrt{N}\widehat{\vDelta}_k
&=
\sqrt{N}
\left[
(\widehat{\vgamma}_{\cc,k}-\vgamma_k^\ast)
-
(\widehat{\vgamma}_{\ic,k}-\vgamma_k^\ast)
\right] \\
&=
\frac{1}{\sqrt{N}}
\sum_{i=1}^N
\left[
\varphi_{\cc,i}^{(k)}(\vgamma_k^\ast,\vomega^\ast)
-
\varphi_{\ic,i}^{(k)}(\vgamma_k^\ast,\vomega^\ast)
\right]
+
o_p(1).
\end{align*}
Define
\[
\varphi_{k,i}(\vgamma_k^\ast,\vomega^\ast)
=
\varphi_{\cc,i}^{(k)}(\vgamma_k^\ast,\vomega^\ast)
-
\varphi_{\ic,i}^{(k)}(\vgamma_k^\ast,\vomega^\ast).
\]
Then
\[
\sqrt{N}\widehat{\vDelta}_k
=
\frac{1}{\sqrt{N}}
\sum_{i=1}^N
\varphi_{k,i}(\vgamma_k^\ast,\vomega^\ast)
+
o_p(1).
\]
Correspondingly,
\[
\sqrt{N}\widehat{\vDelta}_k
\overset{d}{\longrightarrow}
N(\vzero,\Sigma_{\vDelta_k}),
\]
where
\[
\Sigma_{\vDelta_k}
=
\E\left[
\varphi_{k,i}(\vgamma_k^\ast,\vomega^\ast)
\varphi_{k,i}^{\top}(\vgamma_k^\ast,\vomega^\ast)
\right].
\]
This proves the consistency and asymptotic normality of
$\widehat{\vDelta}_k$.
\end{proof}

\subsection{Proof of Theorem~\ref{thrm:general_missing_patterns}}
\label{appendix:proof:thrm1}

\begin{proof}

To prove Theorem~\ref{thrm:general_missing_patterns}, we first stack the estimating equations and obtain:

\begin{equation*}
        \vzero = \sum^N_{i=1} \bm{\Psi}_i(\vtheta, \vgamma_\cc, \vgamma_\ic, \vomega) =
        \begin{bmatrix}
          \sum^N_{i=1}\vpsi_{\text{WCC}, i}(\vtheta, \vomega) \\[6pt]
          \sum^N_{i=1}\vpsi^{(1)}_{\cc, i}(\vgamma_{\cc,1}, \vomega) \\[2pt]
          \vdots \\[2pt]
          \sum^N_{i=1}\vpsi^{(K)}_{\cc, i}(\vgamma_{\cc,K}, \vomega) \\[6pt]
          \sum^N_{i=1}\vpsi^{(1)}_{\ic, i}(\vgamma_{\ic,1}, \vomega) \\[2pt]
          \vdots \\[2pt]
          \sum^N_{i=1}\vpsi^{(K)}_{\ic, i}(\vgamma_{\ic,K}, \vomega) \\[6pt]
          \sum^N_{i=1}\vphi_i(\vomega)
        \end{bmatrix}.
\end{equation*}

We denote $\widehat{\bm{\vartheta}} = (\widehat{\vtheta}_{\text{WCC}}, \widehat{\vgamma}_{\cc,1}, \dots, \widehat{\vgamma}_{\cc,K}, \widehat{\vgamma}_{\ic,1}, \dots, \widehat{\vgamma}_{\ic,K},  \widehat{\vomega})^\top$, and $\bm{\vartheta}^\ast = (\vtheta^\ast, \vgamma^\ast_1, \dots, \vgamma^\ast_K, \\ \vgamma^\ast_1, \dots, \vgamma^\ast_K, \vomega^\ast)^\top$. For simplicity, we will mark all the estimating equations with true parameters plugged in with an asterisk. For example, $\vpsi^{(i)}_{\text{WCC}}$ evaluated at $(\vomega^\ast, \vtheta^\ast)$ is $\vpsi^{\ast(i)}_{\text{WCC}}$. We follow~\citet[Theorem 5.21][]{van2000asymptotic} to make the following assumption:

\begin{assumption}
    \label{assumption:main}
    We assume that $\bm{\Psi}$ is a measurable vector-valued function and for every $\bm{\vartheta}_1$ and $\bm{\vartheta}_2$ in a neighborhood of $\bm{\vartheta}^\ast$,
    $$|| \bm{\Psi}(\vY; \bm{\vartheta}_1) - \bm{\Psi}(\vY; \bm{\vartheta}_2)|| \leq \dot{\bm{\Psi}}(\vY)||\bm{\vartheta}_1 - \bm{\vartheta}_2||,$$
    where the first-order derivative $\dot{\bm{\Psi}}$ is a measurable function with $P\dot{\bm{\Psi}}^2 < \infty$. Furthermore, assume that $P||\bm{\Psi}||^2 < \infty$ and $P\bm{\Psi}$ is differentiable at $\bm{\vartheta}^\ast$ with a nonsingular derivative matrix. Assume $P_n \bm{\Psi}(\vY; \widehat{\bm{\vartheta}}) = o_p(n^{-1/2})$ and $\widehat{\bm{\vartheta}} \overset{p}{\rightarrow} \bm{\vartheta}^\ast$.
\end{assumption}

The proof sketch of Theorem~\ref{thrm:general_missing_patterns} is that: 1) we derive the asymptotic linear expansion of the stacked estimator $\widehat{\bm{\vartheta}}$; 2) we show that $\widehat{\vtheta}_{\text{PS-PPI}}$ is a linear transformation of the elements in $\widehat{\bm{\vartheta}}$, thereby obtaining the asymptotic linear expansion of $\widehat{\vtheta}_{\text{PS-PPI}}$.

By working on the stacked estimating equations, we have the following asymptotic linear expansion:
\begin{equation*}
        \sqrt{N} (\widehat{\bm{\vartheta}} - \bm{\vartheta}^\ast) = \frac{1}{\sqrt{N}} \sum^N_{i=1} -\dot{\bm{\Psi}}^{-1} \bm{\Psi}^{\ast(i)} + o_p(1),
\end{equation*}
where
\begin{equation*}
    \dot{\bm{\Psi}}
    = \E\Bigl[\tfrac{\partial \bm{\Psi}^{(i)}}{\partial \bm{\vartheta}}\Bigr|_{\bm{\vartheta}=\bm{\vartheta}^\ast}\Bigr]
    = \E\!\left[\,
    \begin{bmatrix}
      \frac{\partial \vpsi^{(i)}_{\text{WCC}}}{\partial \vtheta} &
      \frac{\partial \vpsi^{(i)}_{\text{WCC}}}{\partial \vgamma_{\cc,1}} &
      \cdots &
      \frac{\partial \vpsi^{(i)}_{\text{WCC}}}{\partial \vgamma_{\cc,K}} &
      \frac{\partial \vpsi^{(i)}_{\text{WCC}}}{\partial \vgamma_{\ic,1}} &
      \cdots &
      \frac{\partial \vpsi^{(i)}_{\text{WCC}}}{\partial \vgamma_{\ic,K}} &
      \frac{\partial \vpsi^{(i)}_{\text{WCC}}}{\partial \vomega} \\[6pt]
      \frac{\partial \vpsi^{(1)}_{\cc,i}}{\partial \vtheta} &
      \frac{\partial \vpsi^{(1)}_{\cc,i}}{\partial \vgamma_{\cc,1}} &
      \cdots &
      \frac{\partial \vpsi^{(1)}_{\cc,i}}{\partial \vgamma_{\cc,K}} &
      \frac{\partial \vpsi^{(1)}_{\cc,i}}{\partial \vgamma_{\ic,1}} &
      \cdots &
      \frac{\partial \vpsi^{(1)}_{\cc,i}}{\partial \vgamma_{\ic,K}} &
      \frac{\partial \vpsi^{(1)}_{\cc,i}}{\partial \vomega} \\[2pt]
      \vdots & \vdots & \ddots & \vdots & \vdots & \ddots & \vdots & \vdots \\[2pt]
      \frac{\partial \vphi_i}{\partial \vtheta} &
      \frac{\partial \vphi_i}{\partial \vgamma_{\cc,1}} &
      \cdots &
      \frac{\partial \vphi_i}{\partial \vgamma_{\cc,K}} &
      \frac{\partial \vphi_i}{\partial \vgamma_{\ic,1}} &
      \cdots &
      \frac{\partial \vphi_i}{\partial \vgamma_{\ic,K}} &
      \frac{\partial \vphi_i}{\partial \vomega}
    \end{bmatrix}_{\bm{\vartheta}=\bm{\vartheta}^\ast}
    \,\right].
\end{equation*}

To obtain the inverse of $\dot{\bm{\Psi}}$, we first introduce the following notation:

$$\dot{\vpsi}_\infty = \E \left [ \frac{\partial \vpsi^{(i)}_{\text{WCC}}}{\partial \vtheta} \Bigr|_{(\vtheta, \vomega) = (\vtheta^\ast, \vomega^\ast)} \right ],~ \dot{\vpsi}_k = \E \left [\frac{\partial \vpsi^{(k)}_{\cc,i}}{\partial \vgamma_{\cc,k}} \Bigr|_{(\vgamma_{\cc,k}, \vomega) = (\vgamma^\ast_k, \vomega^\ast)} \right ] = \E \left [\frac{\partial \vpsi^{(k)}_{\ic,i}}{\partial \vgamma_{\ic,k}} \Bigr|_{(\vgamma_{\ic,k}, \vomega) = (\vgamma^\ast_k, \vomega^\ast)} \right ],$$
$$\dot{\vpsi}_\infty^\vomega = \E \left [ \frac{\partial \vpsi^{(i)}_{\text{WCC}}}{\partial \vomega} \Bigr|_{(\vtheta, \vomega) = (\vtheta^\ast, \vomega^\ast)} \right ],~ \dot{\vpsi}^\vomega_{\cc,k} = \E \left [\frac{\partial \vpsi^{(k)}_{\cc,i}}{\partial \vomega} \Bigr|_{(\vgamma_{\cc,k}, \vomega) = (\vgamma^\ast_k, \vomega^\ast)} \right ],$$
$$\dot{\vpsi}^\vomega_{\ic,k}  = \E \left [\frac{\partial \vpsi^{(k)}_{\ic,i}}{\partial \vomega} \Bigr|_{(\vgamma_{\ic,k}, \vomega) = (\vgamma^\ast_k, \vomega^\ast)} \right ],~ \dot{\vphi} = \E \left [ \frac{\partial \vphi_i}{\partial \vomega} \Bigr |_{\vomega = \vomega^\ast} \right ].$$

Then, by noticing that all the non-diagonal elements except those in the last column are $\vzero$, we have:
\begin{equation*}
    \dot{\bm{\Psi}}
    =
    \begin{bmatrix}
        \dot{\vpsi}_\infty & \vzero & \cdots & \dot{\vpsi}_\infty^{\vomega} \\
        \vzero & \dot{\vpsi}_{1} & \cdots & \dot{\vpsi}^{\vomega}_{\cc,1} \\
        \vdots & \vdots & \ddots & \vdots \\
        \vzero & \vzero & \cdots & \dot{\vphi}
    \end{bmatrix}
    =
    \begin{bmatrix}
        \Jb_{1} & \Jb_{2} \\
        \vzero & \dot{\vphi}
    \end{bmatrix}
\end{equation*}
where $\Jb_1 = \text{diag}(\dot{\vpsi}_\infty, \dot{\vpsi}_1, \cdots, \dot{\vpsi}_K, \dot{\vpsi}_1, \cdots \dot{\vpsi}_K)$, $\Jb_2 = (\dot{\vpsi}_\infty^\vomega, \dot{\vpsi}^\vomega_{\cc, 1}, \cdots, \dot{\vpsi}^\vomega_{\cc, K}, \dot{\vpsi}^\vomega_{\ic,1}, \cdots \dot{\vpsi}^\vomega_{\ic, K})^\top$. From \cite{lu2002inverses}, the inverse of such an upper triangular $2 \times 2$ block matrix $\dot{\bm{\Psi}}$ is:
\begin{equation*}
\dot{\bm{\Psi}}^{-1}
=
\begin{bmatrix}
  \dot{\vpsi}^{-1}_\infty                              & \vzero                                   & \cdots & -\dot{\vpsi}_\infty^{-1}\dot{\vpsi}_\infty^{\vomega}\dot{\vphi}^{-1}       \\[6pt]
  \vzero                                         & \dot{\vpsi}_{1}^{-1}                     & \cdots & -\dot{\vpsi}_{1}^{-1}\dot{\vpsi}_{\cc,1}^{\vomega}\dot{\vphi}^{-1} \\[2pt]
  \vdots                                         & \vdots                                   & \ddots & \vdots                                                   \\[2pt]
  \vzero                                         & \vzero                                   & \cdots & \dot{\vphi}^{-1}
\end{bmatrix}
\end{equation*}

From~\citet[Theorem 5.21][]{van2000asymptotic} and Assumption~\ref{assumption:main}, $\sqrt{N}(\widehat{\bm{\vartheta}} - \bm{\vartheta}^\ast)$ is asymptotically normal. To show that $\sqrt{N}(\widehat{\vtheta}_{\text{PS-PPI}} - \vtheta^\ast)$ is asymptotically normal, we want to show that $\widehat{\vtheta}_{\text{PS-PPI}}$ is a linear combination of the elements in $\widehat{\bm{\vartheta}}$. We know that $\widehat{\vtheta}_{\text{PS-PPI}} = \widehat{\vtheta}_{\text{WCC}} - \sum^{K}_{k = 1}\widehat{\Wb}_k (\widehat{\vgamma}_{\cc, k} - \widehat{\vgamma}_{\ic, k}) = \widehat{\vc}^\top \widehat{\bm{\vartheta}}$, where $\widehat{\vc} = (\Ib, -\widehat{\Wb}_1, \cdots, -\widehat{\Wb}_K, \widehat{\Wb}_1, \cdots, \widehat{\Wb}_K, \vzero)^\top$. On the other hand, since the true parameter of interest $\vtheta^\ast$ is the first element in $\bm{\vartheta}^\ast$, we can rewrite it as $\vtheta^\ast = \vtheta^\ast - \sum^K_{k=1} \widehat{\Wb}_k\vgamma^\ast_k  + \sum^K_{k=1} \widehat{\Wb}_k\vgamma^\ast_k = \widehat{\vc}^\top \bm{\vartheta}^\ast$. Define $\vc = (\Ib, -\Wb_1, \cdots, -\Wb_K, \Wb_1, \cdots, \\ \Wb_K, \vzero)^\top$; since $\widehat{\Wb}_k \overset{p}{\rightarrow} \Wb_k$, we have $\widehat{\vc} \overset{p}{\rightarrow} \vc$. The asymptotic linear expansion of $\sqrt{N}(\widehat{\vtheta}_{\text{PS-PPI}} - \vtheta^\ast)$ becomes:
\begin{equation*}
    \sqrt{N}(\widehat{\vtheta}_{\text{PS-PPI}} - \vtheta^\ast) = \widehat{\vc}^\top \sqrt{N}(\widehat{\bm{\vartheta}} - \bm{\vartheta}^\ast) = \frac{1}{\sqrt{N}} \sum^N_{i=1} -\vc^\top \dot{\bm{\Psi}}^{-1} \bm{\Psi}^{\ast(i)} + o_p(1).
\end{equation*}
Hence, $\sqrt{N}(\widehat{\vtheta}_{\text{PS-PPI}} - \vtheta^\ast)$ is asymptotically normal. We now derive the form of its asymptotic variance. First,
\begin{equation*}
    \vc^\top \dot{\bm{\Psi}}^{-1} =
    \begin{bmatrix}
        \dot{\vpsi}_\infty^{-1} \\[6pt]
        -\Wb_1 \dot{\vpsi}^{-1}_{1} \\[6pt]
        \vdots \\
        -\Wb_K \dot{\vpsi}^{-1}_{K} \\[6pt]
        \Wb_1 \dot{\vpsi}^{-1}_{1} \\[6pt]
        \cdots \\ \Wb_K \dot{\vpsi}^{-1}_{K} \\[6pt]
        - \left ( \dot{\vpsi}_\infty^{-1}\dot{\vpsi}_\infty^\vomega - \sum^K_{k=1} \Wb_k \dot{\vpsi}^{-1}_{k}(\dot{\vpsi}^\vomega_{\cc, k} - \dot{\vpsi}^\vomega_{\ic, k}) \right ) \dot{\vphi}^{-1}
    \end{bmatrix}^\top
\end{equation*}
Since $\widehat{\Wb}_k \overset{p}{\rightarrow} \Wb_k$, we have:
\begin{align}
    &\; \; \; \; \sqrt{N}(\widehat{\vtheta}_{\text{PS-PPI}} - \vtheta^\ast) \nonumber \\
    &= \frac{1}{\sqrt{N}} \sum^N_{i=1} -\vc^\top \dot{\bm{\Psi}}^{-1} \bm{\Psi}^{\ast(i)} + o_p(1) \nonumber \\
    &= \frac{1}{\sqrt{N}} \sum^N_{i=1} \left [ -\dot{\vpsi}_\infty^{-1}\vpsi^{\ast(i)}_{\text{WCC}} + \sum^K_{k=1} \Wb_k\dot{\vpsi}^{-1}_k (\vpsi^{(k)\ast}_{\cc,i} - \vpsi^{(k)\ast}_{\ic,i}) \right . \nonumber \\
    &\; \; \; +  \left . \left ( \dot{\vpsi}_\infty^{-1}\dot{\vpsi}_\infty^\vomega - \sum^K_{k=1} \Wb_k \dot{\vpsi}^{-1}_{k}(\dot{\vpsi}^\vomega_{\cc, k} - \dot{\vpsi}^\vomega_{\ic, k}) \right ) \dot{\vphi}^{-1}\vphi_i^\ast \right ] + o_p(1) \nonumber \\
    &= \frac{1}{\sqrt{N}} \sum^N_{i=1} \left [ \varphi_{\infty, i}(\vtheta^\ast, \vomega^\ast) - \Wb \varphi_{\vDelta, i}(\vgamma^\ast, \vomega^\ast) \right] + o_p(1) \label{eq:psppi_asymptotic_linear_expansion}
\end{align}
Denote $\varphi_i(\vtheta^\ast, \vgamma^\ast, \vomega^\ast, \Wb) := \varphi_{\infty, i}(\vtheta^\ast, \vomega^\ast) - \Wb \varphi_{\vDelta, i}(\vgamma^\ast, \vomega^\ast)$. Accordingly, we have the following asymptotic property: $\sqrt{N}(\widehat{\vtheta}_{\text{PS-PPI}} - \vtheta^\ast) \overset{d}{\rightarrow} \cN(0, \Sigma_{\text{PS-PPI}}(\vtheta^\ast, \vgamma^\ast, \vomega^\ast, \Wb))$, where
\begin{equation}
    \label{eq:full_asymptotic_variance_proof}
    \Sigma_{\text{PS-PPI}}(\vtheta^\ast, \vgamma^\ast, \vomega^\ast, \Wb) = \E \left [ \varphi_i(\vtheta^\ast, \vgamma^\ast, \vomega^\ast, \Wb) \varphi^\top_i(\vtheta^\ast, \vgamma^\ast, \vomega^\ast, \Wb) \right ].
\end{equation}
\end{proof}

\subsection{Proof of equivalences in Section~\ref{sec:incorporate_software}}
\label{sec:incorporate_software_proof}

In this subsection, we prove the statements made in
Section~\ref{sec:incorporate_software} one by one.

First, for a given $k \in \{1, \dots, K\}$, from Equation (\ref{eq:sigma_theta_delta_k}) we have
$$\widehat{\widetilde{\Sigma}}_{\vtheta, \vDelta_k} = \frac{1}{N} \sum^N_{i=1} \left [ \widetilde{\varphi}_{\infty, i}(\widehat{\vtheta}_{\text{WCC}}, \widehat{\vomega}) \widetilde{\varphi}_{\cc,i}^{\top(k)}(\widehat{\vgamma}_k,\widehat{\vomega}) - \widetilde{\varphi}_{\infty, i}(\widehat{\vtheta}_{\text{WCC}}, \widehat{\vomega}) \widetilde{\varphi}_{\ic,i}^{\top(k)}(\widehat{\vgamma}_k,\widehat{\vomega}) \right ].$$
We know that
\begin{align*}
    &\widetilde{\varphi}_{\infty, i}(\widehat{\vtheta}_{\text{WCC}}, \widehat{\vomega}) \widetilde{\varphi}_{\ic,i}^{\top(k)}(\widehat{\vgamma}_k,\widehat{\vomega}) \\
    &= \left\{ - \dot{\vpsi}_\infty^{-1}(\vtheta,\vomega) \frac{I(\cC_i = \infty)}{\pi_\infty(\boldsymbol{Y}_i; \widehat{\vomega})} \vpsi \left (\boldsymbol{Y}_i; \vtheta \right ) \right\} \left\{ -\dot{\vpsi}^{-1}_k(\vgamma_k,\vomega) \frac{I(\cC_i = k)}{\pi_k(G_{k}(\boldsymbol{Y}_i); \widehat{\vomega})} \vpsi \left ( \vYimp_{k, i}; \vgamma_{\ic, k} \right ) \right\}^\top,
\end{align*}
which equals zero because $I(\cC_i = \infty)I(\cC_i = k) = 0$. Therefore,
$$\widehat{\widetilde{\Sigma}}_{\vtheta, \vDelta_k} = \widehat{\widetilde{\Sigma}}_{\vtheta, \vgamma_{\cc, k}} = \frac{1}{N} \sum^N_{i=1} \left [ \widetilde{\varphi}_{\infty, i}(\widehat{\vtheta}_{\text{WCC}}, \widehat{\vomega}) \widetilde{\varphi}_{\cc,i}^{\top(k)}(\widehat{\vgamma}_k,\widehat{\vomega}) \right ],$$
where $\widehat{\widetilde{\Sigma}}_{\vtheta, \vgamma_{\cc, k}}$ is a conservative estimate of $\Sigma_{\vtheta, \vgamma_{\cc, k}} = \E \left [ \varphi_{\infty, i}(\vtheta^\ast, \vomega^\ast) \varphi_{\cc,i}^{\top(k)}(\vgamma^\ast_k,\vomega^\ast) \right ]$, the asymptotic covariance matrix between $\widehat{\vtheta}_{\text{WCC}}$ and $\widehat{\vgamma}_{\cc, k}$.

Second, given a pair of $k, \ell \in \{1, \dots, K\}$, from Equation (\ref{eq:sigma_delta_kl_block}) we have
\begin{align*}
\widehat{\widetilde{\Sigma}}_{\vDelta_k,\vDelta_\ell} &= \frac{1}{N} \sum^N_{i=1} \left [ \widetilde{\varphi}_{\cc,i}^{(k)}(\widehat{\vgamma}_k,\widehat{\vomega}) \widetilde{\varphi}_{\cc,i}^{\top(\ell)}(\widehat{\vgamma}_\ell,\widehat{\vomega}) - \widetilde{\varphi}_{\cc,i}^{(k)}(\widehat{\vgamma}_k,\widehat{\vomega}) \widetilde{\varphi}_{\ic,i}^{\top(\ell)}(\widehat{\vgamma}_\ell,\widehat{\vomega}) \right . \\
&\left . - \widetilde{\varphi}_{\ic,i}^{(k)}(\widehat{\vgamma}_k,\widehat{\vomega}) \widetilde{\varphi}_{\cc,i}^{\top(\ell)}(\widehat{\vgamma}_\ell,\widehat{\vomega}) + \widetilde{\varphi}_{\ic,i}^{(k)}(\widehat{\vgamma}_k,\widehat{\vomega}) \widetilde{\varphi}_{\ic,i}^{\top(\ell)}(\widehat{\vgamma}_\ell,\widehat{\vomega})\right ].
\end{align*}
Similar to how we proved the first statement, the second and third terms are equal to zero because $I(\cC_i=\infty)I(\cC_i=\ell)=0$ and $I(\cC_i=k)I(\cC_i=\infty)=0$. Moreover, the fourth term is equal to zero when \(k\neq \ell\), because $I(\cC_i=k)I(\cC_i=\ell)=0$. Therefore, only the first term is nonzero when \(k\neq \ell\), whereas both the first and fourth terms are nonzero when \(k=\ell\). Hence,
\[
\widehat{\widetilde{\Sigma}}_{\vDelta_k,\vDelta_\ell}
=
\begin{cases}
\widehat{\widetilde{\Sigma}}_{\vgamma_{\cc,k},\vgamma_{\cc,\ell}},
& k\neq \ell, \\[6pt]
\widehat{\widetilde{\Sigma}}_{\vgamma_{\cc,k},\vgamma_{\cc,k}}
+
\widehat{\widetilde{\Sigma}}_{\vgamma_{\ic,k},\vgamma_{\ic,k}},
& k=\ell.
\end{cases}
\]
Here,
\[
\widehat{\widetilde{\Sigma}}_{\vgamma_{\cc,k},\vgamma_{\cc,\ell}}
=
\frac{1}{N}\sum_{i=1}^N
\widetilde{\varphi}_{\cc,i}^{(k)}(\widehat{\vgamma}_k,\widehat{\vomega})
\widetilde{\varphi}_{\cc,i}^{\top(\ell)}(\widehat{\vgamma}_\ell,\widehat{\vomega})
\]
and
\[
\widehat{\widetilde{\Sigma}}_{\vgamma_{\ic,k},\vgamma_{\ic,k}}
=
\frac{1}{N}\sum_{i=1}^N
\widetilde{\varphi}_{\ic,i}^{(k)}(\widehat{\vgamma}_k,\widehat{\vomega})
\widetilde{\varphi}_{\ic,i}^{\top(k)}(\widehat{\vgamma}_k,\widehat{\vomega})
\]
are conservative estimates of
\[
\Sigma_{\vgamma_{\cc,k},\vgamma_{\cc,\ell}}
=
\E\left[
\varphi_{\cc,i}^{(k)}(\vgamma_k^\ast,\vomega^\ast)
\varphi_{\cc,i}^{\top(\ell)}(\vgamma_\ell^\ast,\vomega^\ast)
\right]
\]
and
\[
\Sigma_{\vgamma_{\ic,k},\vgamma_{\ic,k}}
=
\E\left[
\varphi_{\ic,i}^{(k)}(\vgamma_k^\ast,\vomega^\ast)
\varphi_{\ic,i}^{\top(k)}(\vgamma_k^\ast,\vomega^\ast)
\right],
\]
respectively. In other words,
\(\widehat{\widetilde{\Sigma}}_{\vgamma_{\cc,k},\vgamma_{\cc,\ell}}\)
estimates the asymptotic covariance matrix between
\(\widehat{\vgamma}_{\cc,k}\) and \(\widehat{\vgamma}_{\cc,\ell}\), whereas
\(\widehat{\widetilde{\Sigma}}_{\vgamma_{\ic,k},\vgamma_{\ic,k}}\)
estimates the asymptotic variance-covariance matrix of
\(\widehat{\vgamma}_{\ic,k}\). This completes the proof of the second statement.

\subsection{Proof of optimally weighted PS-PPI as a \textit{safe} estimator in~\citet{xu2025unified}}
\label{sec:psppi_safe_xu}

\begin{proof}
Consider the outcome-missing setting under MCAR, and suppose the observed data are $\{\vX_i,Y_i, \cC_i\}^N_{i=1}$, where $\vX_i$ are the covariates, $Y_i$ is the outcome, and $\cC_i$ is the missingness index, with $\cC_i = \infty$ indicating that the subject is fully observed and $\cC_i = 1$ indicating that $Y_i$ is missing. We define $n_\ell = \sum I(\cC_i = \infty)$ as the number of subjects with observed $Y_i$ and $n_u = N - n_\ell$ as the number of subjects with missing $Y_i$. Under MCAR, we denote $\Pr(\cC_i = \infty) = \pi_\infty$, and therefore, as $N \rightarrow \infty$, we have
$$\frac{n_u}{n_\ell+n_u} \rightarrow 1-\pi_\infty,$$
which is the quantity denoted by $\gamma$ in~\citet{xu2025unified}. The predictive model \(\widehat{Y} = h(\vX) \) is assumed to be trained on external data, so it can be treated as fixed and independent of the labeled and unlabeled subsets of data.

To prove that $\widehat{\vtheta}_\psppi$ is a \textit{safe} estimator, we first derive the $i$-th influence function of $\widehat{\vtheta}_\psppi$ in this setting. Then we construct the \textit{safe} estimator defined in~\citet{xu2025unified} and show that the two estimators have the same $i$-th influence function.

According to Section~\ref{sec:methodology}, the asymptotic linear expansion of $\widehat{\vtheta}_\psppi$ is:
\begin{align}
    &\sqrt{N}(\widehat{\vtheta}_{\psppi} - \vtheta^\ast)
    = \frac{1}{\sqrt{N}} \sum_{i=1}^N
    \left\{
    - \dot{\vpsi}^{-1}_{\infty}(\vtheta^\ast)
    \frac{I(\cC_i = \infty)}{\pi_\infty}
    \vpsi(\vX_i, Y_i; \vtheta^\ast) \right .
    \notag \\
    &\qquad \left .
    + \Wb \dot{\vpsi}^{-1}_{1}
    \vpsi(\vX_i, h(\vX_i); \vgamma^\ast_1)
    \left[
    \frac{I(\cC_i = \infty)}{\pi_\infty}
    -
    \frac{I(\cC_i = 1)}{1 - \pi_\infty}
    \right]
    \right\}
    + o_p(1).
\end{align}
The corresponding asymptotic variance-covariance matrix of $\sqrt{N}(\widehat{\vtheta}_{\psppi} - \vtheta^\ast)$ is:
\begin{align}
    \Sigma_{\psppi}(\Wb)
    &=
    \frac{1}{\pi_\infty}
    \Var \left (- \dot{\vpsi}_{\infty}^{-1}
    \vpsi(\vX_i, Y_i; \vtheta^\ast) + \Wb \dot{\vpsi}_{1}^{-1}
    \vpsi(\vX_i, h(\vX_i); \vgamma_1^\ast) \right )
    \notag\\
    &\quad
    +
    \frac{1}{1-\pi_\infty} \Var \left ( \Wb \dot{\vpsi}^{-1}_{1} \vpsi(\vX_i, h(\vX_i); \vgamma^\ast_1) \right ),
\end{align}
where the minimum variance is attained at the following optimal choice of $\Wb$:
$$\Wb_\opt = (1 - \pi_\infty) \Cov \left (\dot{\vpsi}_{\infty}^{-1}
    \vpsi(\vX_i, Y_i; \vtheta^\ast), \dot{\vpsi}_{1}^{-1}
    \vpsi(\vX_i, h(\vX_i); \vgamma_1^\ast) \right ) \Var \left ( \dot{\vpsi}_{1}^{-1}
    \vpsi(\vX_i, h(\vX_i); \vgamma_1^\ast) \right )^{-1}.$$
Under MCAR, the sample estimate of the constant $\pi_\infty$ multiplies each corresponding estimating equation by a common scalar and therefore does not change its root.

On the other hand, constructing the \textit{safe} estimator defined in~\citet{xu2025unified} requires two terms, $g(\vX_i)$ and $\varphi_{\tilde{\eta}_n}(Z_i)$, which we define as
$$g(\vX_i) = -\dot{\vpsi}_{1}^{-1}\vpsi(\vX_i, h(\vX_i); \vgamma_1^\ast),$$
$$\varphi_{\tilde{\eta}_n}(Z_i) = - \dot{\vpsi}^{-1}_{\infty}(\vtheta^\ast) \vpsi(\vX_i, Y_i; \vtheta^\ast), \qquad Z_i=(\vX_i,Y_i).$$
Since $\E [g(\vX_i)] = 0$, we have
\begin{align}
    \Bb^g &= \E \left [ \varphi_{\tilde{\eta}_n}(Z_i) g(\vX_i)^\top \right ] \E \left [ g(\vX_i) g(\vX_i)^\top \right ]^{-1} \notag \\
    &= \Cov \left (\dot{\vpsi}_{\infty}^{-1}
    \vpsi(\vX_i, Y_i; \vtheta^\ast), \dot{\vpsi}_{1}^{-1}
    \vpsi(\vX_i, h(\vX_i); \vgamma_1^\ast) \right ) \Var \left ( \dot{\vpsi}_{1}^{-1}
    \vpsi(\vX_i, h(\vX_i); \vgamma_1^\ast) \right )^{-1} \notag \\
    &= \frac{1}{1 - \pi_\infty} \Wb_\opt.
\end{align}
Because $(1-\pi_\infty)\Bb^g=\Wb_\opt$, rewriting $\Bb^g$ using $\Wb_\opt$ in the asymptotic variance-covariance matrix of $\sqrt{n_\ell}(\widehat{\vtheta}^{safe} - \vtheta^\ast)$ given in~\citet{xu2025unified} yields
\begin{align}
    \Sigma^{\textit{safe}}
    &= \Var \left [ - \dot{\vpsi}^{-1}_{\infty}(\vtheta^\ast) \vpsi(\vX_i, Y_i; \vtheta^\ast) + \Wb_\opt \dot{\vpsi}_{1}^{-1}\vpsi(\vX_i, h(\vX_i); \vgamma_1^\ast)  \right ] \notag \\
    &\quad + \frac{\pi_\infty}{1-\pi_\infty}\Var \left [ \Wb_\opt \dot{\vpsi}_{1}^{-1}\vpsi(\vX_i, h(\vX_i); \vgamma_1^\ast) \right ] \notag \\
    &= \pi_\infty \Sigma_{\psppi}(\Wb_\opt).
\end{align}
Therefore, by Slutsky's theorem, we know the asymptotic variance-covariance matrix of $\sqrt{N}(\widehat{\vtheta}^{safe} - \vtheta^\ast)$ is $\Sigma_{\psppi}(\Wb_\opt)$, which proves that $\widehat{\vtheta}_\psppi$ is a \textit{safe} estimator.
\end{proof}

\subsection{Proof of PS-PPI as a subclass of AIPWCC estimators}
\label{sec:proof_psppi_aipwcc}

We first consider the case in which the missing data model $\pi$ is known. A similar conclusion can be derived when the missing data model is unknown but correctly specified and estimated, because one only needs to further expand $\widehat{\vomega}$ in the leading influence function $\widetilde{\varphi}_i$. From~\citet{tsiatis2006semiparametric}, the AIPWCC influence functions are of the following form:
\begin{align}
\widetilde{\varphi}_{\mathrm{AIPWCC},i}
&=
-\dot{\vpsi}_\infty^{-1}
\left[
\frac{I(\cC_i=\infty)}{\pi_\infty(\vY_i)}
\vpsi(\vY_i;\vtheta^\ast)
+
L_2(\cC_i, G_{\cC_i}(\vY_i), \vomega^\ast)
\right],
\label{eq:aipwcc_if_psppi_proof}
\end{align}
where
$$L_2(\cC_i, G_{\cC_i}(\vY_i), \vomega^\ast) = \frac{I(\cC_i=\infty)}{\pi_\infty(\vY_i)} \sum_{k=1}^K \pi_k\{G_k(\vY_i)\} L_{2k}\{G_k(\vY_i)\} - \sum_{k=1}^K I(\cC_i=k) L_{2k}\{G_k(\vY_i)\},$$
and \(L_{2k}\{G_k(\vY_i)\}\) is an arbitrary square-integrable
augmentation function of the observed variables $G_k(\vY_i)$ under the $k$-th missingness pattern. From Theorem~\ref{thrm:general_missing_patterns},
the leading influence function of PS-PPI with fixed \(\Wb=(\Wb_1,\ldots,\Wb_K)\)
and known $\pi$ is
\begin{align}
\widetilde{\varphi}_i(\vtheta^\ast, \vgamma^\ast, \vomega^\ast, \Wb)
&=
-\dot{\vpsi}_\infty^{-1}
\frac{I(\cC_i=\infty)}{\pi_\infty(\vY_i)}
\vpsi(\vY_i;\vtheta^\ast) \nonumber\\
&\quad+
\sum_{k=1}^K
\Wb_k\dot{\vpsi}_k^{-1}
\left[
\frac{I(\cC_i=\infty)}{\pi_\infty(\vY_i)}
-
\frac{I(\cC_i=k)}{\pi_k\{G_k(\vY_i)\}}
\right]
\vpsi\left(\vYimp_{k,i};\vgamma_k^\ast\right).
\label{eq:psppi_known_ps_if_proof}
\end{align}
This indicates that for each $k$, $L_{2k}$ in Equation (\ref{eq:aipwcc_if_psppi_proof}) is defined as:
\begin{equation}
\label{eq:psppi_aipwcc_l2k_choice}
L_{2k}\{G_k(\vY_i)\}
=
-
\frac{1}{\pi_k\{G_k(\vY_i)\}}
\dot{\vpsi}_\infty
\Wb_k
\dot{\vpsi}_k^{-1}
\vpsi\left(\vYimp_{k,i};\vgamma_k^\ast\right).
\end{equation}
Note that by definition, $\vYimp_{k, i}$ is a function of the observed component $G_k(\vY_i)$ given $h_k$. Therefore, PS-PPI estimators are a subclass of the AIPWCC estimators when $\pi$ is known.

\subsection{Proof of Proposition~\ref{proposition:equivalence_gronsbell}}
\label{sec:proof_equivalence_gronsbell}

\begin{proof}
Consider the setting of Proposition~\ref{proposition:equivalence_gronsbell}: in addition to the fully observed subjects, there is a single missingness pattern ($K = 1$) under which only the outcome is missing, and the missingness is MCAR with known probabilities $\pi_\infty$ and $\pi_1 = 1 - \pi_\infty$. Write $\vpsi^\ast_i = \vpsi(\vY_i; \vtheta^\ast)$ and $\vpsi^\ast_{1, i} = \vpsi(\vYimp_{1, i}; \vgamma_1^\ast)$, both of which have mean zero. From~\citet{gronsbell2024another}, the $i$-th influence function of $\widehat{\vtheta}_{\text{Chen}}$ is
$$
\varphi_{\text{Chen}, i} = -\dot{\vpsi}_\infty^{-1} \left\{ \frac{I(\cC_i = \infty)}{\pi_\infty} \vpsi^\ast_i + \Cov(\vpsi^\ast_i, \vpsi^\ast_{1, i}) \Var(\vpsi^\ast_{1, i})^{-1} \vpsi^\ast_{1, i} \, \frac{\pi_\infty - I(\cC_i = \infty)}{\pi_\infty} \right\},
$$
where $-\dot{\vpsi}_\infty^{-1}$ coincides with the matrix $\Ab = \E[\Xb \Xb^\top]^{-1}$ of~\citet{gronsbell2024another} in the linear regression setting they consider. On the other hand, by Equation~(\ref{eq:psppi_known_ps_if_proof}) with $K = 1$, the influence function of $\widehat{\vtheta}_{\text{PS-PPI}}$ with a fixed weight $\Wb = \Wb_1$ is
$$
\widetilde{\varphi}_i(\vtheta^\ast, \vgamma^\ast, \vomega^\ast, \Wb) = -\dot{\vpsi}_\infty^{-1} \frac{I(\cC_i = \infty)}{\pi_\infty} \vpsi^\ast_i + \Wb_1 \dot{\vpsi}_1^{-1} \left( \frac{I(\cC_i = \infty)}{\pi_\infty} - \frac{I(\cC_i = 1)}{\pi_1} \right) \vpsi^\ast_{1, i}.
$$
By Corollary~\ref{corollary:optimal_w}, the optimal weight is $\Wb_\opt = \Sigma_{\vtheta, \vDelta} \Sigma_{\vDelta}^{-1}$. Since $I(\cC_i = \infty) I(\cC_i = 1) = 0$ and the missingness is MCAR,
$$
\Sigma_{\vtheta, \vDelta} = \dot{\vpsi}_\infty^{-1} \E\left[ \frac{I(\cC_i = \infty)}{\pi_\infty^2} \vpsi^\ast_i \vpsi^{\ast\top}_{1, i} \right] \dot{\vpsi}_1^{-1\top} = \frac{1}{\pi_\infty} \dot{\vpsi}_\infty^{-1} \Cov(\vpsi^\ast_i, \vpsi^\ast_{1, i}) \dot{\vpsi}_1^{-1\top},
$$
$$
\Sigma_{\vDelta} = \dot{\vpsi}_1^{-1} \E\left[ \left( \frac{I(\cC_i = \infty)}{\pi_\infty^2} + \frac{I(\cC_i = 1)}{\pi_1^2} \right) \vpsi^\ast_{1, i} \vpsi^{\ast\top}_{1, i} \right] \dot{\vpsi}_1^{-1\top} = \frac{1}{\pi_\infty \pi_1} \dot{\vpsi}_1^{-1} \Var(\vpsi^\ast_{1, i}) \dot{\vpsi}_1^{-1\top},
$$
where the last equality uses $1/\pi_\infty + 1/\pi_1 = 1/(\pi_\infty \pi_1)$. Here, $\Sigma_{\vtheta, \vDelta}$ is the asymptotic covariance matrix between $\widehat{\vtheta}_{\cc}$ and $\widehat{\vgamma}_{\cc, 1}$, and $\Sigma_{\vDelta}$ is the sum of the asymptotic variance-covariance matrices of $\widehat{\vgamma}_{\cc, 1}$ and $\widehat{\vgamma}_{\ic, 1}$. Therefore,
\begin{align*}
\Wb_\opt &= \left( \frac{1}{\pi_\infty} \dot{\vpsi}_\infty^{-1} \Cov(\vpsi^\ast_i, \vpsi^\ast_{1, i}) \dot{\vpsi}_1^{-1\top} \right) \left( \dot{\vpsi}_1^{\top} \Var(\vpsi^\ast_{1, i})^{-1} \dot{\vpsi}_1 \right) \pi_\infty \pi_1 \\
&= \pi_1 \dot{\vpsi}_\infty^{-1} \Cov(\vpsi^\ast_i, \vpsi^\ast_{1, i}) \Var(\vpsi^\ast_{1, i})^{-1} \dot{\vpsi}_1.
\end{align*}
Because $I(\cC_i = 1) = 1 - I(\cC_i = \infty)$ and $\pi_1 = 1 - \pi_\infty$, we have $\pi_1 \{ I(\cC_i = \infty)/\pi_\infty - I(\cC_i = 1)/\pi_1 \} = \{ I(\cC_i = \infty) - \pi_\infty \}/\pi_\infty$. Plugging $\Wb_\opt$ into $\widetilde{\varphi}_i$ then gives
\begin{align*}
& \widetilde{\varphi}_i(\vtheta^\ast, \vgamma^\ast, \vomega^\ast, \Wb_\opt) \\
&\quad = -\dot{\vpsi}_\infty^{-1} \frac{I(\cC_i = \infty)}{\pi_\infty} \vpsi^\ast_i + \dot{\vpsi}_\infty^{-1} \Cov(\vpsi^\ast_i, \vpsi^\ast_{1, i}) \Var(\vpsi^\ast_{1, i})^{-1} \vpsi^\ast_{1, i} \, \frac{I(\cC_i = \infty) - \pi_\infty}{\pi_\infty} \\
&\quad = -\dot{\vpsi}_\infty^{-1} \left\{ \frac{I(\cC_i = \infty)}{\pi_\infty} \vpsi^\ast_i + \Cov(\vpsi^\ast_i, \vpsi^\ast_{1, i}) \Var(\vpsi^\ast_{1, i})^{-1} \vpsi^\ast_{1, i} \, \frac{\pi_\infty - I(\cC_i = \infty)}{\pi_\infty} \right\} \\
&\quad = \varphi_{\text{Chen}, i}.
\end{align*}
Therefore, the two estimators are asymptotically equivalent.
\end{proof}

\subsection{Proof of Proposition~\ref{proposition:equivalence_kluger}}
\label{sec:proof_equivalence_kluger}

\begin{proof}
In the setting of~\citet{kluger2025prediction}, there are two missingness patterns, namely the fully observed pattern and a single pattern ($K = 1$) under which the same variables are missing, and the probability of each pattern is known. From~\citet{kluger2025prediction}, the asymptotic variance-covariance matrix of $\widehat{\vtheta}_{\text{PTD}}$ is
$$
\Sigma_{\text{PTD}} = \Sigma_{\vtheta} - \Sigma_{\vtheta, \vgamma_{\cc, 1}} \left( \Sigma_{\vgamma_{\cc, 1}} + \Sigma_{\vgamma_{\ic, 1}} \right)^{-1} \Sigma_{\vtheta, \vgamma_{\cc, 1}}^\top,
$$
where, in our notation, $\Sigma_{\vtheta}$ is the asymptotic variance-covariance matrix of $\widehat{\vtheta}_{\wcc}$; $\Sigma_{\vgamma_{\cc, 1}}$ and $\Sigma_{\vgamma_{\ic, 1}}$ are the asymptotic variance-covariance matrices of $\widehat{\vgamma}_{\cc, 1}$ and $\widehat{\vgamma}_{\ic, 1}$, respectively; and $\Sigma_{\vtheta, \vgamma_{\cc, 1}}$ is the asymptotic covariance matrix between $\widehat{\vtheta}_{\wcc}$ and $\widehat{\vgamma}_{\cc, 1}$.

On the other hand, with $K = 1$ and known missingness probabilities, Equation~(\ref{eq:optimal_asymptotic_variance}) becomes $\Sigma_{\text{PS-PPI}}(\vtheta^\ast, \vgamma^\ast, \vomega^\ast, \Wb_\opt) = \Sigma_{\vtheta} - \Sigma_{\vtheta, \vDelta} \Sigma_{\vDelta}^{-1} \Sigma_{\vtheta, \vDelta}^\top$ with $\vDelta = \vDelta_1$. Because $\widehat{\vgamma}_{\ic, 1}$ is computed from the subjects with $\cC_i = 1$, whereas $\widehat{\vtheta}_{\wcc}$ and $\widehat{\vgamma}_{\cc, 1}$ are computed from the fully observed subjects, $I(\cC_i = \infty) I(\cC_i = 1) = 0$ implies $\Sigma_{\vtheta, \vDelta} = \Sigma_{\vtheta, \vgamma_{\cc, 1}}$ and $\Sigma_{\vDelta} = \Sigma_{\vgamma_{\cc, 1}} + \Sigma_{\vgamma_{\ic, 1}}$ (see also Section~\ref{sec:incorporate_software_proof}). Hence,
$$
\Sigma_{\text{PS-PPI}}(\vtheta^\ast, \vgamma^\ast, \vomega^\ast, \Wb_\opt) = \Sigma_{\vtheta} - \Sigma_{\vtheta, \vgamma_{\cc, 1}} \left( \Sigma_{\vgamma_{\cc, 1}} + \Sigma_{\vgamma_{\ic, 1}} \right)^{-1} \Sigma_{\vtheta, \vgamma_{\cc, 1}}^\top = \Sigma_{\text{PTD}}.
$$
This shows that $\widehat{\vtheta}_{\text{PTD}}$ and $\widehat{\vtheta}_{\text{PS-PPI}}$ have the same asymptotic variance-covariance matrix. As they also have the same asymptotic mean, they are asymptotically equivalent.
\end{proof}

\section{Supplementary Preliminary}

\subsection{Differences between MCAR and MAR under outcome-missing and general missingness settings}
\label{sec:mcar-vs-mar}

Figure~\ref{fig:mcar-vs-mar} provides a four-way comparison across the two dimensions of MCAR vs.\ MAR and general missingness vs.\ outcome-only missingness. It can be seen that first, under MCAR, the missingness probabilities are constants. However, under MAR imposed in Assumption~\ref{assumption:missing-data}, the missing probabilities become functions of observed variables. Second, under outcome missingness, these missing probabilities are functions of $X_1, X_2$ and $X_3$, because they are always observed across subjects. On the other hand, under general missingness imposed in Assumption~\ref{assumption:missing-data}, the missing probabilities are functions of different variables across patterns, depending on which variables are observed in that pattern. For example, the probability of $\cC_i = 1$ (missing $X_2$ and $X_3$ simultaneously) is denoted by $\pi_1(Y, X_1)$, a function of the observed variables $Y$ and $X_1$ under this pattern.

\begin{figure}[!tp]
    \centering
    \includegraphics[width=1.0\linewidth]{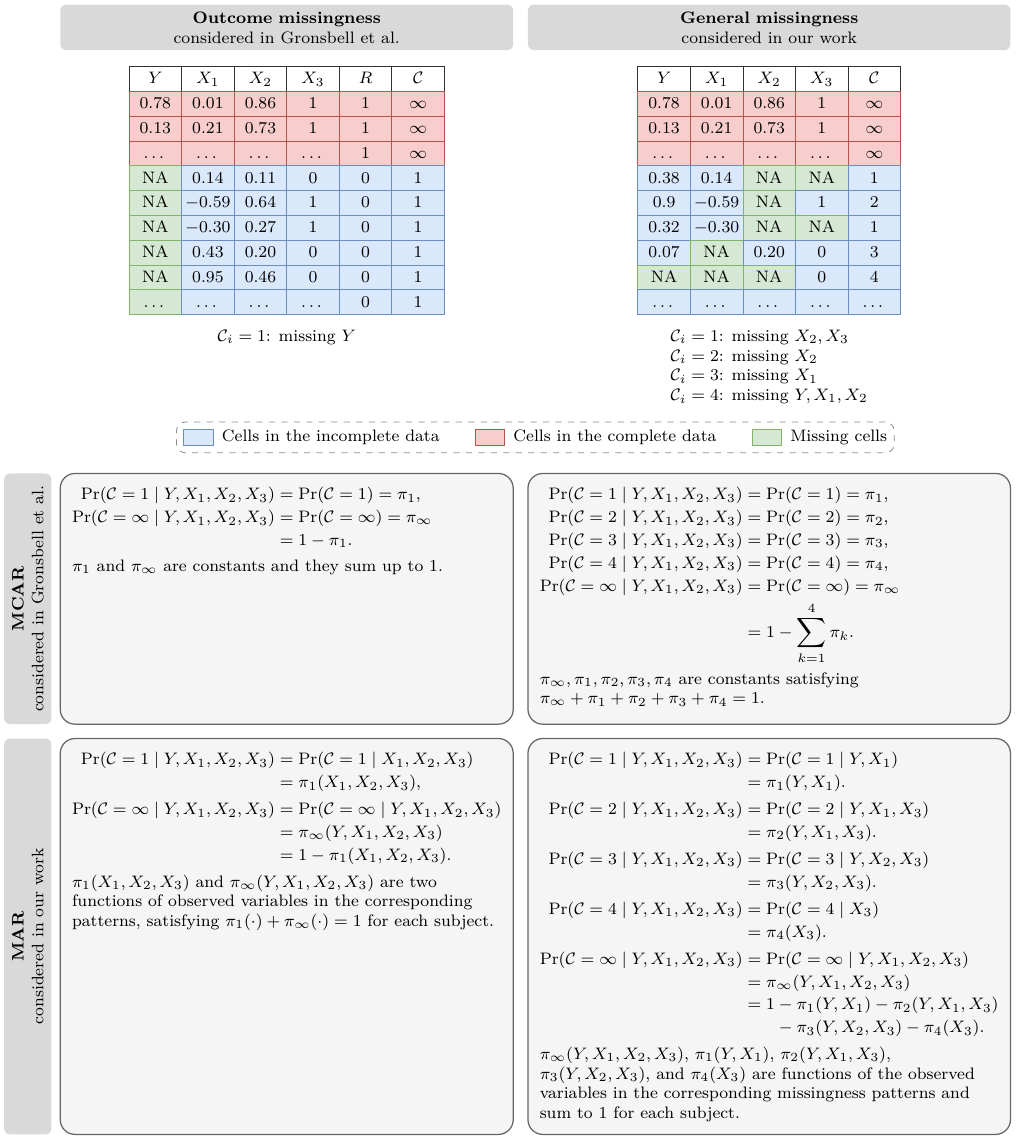}
    \caption{Comparison of the MCAR and MAR assumptions under outcome-missing and general missingness settings. The data comprise four variables, $(Y, X_1, X_2, X_3)$. Under general missingness, four distinct missingness patterns are considered in addition to the fully observed pattern.}
    \label{fig:mcar-vs-mar}
\end{figure}

\subsection{Base predictive models and pattern-specific predictive functions}
\label{sec:base-prediction-model}

For each missingness pattern $k$, the predictive function $h_k$ maps the variables observed under that pattern to predictions for the variables missing under that pattern. When always observed variables are available, different predictive functions can be assembled from a shared collection of base predictive models. In particular, let $\vY_{c, i}$ denote the collection of variables that are fully observed in the data. This collection includes all type (c) variables and any type (a) or (b) variables that are fully observed. Let $\mathcal{J}_{\mis}\subseteq\{1,\ldots,q\}$ be the index set of variables subject to missingness. For each $j\in\mathcal{J}_{\mis}$, one may obtain an externally trained base predictive model
\begin{equation*}
    g_j: \cY_c \mapsto \cY_j \subseteq \R,
\end{equation*}
where $\cY_c$ is the sample space of $\vY_{c, i}$. The predictive function $h_k$ can then be constructed as
\begin{equation*}
    h_k\{G_k(\vY_i)\}
    =
    \left(g_j(\vY_{c,i}), \; \text{where $j$ satisfies } R^{(k)}_j=0 \right).
\end{equation*}
Instead of a separate predictive model for every missingness pattern, this construction requires at most $q$ base predictive models, and they are reused across different missingness patterns. To further illustrate the idea, consider the following example:

\begin{example}
\label{example:base-predictive-models}
\end{example}
Consider a dataset with three variables: diabetes status, BMI, and smoking status. Suppose BMI is always observed, whereas diabetes and smoking status may be missing, and consider three missingness patterns: 1) only diabetes status is missing; 2) only smoking status is missing; and 3) both diabetes and smoking status are missing. The corresponding predictive functions are
\begin{align*}
    h_1 &: (\texttt{BMI},\texttt{smoking}) \mapsto \texttt{diabetes}, \\
    h_2 &: (\texttt{BMI},\texttt{diabetes}) \mapsto \texttt{smoking}, \\
    h_3 &: \texttt{BMI} \mapsto (\texttt{diabetes},\texttt{smoking}).
\end{align*}
Since BMI is always observed, we can use it as the predictor for the other two variables and train the following two base predictive models:
\begin{align*}
    g_{\texttt{d}} : \texttt{BMI} \mapsto \texttt{diabetes}, \qquad g_{\texttt{s}} : \texttt{BMI} \mapsto \texttt{smoking}.
\end{align*}
The three predictive functions can thereby be constructed as follows:
\begin{align*}
    h_1(\texttt{BMI},\texttt{smoking}) &= g_{\texttt{d}}(\texttt{BMI}), \\
    h_2(\texttt{BMI},\texttt{diabetes}) &= g_{\texttt{s}}(\texttt{BMI}), \\
    h_3(\texttt{BMI}) &= \bigl(g_{\texttt{d}}(\texttt{BMI}),g_{\texttt{s}}(\texttt{BMI})\bigr).
\end{align*}
In particular, although $h_1$ and $h_2$ can take inputs other than BMI, they use only BMI and predict the missing variables respectively through $g_{\texttt{d}}$ and $g_{\texttt{s}}$. $h_3$, on the other hand, applies both base predictive models to predict the missing variables. With this construction, we require only two base predictive models to construct the three predictive functions.

In the \textit{All of Us} analysis detailed in Section~\ref{sec:data_example}, the $h_k$ for all retained missingness patterns are constructed from two base predictive models, one for race and one for C-reactive protein (CRP) level.

\subsection{Detailed derivation of Example~\ref{example:mean}}
\label{sec:mean-example-derivation}

In Example~\ref{example:mean}, we assume that the data consist of two partitions of sizes $m$ and $M$ and are independently and identically distributed according to a population distribution with marginal mean outcome $\mu=\E[Y]$. For the first partition, we observe the outcome $Y^{\tt lab}_i$ and its corresponding ML prediction $\widehat{Y}^{\tt lab}_i$: $\cD_{\tt lab} = \{(Y^{\tt lab}_i, \widehat{Y}^{\tt lab}_i)\}^m_{i=1}$, which is referred to as the ``labeled subset''. For the second, often larger, partition of data, we do not observe the outcome but can obtain its ML prediction, i.e., $\cD_{\tt unlab} = \{\widehat{Y}^{\tt unlab}_i\}^M_{i=1}$, referred to as the ``unlabeled subset''. When the outcome is assumed to be missing completely at random, the prediction-powered inference estimator is defined as
$$\widehat{\theta}_{\text{PP}} = \widehat{\theta}_{\text{CC}}-\widehat{\Delta}=\bar{Y}^{\tt lab} -(\bar{\widehat{Y}}^{\tt lab} - \bar{\widehat{Y}}^{\tt unlab}),$$
where $\widehat{\Delta} = \bar{\widehat{Y}}^{\tt lab} - \bar{\widehat{Y}}^{\tt unlab}$ is an estimate of zero. By rewriting the above equation, we have
\begin{align*}
\widehat{\theta}_{\mathrm{PP}}
&=
\widehat{\theta}_{\mathrm{CC}}-\widehat{\Delta} =
\bar{Y}^{\tt lab}
-
\left(
\bar{\widehat{Y}}^{\tt lab}
-
\bar{\widehat{Y}}^{\tt unlab}
\right) \\
&=
\bar{Y}^{\tt lab}
-
\bar{\widehat{Y}}^{\tt lab}
+
\bar{\widehat{Y}}^{\tt unlab} \\
&=
\frac{1}{m}\sum_{i=1}^m
\left(
Y_i^{\tt lab}
-
\widehat{Y}_i^{\tt lab}
\right)
+
\frac{1}{M}\sum_{i=1}^M
\widehat{Y}_i^{\tt unlab}.
\end{align*}

When computing the variance of $\widehat{\theta}_{\text{PP}}$, it is worth noting that since the labeled subset and the unlabeled subset are independent of each other, the covariance between $\sum_{i=1}^m\left(Y_i^{\tt lab} - \widehat{Y}_i^{\tt lab} \right) / m$ and $\sum_{i=1}^M\widehat{Y}_i^{\tt unlab} / M$ is zero. Therefore, we have
\begin{align*}
\Var\left(\widehat{\theta}_{\mathrm{PP}}\right)
&=
\Var\left[
\frac{1}{m}\sum_{i=1}^m
\left(
Y_i^{\tt lab}
-
\widehat{Y}_i^{\tt lab}
\right)
+
\frac{1}{M}\sum_{i=1}^M
\widehat{Y}_i^{\tt unlab}
\right] \\
&=
\Var\left[
\frac{1}{m}\sum_{i=1}^m
\left(
Y_i^{\tt lab}
-
\widehat{Y}_i^{\tt lab}
\right)
\right]
+
\Var\left[
\frac{1}{M}\sum_{i=1}^M
\widehat{Y}_i^{\tt unlab}
\right] \\
&=
\frac{1}{m}
\Var\left(
Y_i^{\tt lab}
-
\widehat{Y}_i^{\tt lab}
\right)
+
\frac{1}{M}
\Var\left(
\widehat{Y}_i^{\tt unlab}
\right) \\
&= \frac{1}{m} \Var \left (Y_i \right ) + \frac{1}{m} \Var \left (\widehat{Y}_i \right ) - \frac{2}{m} \Cov \left ( Y_i, \widehat{Y}_i \right ) + \frac{1}{M} \Var \left (\widehat{Y}_i \right ) \\
&= \frac{m + M}{mM} \Var \left (\widehat{Y}_i \right ) + \frac{1}{m} \Var \left (Y_i \right ) - \frac{2}{m} \Cov \left ( Y_i, \widehat{Y}_i \right ).
\end{align*}
To discuss under what circumstance $\widehat{\theta}_{\mathrm{PP}}$ has a smaller variance than the complete-case estimator $\sum^m_{i=1} Y^{\tt lab}_i / m$, we have
\begin{align*}
\Var\left(\widehat{\theta}_{\mathrm{PP}}\right)
<
\Var\left(\widehat{\theta}_{\mathrm{CC}}\right)
&\Longleftrightarrow
\frac{m+M}{mM}\Var(\widehat{Y}_i)
-
\frac{2}{m}\Cov(Y_i,\widehat{Y}_i)
<0 \\
&\Longleftrightarrow
\frac{M + m}{M}\Var(\widehat{Y}_i)
<
2\Cov(Y_i,\widehat{Y}_i) \\
&\Longleftrightarrow
\frac{M + m}{M}\Var(\widehat{Y}_i)
<
2\text{Corr}(Y_i,\widehat{Y}_i)
\sqrt{\Var(Y_i)\Var(\widehat{Y}_i)} \\
&\Longleftrightarrow
\text{Corr}(Y_i,\widehat{Y}_i)
>
\frac{M + m}{2M}
\sqrt{
\frac{\Var(\widehat{Y}_i)}
{\Var(Y_i)}
}.
\end{align*}
The above example indicates that when high-quality predictions (the correlation between the true and the predicted value is greater than the above cutoff) of $Y_i$ are available, the prediction-based inference is able to incorporate this extra information and achieve better efficiency.

\section{Supplementary Theoretical Analysis}

\subsection{Optimality of including all missingness patterns}
\label{sec:optimal_pattern_selection}

In this subsection, we elaborate on why it is optimal to include estimates of zero from all the missingness patterns in $\widehat{\vtheta}_{\text{PS-PPI}}$. Suppose a subset of missingness patterns, denoted by $\cS \subseteq \{1, \dots, K\}$, is chosen to construct estimates of zero; the remaining missingness patterns are discarded. As such, we define the corresponding subset-based estimator
\[
    \widehat{\vtheta}_{\text{PS-PPI}, (\cS)}
    =
    \widehat{\vtheta}_{\text{WCC}}
    -
    \sum_{k \in \cS} \widehat{\Wb}_k \widehat{\vDelta}_k.
\]
Let $\Wb_{\opt,(\cS)}$ denote an optimal collection of tuning matrices for this subset, and let $\Sigma_{\text{PS-PPI}, (\cS)}(\Wb_{\opt,(\cS)})$ denote the corresponding asymptotic variance-covariance matrix. We have the following proposition:
\begin{proposition}
    \label{proposition:optimal_pattern_selection}
    Under the conditions of Theorem~\ref{thrm:general_missing_patterns} and Corollary~\ref{corollary:optimal_w}, for every $\cS \subseteq \{1, \dots, K\}$,
    \[
        \Sigma_{\text{PS-PPI}}(\Wb_\opt)
        \preceq
        \Sigma_{\text{PS-PPI}, (\cS)}(\Wb_{\opt,(\cS)}).
    \]
\end{proposition}
\begin{proof}
    Fix any $\cS \subseteq \{1, \dots, K\}$. By Theorem~\ref{thrm:general_missing_patterns}, $\widehat{\vtheta}_{\text{PS-PPI}, (\cS)}$ is consistent and asymptotically normal with mean $\vtheta^\ast$, so it can serve as the base estimator of another PS-PPI estimator. Specifically, define
    $$\widehat{\vtheta}^\dag_{\text{PS-PPI}, (\cS)} = \widehat{\vtheta}_{\text{PS-PPI}, (\cS)} - \sum_{k \notin \cS} \widehat{\Wb}_k \widehat{\vDelta}_k.$$
    By Corollary~\ref{corollary:optimal_w}, there exists a $\Wb_\opt$ for $\widehat{\vtheta}^\dag_{\text{PS-PPI}, (\cS)}$ such that its asymptotic variance-covariance matrix is smaller than or equal to $\Sigma_{\text{PS-PPI}, (\cS)}(\Wb_{\opt, (\cS)})$. On the other hand, we can rewrite
    $$\widehat{\vtheta}^\dag_{\text{PS-PPI}, (\cS)} = \widehat{\vtheta}_{\text{WCC}} - \sum_{k \in \cS} \widehat{\Wb}_k \widehat{\vDelta}_k - \sum_{k \notin \cS} \widehat{\Wb}_k \widehat{\vDelta}_k.$$
    Since $\Wb_\opt$ is globally optimal for $\widehat{\vtheta}_{\text{PS-PPI}}$, the asymptotic variance-covariance matrix of $\widehat{\vtheta}_{\text{PS-PPI}}$ is smaller than or equal to that of $\widehat{\vtheta}^\dag_{\text{PS-PPI}, (\cS)}$. Combining the two inequalities gives $\Sigma_{\text{PS-PPI}}(\Wb_\opt) \preceq \Sigma_{\text{PS-PPI}, (\cS)}(\Wb_{\opt, (\cS)})$.
\end{proof}
The inequality may be an equality, so an added pattern need not yield a strict asymptotic efficiency gain. Proposition~\ref{proposition:optimal_pattern_selection} provides an alternative explanation of the efficiency guarantee of $\widehat{\vtheta}_{\psppi}$ over $\widehat{\vtheta}_\wcc$ stated in Corollary~\ref{corollary:optimal_w}. In particular, by choosing $\cS = \varnothing$, PS-PPI is reduced to weighted complete-case analysis, with $\Sigma_{\text{PS-PPI}, (\varnothing)}(\Wb_{\opt, (\varnothing)}) = \Sigma_\vtheta$, and $\Sigma_{\text{PS-PPI}}(\Wb_{\opt}) \preceq \Sigma_\vtheta$ holds because $\cS = \varnothing$ is also a subset of $\{1, \dots, K\}$.

More broadly, Proposition~\ref{proposition:optimal_pattern_selection} implies that, with globally optimal tuning matrices $\Wb_k$, any available estimate of zero can be incorporated directly into the PS-PPI estimator without degrading its asymptotic efficiency. One application is the setting in which multiple predictive functions are available, which we discuss in Section~\ref{sec:multiple_predictive_functions}.

\subsection{PS-PPI with multiple predictive functions}
\label{sec:multiple_predictive_functions}

In Section~\ref{sec:methodology}, we introduced the PS-PPI framework for the case in which only a single predictive function $h_k$ is available for imputation under the $k$-th missingness pattern. In practice, however, one may have access to multiple predictive functions, which may be trained on different resources or have different model structures. Recent work has investigated this setting in more specialized applications. In particular,~\citet{cowen2026multiple} studied how to obtain efficient and consistent estimates for AI model evaluation under a budget constraint when evaluation tools of varying quality and cost are available, and~\citet{brawand2026active} considered a clinical post-deployment monitoring setting in which human experts are expensive and multiple clinical evaluation pipelines are available.

In this subsection, we propose a framework for data with general missingness when multiple predictive functions are available for imputation. Formally, let $\cH_k = \{h_{k1}, \dots, h_{kn_k}\}$ be the set of $n_k$ predictive functions that can be used to impute the missing values under the $k$-th missingness pattern; the number $n_k$ may differ across patterns. As in Section~\ref{sec:preliminary}, we denote the imputed random vector of the $i$-th subject under the $k$-th missingness pattern and the predictive function $h_{kj}$ by $\widetilde{\vY}_{kj, i} = \{G_k(\vY_i), h_{kj}(G_k(\vY_i))\}$. In this setting, a natural question is which predictive function in each $\cH_k$ is optimal in terms of variance reduction for the PS-PPI estimator.

Rather than selecting one predictive function, we construct an estimate of zero from each predictive function available under the $k$-th missingness pattern and use their weighted sum to augment the base estimator. We use the same base estimator and missing data model as in Section~\ref{sec:methodology}. Let $\widehat{\vDelta}_{kj} = \widehat{\vgamma}_{\cc, kj} - \widehat{\vgamma}_{\ic, kj}$ denote the estimate of zero for the $k$-th missingness pattern and the $j$-th predictive function in $\cH_k$. As in Section~\ref{sec:methodology}, $\widehat{\vgamma}_{\cc, kj}$ is obtained by solving the following weighted estimating equations:
\begin{equation}
\label{eq:general_missing_pattern_ee_gamma_1_multiple}
    \sum_{i = 1}^{N} \vpsi^{(kj)}_{\cc, i}(\vgamma_{\cc, kj}, \widehat{\vomega}) = \sum^{N}_{i= 1}  \frac{I(\cC_i = \infty)}{\pi_\infty(\vY_{\obs, i}; \widehat{\vomega})}  \vpsi \left ( \vYimp_{kj, i}; \vgamma_{\cc, kj} \right ) = \vzero;
\end{equation}
$\widehat{\vgamma}_{\ic,kj}$, on the other hand, is obtained by solving the following weighted estimating equations:
\begin{equation}
\label{eq:general_missing_pattern_ee_gamma_2_multiple}
    \sum_{i = 1}^{N} \vpsi^{(kj)}_{\ic,i}(\vgamma_{\ic, kj}, \widehat{\vomega}) = \sum^{N}_{i=1} \frac{I(\cC_i = k)}{\pi_k(G_{k}(\boldsymbol{Y}_i); \widehat{\vomega})} \vpsi \left ( \vYimp_{kj, i}; \vgamma_{\ic, kj} \right ) = \vzero.
\end{equation}
The resulting estimator, which we call MPS-PPI, is
\begin{equation}
\label{eq:general_missing_pattern_estimator_multiple}
    \widehat{\vtheta}_\mpsppi =  \widehat{\vtheta}_{\text{WCC}} - \sum^K_{k=1} \sum^{n_k}_{j = 1}  \widehat{\Wb}_{kj}\widehat{\vDelta}_{kj},
\end{equation}
which has the same form as Equation~(\ref{eq:general_missing_pattern_estimator}) after stacking all predictor-specific estimates of zero. Theorem~\ref{thrm:general_missing_pattern_estimator_multiple} extends Theorem~\ref{thrm:general_missing_patterns} and Corollary~\ref{corollary:optimal_w} to establish asymptotic normality and characterize the optimal tuning matrices in this setting.
\begin{theorem}[Asymptotic normality and optimality of $\widehat{\vtheta}_\mpsppi$]
    \label{thrm:general_missing_pattern_estimator_multiple}
    Let $M=\sum_{k=1}^K n_k$, and stack the predictor-specific estimates of zero and tuning matrices as
    \[
    \widehat{\vDelta}_{\mathrm{M}}
    =
    \left(
    \widehat{\vDelta}_{11}^{\top},\ldots,
    \widehat{\vDelta}_{1n_1}^{\top},\ldots,
    \widehat{\vDelta}_{K1}^{\top},\ldots,
    \widehat{\vDelta}_{Kn_K}^{\top}
    \right)^{\top}
    \in\mathbb R^{Md}
    \]
    and $\widehat{\Wb}_{\mathrm{M}}=(\widehat{\Wb}_{11},\ldots,\widehat{\Wb}_{1n_1},\ldots,\widehat{\Wb}_{K1},\ldots,\widehat{\Wb}_{Kn_K})\in\mathbb R^{d\times Md}$, respectively. Suppose the missing data model is correctly specified and estimated under the MAR assumption, and $\widehat{\Wb}_{kj}\overset{p}{\rightarrow}\Wb_{kj}$ for every pair $(k,j)$; let $\Wb_{\mathrm{M}}$ stack the $\Wb_{kj}$ in the same order as $\widehat{\Wb}_{\mathrm{M}}$. Then $\widehat{\vtheta}_\mpsppi$ is consistent for $\vtheta^\ast$, and
    \[
    \sqrt{N}\left(\widehat{\vtheta}_\mpsppi-\vtheta^\ast\right)
    \overset{d}{\longrightarrow}
    \mathcal N\left\{\vzero,\Sigma_{\mathrm{MPS\text{-}PPI}}(\Wb_{\mathrm{M}})\right\},
    \]
    where
    \begin{align}
    \Sigma_{\mathrm{MPS\text{-}PPI}}(\Wb_{\mathrm{M}})
    ={}&
    \Sigma_{\vtheta}
    -\Sigma_{\vtheta,\vDelta_{\mathrm{M}}}\Wb_{\mathrm{M}}^{\top}
    -\Wb_{\mathrm{M}}\Sigma_{\vtheta,\vDelta_{\mathrm{M}}}^{\top}
    +\Wb_{\mathrm{M}}\Sigma_{\vDelta_{\mathrm{M}}}\Wb_{\mathrm{M}}^{\top},
    \label{eq:mpsppi-asymptotic-variance}
    \end{align}
    $\Sigma_{\vtheta}$ is defined in Proposition~\ref{proposition:wcc}, and
    \[
    \Sigma_{\vtheta,\vDelta_{\mathrm{M}}}
    =\E\left[\varphi_{\infty,i}\varphi_{\vDelta_{\mathrm{M}},i}^{\top}\right],
    \qquad
    \Sigma_{\vDelta_{\mathrm{M}}}
    =\E\left[\varphi_{\vDelta_{\mathrm{M}},i}\varphi_{\vDelta_{\mathrm{M}},i}^{\top}\right].
    \]
    Here, $\varphi_{\vDelta_{\mathrm{M}},i}$ stacks the influence functions $\varphi_{kj,i}=\varphi_{\cc,i}^{(kj)}-\varphi_{\ic,i}^{(kj)}$, the $i$-th influence function of $\widehat{\vDelta}_{kj}$, in the same order as $\widehat{\vDelta}_{\mathrm{M}}$. The globally optimal $\Wb_{\mathrm{M}}$ is
    \[
    \Wb_{\mathrm{M},\opt}
    =\Sigma_{\vtheta,\vDelta_{\mathrm{M}}}\Sigma_{\vDelta_{\mathrm{M}}}^{-1},
    \]
    and the corresponding asymptotic variance-covariance matrix is
    \[
    \Sigma_{\mathrm{MPS\text{-}PPI}}(\Wb_{\mathrm{M},\opt})
    =
    \Sigma_{\vtheta}
    -\Sigma_{\vtheta,\vDelta_{\mathrm{M}}}
    \Sigma_{\vDelta_{\mathrm{M}}}^{-1}
    \Sigma_{\vtheta,\vDelta_{\mathrm{M}}}^{\top}
    \preceq \Sigma_{\vtheta}.
    \]
    Moreover, $\Sigma_{\mathrm{MPS\text{-}PPI}}(\Wb_{\mathrm{M},\opt})\preceq
    \Sigma_{\mathrm{MPS\text{-}PPI}}(\Wb_{\mathrm{M}})$ for every $\Wb_{\mathrm{M}}\in\mathbb R^{d\times Md}$.
\end{theorem}
The proof of Theorem~\ref{thrm:general_missing_pattern_estimator_multiple} is similar to those of Theorem~\ref{thrm:general_missing_patterns} and Corollary~\ref{corollary:optimal_w}. By incorporating the estimates of zero from all predictive functions, $\widehat{\vtheta}_\mpsppi$ in Equation (\ref{eq:general_missing_pattern_estimator_multiple}) avoids the need to choose a single optimal predictive function for imputation; instead, it dynamically weights the predictive functions through their estimates of zero. In particular, by Section~\ref{sec:optimal_pattern_selection}, with the optimal tuning matrices, adding an estimate of zero does not increase the asymptotic variance, and discarding any estimate of zero does not improve efficiency.

We evaluate the performance of $\widehat{\vtheta}_{\text{MPS-PPI}}$ by replicating the real-data analysis in Section~\ref{sec:data_example}. Recall that, for PS-PPI, both random forest and XGBoost are fitted to predict CRP and race, and the best-performing model is selected for imputation. For MPS-PPI, we instead use both models to construct multiple predictive functions for the pattern in which only CRP or race is missing. For the pattern where both race and CRP are missing, we consider all combinations of random forest and XGBoost for the two variables, which yields four predictive functions. Table~\ref{tab:mpsppi_data_example_results} presents the results of MPS-PPI, with WCCA and PS-PPI as baselines. Both PS-PPI and MPS-PPI result in a significant negative adjusted association between education level and CRP level. Furthermore, by incorporating information from multiple predictive functions, MPS-PPI is more efficient: its variance is only 76.9\% of that of WCCA, compared with 77.4\% for PS-PPI.

\begin{table}[htbp]
\centering
\caption{\textit{All of Us} coefficient estimates and 95\% confidence intervals for education level obtained using WCCA, PS-PPI, and MPS-PPI. Variance ratios are relative to WCCA; smaller values indicate greater efficiency.}
\label{tab:mpsppi_data_example_results}
\resizebox{0.8\linewidth}{!}{%
\begin{tabular}{lcccc}
\toprule
Method & $\widehat{\beta}_{\textrm{edu}}$ & Standard Error & 95\% CIs & Variance Ratio \\
\midrule
WCCA (baseline) & -1.162 & 0.0508 & $[-1.262, -1.063]$ & 1 \\
PS-PPI & -1.004 & 0.0447 & $[-1.091, -0.916]$ & 0.774 \\
MPS-PPI & -0.978 & 0.0445 & $[-1.065, -0.891]$ & \textbf{0.769} \\
\bottomrule
\end{tabular}}
\end{table}

Although $\widehat{\vtheta}_\mpsppi$ is asymptotically valid and allows dynamic selection when there are multiple sets of predictions available for imputation, empirically its performance may be constrained by the following two factors. First, when $\sum^K_{k=1} |\cH_k|$ is large, the above procedure may suffer from a computational burden, as a large number of estimates of zero must be incorporated in the procedure. To resolve this, one may have to combine the predictive models at the prediction level such that only one prediction is retained before imputation, i.e., apply model ensembling techniques. Determining how to ensemble the predictions to maximize the efficiency of the resulting estimator remains a challenge, and we leave it for future work. Second, when the use of predictive models is subject to a budget constraint, the analyst may need to select a subset of the predictive functions in each $\cH_k$ before applying MPS-PPI. We refer readers to~\citet{brawand2026active,cowen2026multiple} for detailed discussions of prediction-based inference in budget-constrained settings.

\subsection{Asymptotic equivalence between PS-PPI and~\citet{chen2000unified}}
\label{sec:asymptotic_equivalence_chen}

In this subsection, we elaborate on the asymptotic equivalence between $\widehat{\vtheta}_{\text{PS-PPI}}$ and $\widehat{\vtheta}_{\text{Chen}}$ using an example of perfect predictions, i.e., the predictions are exactly the true values. Specifically, we have the following remark:
\begin{remark}
    \label{remark:perfect_predictions}
    If perfect predictions are available, then for $\widehat{\vtheta}_{\text{Chen}}$, $\widehat{\Wb}_{\text{Chen}} = \mI$; for $\widehat{\vtheta}_{\text{PS-PPI}}$, $\widehat{\Wb} = \widehat{\pi}_1 \mI$.
\end{remark}
Remark~\ref{remark:perfect_predictions} can be proved by deriving the closed-form of the variance-covariance matrices under linear regression and obtaining the closed form of $\Sigma_{\text{PS-PPI}}$. From Remark~\ref{remark:perfect_predictions}, when perfect predictions are available, $\widehat{\vtheta}_{\text{Chen}} = \widehat{\vgamma}_{\text{all}}$, which is equivalent to the results from the underlying full data since the predictions are perfect. On the other hand, $\widehat{\vtheta}_{\text{PS-PPI}} = \widehat{\vtheta}_{\text{CC}} - \widehat{\pi}_1 (\widehat{\vgamma}_{\cc, 1} - \widehat{\vgamma}_{\ic, 1}) = \widehat{\pi}_\infty \widehat{\vtheta}^{(\infty)}_{\text{CC}} + \widehat{\pi}_1 \widehat{\vtheta}^{(1)}_{\text{CC}}$. Here $\widehat{\vtheta}^{(\infty)}_{\text{CC}}$ is the CCA estimate obtained from the complete-case subset; $\widehat{\vtheta}^{(1)}_{\text{CC}}$ is the CCA estimate obtained from the underlying full data of the missing-outcome subset of data. We are able to replace $\widehat{\vgamma}_{\ic, 1}$ with $\widehat{\vtheta}^{(1)}_{\text{CC}}$ because the perfect predictions make the imputed dataset exactly the same as the underlying full data. As a result, $\widehat{\vtheta}_{\text{PS-PPI}}$ can be viewed as a weighted sum of the full data estimators from the complete-case and the missing-outcome subsets of data, respectively; the weights are the proportions of each subset: $\widehat{\pi}_\infty$ and $\widehat{\pi}_1$. Therefore, $\widehat{\vtheta}_{\text{PS-PPI}}$ can be viewed as an inverse variance weighted version of $\widehat{\vtheta}_{\text{Chen}}$.

\subsection{Double robustness properties of PS-PPI}
\label{sec:psppi-doubly-robust}

In Theorem~\ref{thrm:general_missing_patterns}, it is assumed that the missing data model is either known or correctly specified when unknown. This assumption is fundamental in the sense that it allows PS-PPI to impute the missing entries with assumption-free predictive models. In semiparametric theory, a desirable property is double robustness, i.e., valid inference when either the missing data model or the predictive model is correctly specified. Double robustness is closely related to semiparametric efficiency, as locally semiparametric efficient estimators in missing data problems are typically constructed to be doubly robust~\citep{tsiatis2006semiparametric}. As illustrated empirically in Section~\ref{sec:misspecification_robustness}, PS-PPI is not doubly robust. It is therefore of interest to investigate whether PS-PPI can be modified such that it is doubly robust. Under general missingness patterns and the MAR assumption, the efficient influence function, from which doubly robust estimators are constructed, has no closed form~\citep{tsiatis2006semiparametric}. To address this issue,~\citet{chaudhuri2016gmm} introduce a different MAR assumption and derive closed-form expressions for the efficient influence function and the corresponding semiparametric efficient estimator in generalized method of moments (GMM) estimation. Specifically,~\citet{chaudhuri2016gmm} assume that, for all $k \in \{1, \dots, K, \infty\}$,
\begin{equation}
    \label{eq:chaudhuri-mar}
    P(\cC_i = k | \vY_i) = P(\cC_i = k | \vY_{c, i}).
\end{equation}
Here, $\vY_{c, i}$ is a set of variables in $\vY_i$ that are always observed, i.e., $\vY_{c, i} \subseteq G_k(\vY_i)$ for every $k$. In real-data applications, $\vY_{c, i}$ may consist of type (c) variables or type (b) variables that are not subject to missingness; type (b) and (c) variables are defined in Section~\ref{sec:preliminary}. Following~\citet{chaudhuri2016gmm}, we adopt Assumption~\ref{assumption:missing-data-dr} on the missing data mechanism:
\begin{assumption}
        \label{assumption:missing-data-dr}
        The following assumption is imposed on the missing data:
        \begin{enumerate}
            \item \textbf{Missingness pattern:} Components of $\vY_i$ may be missing with arbitrary missingness patterns. The missingness index $\cC_i \in \{1, \dots, K, \infty\}$ denotes the pattern of the $i$-th subject, where $\cC_i = \infty$ indicates that the subject is fully observed. In addition, a set of variables $\vY_{c,i}$ is observed for every subject, such that $\vY_{c,i} \subseteq G_k(\vY_i)$ for every $k \in \{1, \dots, K, \infty\}$.
            \item \textbf{Missing data mechanism:} The data are missing at random (MAR), with the missingness depending on $\vY_i$ only through the always-observed variables $\vY_{c,i}$, i.e., $\Pr(\cC_i = k \mid \vY_i) = \Pr(\cC_i = k \mid \vY_{c,i})$ for every $k \in \{1, \dots, K, \infty\}$. In addition, we impose a positivity assumption such that $\Pr(\cC_i = \infty \mid \vY_{c,i})$ is bounded away from zero.
        \end{enumerate}
    \end{assumption}
The main difference between Assumption~\ref{assumption:missing-data-dr} and Assumption~\ref{assumption:missing-data} is that, under Assumption~\ref{assumption:missing-data-dr}, the conditional probability of the missingness index $\cC_i$ depends only on the always-observed variables $\vY_{c, i}$, whereas under Assumption~\ref{assumption:missing-data}, it may depend on all the observed components $G_k(\vY_i)$ in the $k$-th pattern.

Under Assumption~\ref{assumption:missing-data-dr}, we show that $\widehat{\vtheta}_\psppi$ can be made doubly robust by introducing a correction term. Specifically, suppose that we use an estimated working missing data model $\wt{\bpi}(\vY_c; \widehat{\vnu})$, where $\widehat{\vnu}$ solves the following estimating equation:
\begin{equation}
    \label{eq:missing_data_model_ee_misspecified}
    \sum^N_{i=1}\wt{\vphi}_i(\vnu) = \sum^N_{i=1}\wt{\vphi}(\vY_{c, i}; \vnu) = \vzero.
\end{equation}

We denote the root of the population-level estimating equation in Equation (\ref{eq:missing_data_model_ee_misspecified}) as $\vnu^\ast$, and write $\wt{\bpi}(\vY_{c, i}) = \left \{ \wt{\pi}_\infty(\vY_{c, i}), \wt{\pi}_1(\vY_{c, i}), \cdots, \wt{\pi}_K(\vY_{c, i}) \right \}$ for the working missing data model evaluated at $\vnu^\ast$. The working model is correctly specified if $\wt{\bpi}(\vY_{c, i}) = \bpi(\vY_{c, i}; \vomega^\ast)$ almost surely.

With the estimated working missing data model, $\widehat{\vtheta}_\wcc$, $\widehat{\vgamma}_{\cc, k}$, and $\widehat{\vgamma}_{\ic, k}$ solve the following estimating equations, respectively:
\begin{equation}
\label{eq:general_missing_pattern_ee_misspecified}
    \sum_{i = 1}^{N} \wt{\vpsi}_{\text{WCC}, i}(\vtheta, \widehat{\vnu}) = \sum_{i = 1}^{N} \frac{I(\cC_i = \infty)}{\wt{\pi}_\infty(\vY_{c, i}; \widehat{\vnu})} \vpsi \left ( \vY_i; \vtheta \right ) = \vzero.
\end{equation}
\begin{equation}
\label{eq:general_missing_pattern_ee_gamma_1_misspecified}
    \sum_{i = 1}^{N} \wt{\vpsi}^{(k)}_{\cc, i}(\vgamma_{\cc, k}, \widehat{\vnu}) = \sum^{N}_{i = 1} \frac{I(\cC_i = \infty)}{\wt{\pi}_\infty(\vY_{c, i}; \widehat{\vnu})} \vpsi \left ( \vYimp_{k, i}; \vgamma_{\cc, k} \right ) = \vzero.
\end{equation}
\begin{equation}
\label{eq:general_missing_pattern_ee_gamma_2_misspecified}
    \sum_{i = 1}^{N} \wt{\vpsi}^{(k)}_{\ic, i}(\vgamma_{\ic, k}, \widehat{\vnu}) = \sum^{N}_{i = 1} \frac{I(\cC_i = k)}{\wt{\pi}_k(\vY_{c, i}; \widehat{\vnu})} \vpsi \left ( \vYimp_{k, i}; \vgamma_{\ic, k} \right ) = \vzero.
\end{equation}
We make the following assumption on the roots of the population-level estimating equations in Equations (\ref{eq:general_missing_pattern_ee_misspecified}), (\ref{eq:general_missing_pattern_ee_gamma_1_misspecified}), and (\ref{eq:general_missing_pattern_ee_gamma_2_misspecified}).
\begin{assumption}
    \label{assumption:misspecified-population-level-parameters}
    Unique roots of $\E \left [ \wt{\vpsi}_{\text{WCC}, i}(\vtheta, \vnu) \right ]$, $\E \left [ \wt{\vpsi}^{(k)}_{\cc, i}(\vgamma_{\cc, k}, \vnu) \right ]$, and $\E \left [ \wt{\vpsi}^{(k)}_{\ic, i}(\vgamma_{\ic, k}, \vnu) \right ]$ exist; we denote them as $(\wt{\vtheta}^\ast, \vnu^\ast)$, $(\wt{\vgamma}^\ast_{\cc, k}, \vnu^\ast)$, and $(\wt{\vgamma}^\ast_{\ic, k}, \vnu^\ast)$, respectively.
\end{assumption}

With the above notation of the missing data model, we propose the following variant of $\widehat{\vtheta}_\psppi$, which we later prove to be doubly robust:
\begin{align}
    \label{eq:psppi_dr}
    \widehat{\vtheta}_{\psppi\text{(DR)}} &= \widehat{\vtheta}_{\psppi} - \left ( \widehat{\vzeta}_\cc - \widehat{\vzeta} \right ) \notag \\
    &= \widehat{\vtheta}_\wcc - \sum^K_{k=1} \widehat{\Wb}_k (\widehat{\vgamma}_{\cc, k} - \widehat{\vgamma}_{\ic, k}) - \left ( \widehat{\vzeta}_\cc - \widehat{\vzeta} \right ).
\end{align}
Here, $\vq(\vY_c, \vzeta): \cY_c \times \R^d \mapsto \R^d$ is a user-chosen working model for
$$\E \left \{ \vpsi(\vY_i; \vzeta) \mid \vY_{c, i} \right \},$$
where $\cY_c$ is the sample space of $\vY_{c, i}$, and $\vpsi$ is the estimating equation of scientific interest introduced in Section~\ref{sec:preliminary}. The estimators $\widehat{\vzeta}_\cc$ and $\widehat{\vzeta}$ solve the following two estimating equations, respectively:
\begin{equation}
    \label{eq:dr-correction-ee}
    \sum^N_{i=1} \frac{I(\cC_i = \infty)}{\wt{\pi}_\infty(\vY_{c, i}; \widehat{\vnu})} \vq(\vY_{c, i}; \vzeta_\cc) = \vzero, \quad \sum^N_{i=1} \vq(\vY_{c, i}; \vzeta) = \vzero.
\end{equation}
Note that the second estimating equation in Equation (\ref{eq:dr-correction-ee}) is solvable because $\vY_{c, i}$ is fully observed. We make the following assumption in terms of the roots of the two population-level estimating equations in Equation (\ref{eq:dr-correction-ee}).
\begin{assumption}
    \label{assumption:misspecified-population-level-parameters-q-model}
    Unique roots of the two population-level estimating equations in Equation (\ref{eq:dr-correction-ee}) exist and we denote them as $(\vzeta^\ast_\cc, \vnu^\ast)$ and $(\vzeta^\ast, \vnu^\ast)$, respectively.
\end{assumption}

When the missing data model is correctly specified, $\widehat{\vzeta}_\cc - \widehat{\vzeta}$ is an additional estimate of zero in the sense of Section~\ref{sec:methodology}. Specifically, the roots of the two estimating equations in Equation (\ref{eq:dr-correction-ee}) converge to the same population value, so that $\widehat{\vzeta}_\cc - \widehat{\vzeta} \overset{p}{\rightarrow} \vzero$ and the validity of $\widehat{\vtheta}_{\psppi\text{(DR)}}$ is preserved. When the missing data model is unknown and misspecified, $\widehat{\vtheta}_{\psppi\text{(DR)}}$ is doubly robust in the sense that inference based on it can still be valid, provided that $\vq(\vY_{c, i}; \cdot)$ and each predictive function $h_k$, $k = 1, \cdots, K$, satisfy certain conditions. We formally state this double robustness property in Theorem~\ref{thrm:psppi-dr}.
\begin{theorem}[Double robustness of $\widehat{\vtheta}_{\psppi\text{(DR)}}$]
\label{thrm:psppi-dr}
Under the MAR assumption in Equation (\ref{eq:chaudhuri-mar}), suppose that for each $k$, $\widehat{\Wb}_k \overset{p}{\rightarrow} \Wb_k$, and Assumptions~\ref{assumption:gamma_existence},~\ref{assumption:main},~\ref{assumption:misspecified-population-level-parameters}, and~\ref{assumption:misspecified-population-level-parameters-q-model}
hold. Then $\widehat{\vtheta}_{\psppi\text{(DR)}}$ is consistent and asymptotically normal with mean $\vtheta^\ast$ if either
\begin{itemize}
    \item[(a)] the missing data model is correctly specified, while $\vq$ and $h_1, \cdots, h_K$ are arbitrary; or
    \item[(b)] the missing data model $\wt{\bpi}$ is misspecified, while the following two conditions are satisfied:
    \begin{itemize}
        \item[(b.1)] For every $\vzeta \in \R^d$, the working model satisfies $\vq(\vY_{c,i}; \vzeta) = \E \left \{ \vpsi(\vY_i; \vzeta) \mid \vY_{c,i} \right \}$ almost surely.
        \item[(b.2)] For each $k = 1, \cdots, K$, the predictive function $h_k$ satisfies $\E \left \{ \vpsi \left ( \vYimp_{k, i}; \wt{\vgamma}^\ast_k \right ) \mid \vY_{c, i} \right \} = \vzero$ almost surely for some $\wt{\vgamma}^\ast_k \in \R^d$. Recall that $\vYimp_{k, i} = \left \{ G_k(\vY_i), h_k(G_k(\vY_i)) \right \}$.
    \end{itemize}
\end{itemize}
\end{theorem}
We now prove Theorem~\ref{thrm:psppi-dr}.
\begin{proof}
For $k \in \{1, \cdots, K, \infty\}$, let $a_k(\vY_{c, i}) = \pi_k(\vY_{c, i}; \vomega^\ast) / \wt{\pi}_k(\vY_{c, i})$.
Under Assumption~\ref{assumption:missing-data-dr}, for each $k$ we have
\begin{equation}
    \label{eq:dr-key-identity}
    \E \left [ \frac{I(\cC_i = k)}{\wt{\pi}_k(\vY_{c, i})} \,\middle|\, \vY_i \right ] = a_k(\vY_{c, i}).
\end{equation}
\textbf{Case (a)} When the missing data model is correctly specified, $a_k(\vY_{c, i}) = 1$ for each $k = 1, \cdots, K, \infty$. In this case, $\wt{\vtheta}^\ast = \vtheta^\ast$ because
$$\E_{(\cC, \vY)} \left [ \frac{I(\cC_i = \infty)}{\wt{\pi}_\infty(\vY_{c, i}; \vnu^\ast)} \vpsi \left ( \vY_i; \vtheta^\ast \right ) \right ] = \E_{\vY} \left [ \vpsi \left ( \vY_i; \vtheta^\ast \right )  \right ] = \vzero.$$
A similar derivation gives $\wt{\vgamma}^\ast_{\cc, k} = \wt{\vgamma}^\ast_{\ic, k} = \vgamma^\ast_k$ and $\vzeta^\ast_\cc = \vzeta^\ast$. With a derivation similar to the proof of Theorem~\ref{thrm:general_missing_patterns}, $\widehat{\vtheta}_{\psppi\text{(DR)}}$ is consistent and asymptotically normal with mean $\vtheta^\ast$.

\textbf{Case (b)} When the missing data model is misspecified, we first show that 1) $\wt{\vtheta}^\ast - \left ( \vzeta^\ast_{\cc} - \vzeta^\ast \right) = \vtheta^\ast$ and 2) $\wt{\vgamma}^\ast_{\cc, k} = \wt{\vgamma}^\ast_{\ic, k}$ still holds for each $k$. For 1), by Assumptions~\ref{assumption:misspecified-population-level-parameters} and~\ref{assumption:misspecified-population-level-parameters-q-model} and condition (b.1),
\begin{align}
    \E_{(\cC, \vY)} \left [ \frac{I(\cC_i = \infty)}{\wt{\pi}_\infty(\vY_{c, i}; \vnu^\ast)} \vpsi(\vY_i; \wt{\vtheta}^\ast) \right ] &= \E_{\vY} \left [ a_\infty(\vY_{c, i}) \vpsi(\vY_i; \wt{\vtheta}^\ast) \right ] \notag \\
    &= \E_{\vY_c} \left [ a_\infty(\vY_{c, i}) \E_{\vY } \left [ \vpsi(\vY_i; \wt{\vtheta}^\ast) \mid \vY_{c, i} \right ] \right ] \notag \\
    &= \E_{\vY_c} \left [ a_\infty(\vY_{c, i}) \vq(\vY_{c, i}; \wt{\vtheta}^\ast) \right ] = \vzero_d, \label{eq:misspecified-population-level-parameters}
\end{align}
where the third equality applies condition (b.1), and
\begin{align}
    \label{eq:misspecified-population-level-parameters-q-model-1}
    \E_{(\cC, \vY)} \left [ \frac{I(\cC_i = \infty)}{\wt{\pi}_\infty(\vY_{c, i}; \vnu^\ast)} \vq(\vY_{c, i}; \vzeta^\ast_\cc) \right ] &= \E_{\vY_c} \left [ a_\infty(\vY_{c, i}) \vq(\vY_{c, i}; \vzeta^\ast_\cc) \right ] = \vzero_d.
\end{align}
Since $(\vzeta^\ast_\cc, \vnu^\ast)$ is the unique root of the population-level estimating equation, we have $\wt{\vtheta}^\ast = \vzeta^\ast_\cc$. On the other hand, applying condition (b.1) in the second equality gives
\begin{align}
    \E_{(\cC, \vY)} \left [ \vq(\vY_{c, i}; \vzeta^\ast) \right ] &= \E_{\vY_c} \left [ \vq(\vY_{c, i}; \vzeta^\ast) \right ] \notag \\
    &= \E_{\vY_c} \left [ \E_{\vY} \left \{ \vpsi(\vY_i; \vzeta^\ast)  \mid \vY_c \right \} \right ] \notag \\
    &= \E_{\vY} \left [ \vpsi(\vY_i; \vzeta^\ast ) \right ] = \vzero. \label{eq:misspecified-population-level-parameters-q-model-2}
\end{align}
Recall from Section~\ref{sec:preliminary} that $\vtheta^\ast$ is the unique solution of Equation (\ref{eq:misspecified-population-level-parameters-q-model-2}), so $\vtheta^\ast = \vzeta^\ast$. Therefore, $\wt{\vtheta}^\ast - \left ( \vzeta^\ast_{\cc} - \vzeta^\ast \right) = \vtheta^\ast$, which proves 1). For 2), by Assumption~\ref{assumption:misspecified-population-level-parameters} and condition (b.2),
\begin{align}
    \E_{(\cC, \vY)} \left [ \frac{I(\cC_i = \infty)}{\wt{\pi}_\infty(\vY_{c, i}; \vnu^\ast)} \vpsi \left ( \vYimp_{k, i}; \wt{\vgamma}^\ast_k \right ) \right ] &= \E_{\vY} \left [ a_\infty(\vY_{c, i}) \vpsi \left ( \vYimp_{k, i}; \wt{\vgamma}^\ast_k \right ) \right ] \notag \\
    &= \E_{\vY_c} \left [a_\infty(\vY_{c, i}) \E_{\vY | \vY_c} \left [ \vpsi \left ( \vYimp_{k, i}; \wt{\vgamma}^\ast_k \right ) \right ] \right ] = \vzero_d
\end{align}
and
\begin{align}
    \E_{(\cC, \vY)} \left [ \frac{I(\cC_i = k)}{\wt{\pi}_k(\vY_{c, i}; \vnu^\ast)} \vpsi \left ( \vYimp_{k, i}; \wt{\vgamma}^\ast_k \right ) \right ] &= \E_{\vY} \left [ a_k(\vY_{c, i}) \vpsi \left ( \vYimp_{k, i}; \wt{\vgamma}^\ast_k \right ) \right ] \notag \\
    &= \E_{\vY_c} \left [a_k(\vY_{c, i}) \E_{\vY | \vY_c} \left [ \vpsi \left ( \vYimp_{k, i}; \wt{\vgamma}^\ast_k \right ) \right ] \right ] = \vzero_d.
\end{align}
Because $\E \left [ \wt{\vpsi}^{(k)}_{\cc, i}(\vgamma_{\cc, k}, \vnu) \right ]$ and $\E \left [ \wt{\vpsi}^{(k)}_{\ic, i}(\vgamma_{\ic, k}, \vnu) \right ]$ have unique roots, we have $\wt{\vgamma}^\ast_{\cc, k} = \wt{\vgamma}^\ast_{\ic, k} = \wt{\vgamma}^\ast_k$. In Assumption~\ref{assumption:gamma_existence}, where the missing data model is correctly specified, $\wt{\vgamma}^\ast_k$ is denoted simply by $\vgamma^\ast_k$. Therefore, with a similar derivation as in the proof of Theorem~\ref{thrm:general_missing_patterns}, $\widehat{\vtheta}_{\psppi\text{(DR)}}$ is consistent and asymptotically normal with mean $\vtheta^\ast$.
\end{proof}
Hence, to achieve double robustness, users need not only predictive functions that satisfy condition (b.2) but also a correctly specified $\vq$, such that $\vq(\vY_{c, i}; \vzeta) = \E \left \{ \vpsi(\vY_i; \vzeta) \mid \vY_{c, i} \right \}$ for every $\vzeta$.

A numerical simulation is performed to verify the double robustness property proved in Theorem~\ref{thrm:psppi-dr}. To construct correctly specified $\vq$ and predictive models $h_k$, in the simulation we consider the following data generating process. In particular, we independently generate $Z_i, U_{i1}, U_{i2} \sim \cN(0,1)$ and $\epsilon_i \sim \cN(0, 0.8^2)$, and let
\[
    X_{i1}=Z_i + U_{i1},\qquad X_{i2} = Z_i + U_{i2},\qquad
    Y_i=- 1 + 0.3X_{i1} + X_{i2} + 0.6Z_i + \epsilon_i.
\]
To connect with the notation $\vY_i$ introduced in Section~\ref{sec:preliminary}, we have $\vY_i = (Y_i, X_{i1}, X_{i2}, Z_{i})$ in this simulation setting.

The scientific interest lies in the marginal association between $Y$ and $X_1$ and $X_2$. Therefore, an estimating function $\vpsi(\vY_i;\vtheta)=\bD_i(Y_i-\bD_i^\top\vtheta)$ is considered, where $\bD_i=(1,X_{i1},X_{i2})^\top$. By derivation, the underlying true parameter of interest is $\vtheta^\ast=(-1,0.5,1.2)^\top$.

We assume that three missingness patterns are observed in the data: $\cC_i=\infty$ for complete cases, $\cC_i=1$ when $X_{i1}$ is missing, and $\cC_i=2$ when $X_{i2}$ is missing. $Y_i$ and $Z_i$ are always observed. The following missing data model is considered:
\[
    \pi_\infty(Z_i)=\frac{1}{3}+0.2\operatorname{sign}(Z_i),\qquad
    \pi_1(Z_i)=\frac{1}{3}-0.2\operatorname{sign}(Z_i),\qquad
    \pi_2(Z_i)=\frac{1}{3}.
\]
It can be seen that the above missing data model satisfies Equation~(\ref{eq:chaudhuri-mar}), which assumes that the missing data model only depends on always-observed variables. In the simulation, we consider two working missing data models: 1) correctly specified missing data model; 2) misspecified missing data model, in which the missing data model is assumed to be $\pi_\infty(Z_i) = \pi_1(Z_i) = \pi_2(Z_i) = 1/3$, i.e., missing-completely-at-random.

To verify the double robustness property of PS-PPI (DR), we take $\vY_{c, i} = Z_i$. Although $Y_i$ is also always observed, we exclude it from $\vY_{c, i}$ because the correctly specified predictive functions below satisfy (b.2) of Theorem~\ref{thrm:psppi-dr} with $\vY_{c, i} = Z_i$ but not with $\vY_{c, i} = (Y_i, Z_i)$. By derivation, we have the following correctly specified $\vq$ that satisfies (b.1) of Theorem~\ref{thrm:psppi-dr}:
$$\vq(Z_i;\vzeta)=\boldsymbol{b}(Z_i)-\Mb(Z_i)\vzeta,$$
where
\[
\begin{aligned}
    \boldsymbol{b}(Z_i)&=\E(\bD_iY_i \mid Z_i) \\
    \Mb(Z_i)&=\E(\bD_i\bD_i^\top\mid Z_i)
    =\begin{pmatrix}
        1 & Z_i & Z_i\\
        Z_i & 1 + Z_i^2 & Z_i^2\\
        Z_i & Z_i^2 & 1 + Z_i^2
    \end{pmatrix}.
\end{aligned}
\]
Similarly, by derivation we have $h_1\{G_1(\vY_i)\} = h_2\{G_2(\vY_i)\} = Z_i$ that satisfy (b.2) of Theorem~\ref{thrm:psppi-dr}. For each working missing data model, we consider three variants of PS-PPI (DR): 1) $\vq$ and $h_k$ both correctly specified; 2) $\vq$ misspecified but $h_k$ correctly specified; 3) $\vq$ correctly specified but $h_k$ misspecified. For misspecified $\vq$, we let $\vq_i(Y_i, Z_i; \vzeta) = \bV_i(Y_i - \bV_i^\top\vzeta)$, where $\bV_i = (1,Z_i,Z_i^2)^\top$. Because $Y_i$ is also always observed, this working model can be calculated for every subject although it does not satisfy (b.1). For the misspecified $h_k$, we set $h_1\{G_1(\vY_i)\}=\sin(Z_i)$ and $h_2\{G_2(\vY_i)\}=\cos(Z_i)$ to impute $X_{i1}$ and $X_{i2}$, respectively. WCCA and PS-PPI are included in the simulation for reference.

Table~\ref{tab:dr-theorem-verification} presents the coverage rates and the widths of the 95\% confidence intervals based on 500 replicates with sample size $N=50{,}000$. It can be seen that when the missing data model is correctly specified, WCCA, PS-PPI and all three variants of PS-PPI (DR) reach the nominal coverage rate of 95\%. Furthermore, the average width of the 95\% CIs across coefficients from PS-PPI (DR) with correctly specified $\vq$ and $h_k$ is smaller than those of the remaining methods, indicating that correctly specifying $\vq$ and $h_k$ empirically improves the efficiency of the method. On the other hand, when the missing data model is misspecified, the coverage rates of all three coefficients from WCCA and PS-PPI fail to reach the 95\% nominal level, while PS-PPI (DR) shows valid coverage rates when $\vq$ and $h_k$ are both correctly specified. When one of them is misspecified, PS-PPI (DR) also fails to achieve the nominal coverage level. The simulation results empirically verify that PS-PPI (DR) has the double robustness property, and when the missing data model is misspecified, both the working model and the predictive functions are required to satisfy the conditions in order to reach the double robustness property.

\begin{table*}[htbp]
\centering
\caption{Coverage rates and average widths of the 95\% confidence intervals for WCCA, PS-PPI, and three variants of PS-PPI(DR).}
\label{tab:dr-theorem-verification}
\small
\resizebox{\textwidth}{!}{%
\begin{tabular}{llrrrrrr}
\toprule
\multirow{2}{*}{Missing data model} & \multirow{2}{*}{Method} & \multicolumn{2}{c}{$\beta_0$} & \multicolumn{2}{c}{$\beta_1$} & \multicolumn{2}{c}{$\beta_2$} \\
\cmidrule(lr){3-4}\cmidrule(lr){5-6}\cmidrule(lr){7-8}
& & Coverage & CI width & Coverage & CI width & Coverage & CI width \\
\midrule
\multirow{5}{*}{\shortstack[l]{Correctly specified}} & WCCA & 0.960 & 0.033 & 0.952 & 0.027 & 0.940 & 0.027 \\
 & PS-PPI & 0.946 & 0.026 & 0.956 & 0.023 & 0.948 & 0.021 \\
 & PS-PPI(DR): correct $\vq$, correct $h_k$ & 0.952 & 0.025 & 0.958 & 0.023 & 0.942 & 0.021 \\
 & PS-PPI(DR): incorrect $\vq$, correct $h_k$ & 0.940 & 0.048 & 0.966 & 0.031 & 0.946 & 0.037 \\
 & PS-PPI(DR): correct $\vq$, incorrect $h_k$ & 0.960 & 0.024 & 0.944 & 0.024 & 0.942 & 0.024 \\
\midrule
\multirow{5}{*}{\shortstack[l]{Misspecified}} & WCCA & 0.000 & 0.029 & 0.104 & 0.022 & 0.112 & 0.022 \\
 & PS-PPI & 0.000 & 0.023 & 0.026 & 0.019 & 0.014 & 0.017 \\
 & PS-PPI(DR): correct $\vq$, correct $h_k$ & 0.962 & 0.028 & 0.954 & 0.056 & 0.950 & 0.061 \\
 & PS-PPI(DR): incorrect $\vq$, correct $h_k$ & 0.000 & 0.037 & 0.436 & 0.032 & 0.376 & 0.031 \\
 & PS-PPI(DR): correct $\vq$, incorrect $h_k$ & 0.888 & 0.028 & 0.054 & 0.057 & 0.000 & 0.065 \\
\bottomrule
\end{tabular}%
}
\end{table*}

\paragraph{Efficiency analysis of $\widehat{\vtheta}_{\psppi\text{(DR)}}$} We now discuss the asymptotic efficiency of $\widehat{\vtheta}_{\psppi\text{(DR)}}$ when the missing data model is correctly specified, i.e., under part (a) of Theorem~\ref{thrm:psppi-dr} with $\vq$ and $h_k$ arbitrary. We first state the conclusion: as a price of double robustness, $\widehat{\vtheta}_{\psppi\text{(DR)}}$ may become asymptotically less efficient than both $\widehat{\vtheta}_{\psppi}$ and $\widehat{\vtheta}_\wcc$ when the missing data model is correctly specified but $\vq$ is poorly chosen. We do not compare asymptotic efficiency when the missing data model is misspecified, because $\widehat{\vtheta}_{\psppi\text{(DR)}}$ is then the only method that can remain valid.

When the missing data model is correctly specified, $\widehat{\vzeta}_\cc - \widehat{\vzeta}$ is itself an estimate of zero, and the influence function of $\widehat{\vtheta}_{\psppi\text{(DR)}}$ is $\varphi_{\text{DR}, i}(\vtheta^\ast, \vgamma^\ast, \vzeta^\ast, \vnu^\ast, \Wb) = \varphi_i(\vtheta^\ast, \vgamma^\ast, \vnu^\ast, \Wb) - \varphi_{\vq, i}(\vzeta^\ast, \vnu^\ast)$, where $\varphi_i$ is the influence function of $\widehat{\vtheta}_\psppi$ in Theorem~\ref{thrm:general_missing_patterns} and
\begin{gather}
    \label{eq:dr-correction-if}
    \varphi_{\vq, i}(\vzeta^\ast, \vnu^\ast) = \wt{\varphi}_{\vq, i}(\vzeta^\ast, \vnu^\ast) - \E \left [ \nabla_\vnu \wt{\varphi}_{\vq, i}(\vzeta^\ast, \vnu^\ast) \right ]\dot{\wt{\vphi}}^{-1}(\vnu^\ast)\wt{\vphi}_i(\vnu^\ast), \\
    \wt{\varphi}_{\vq, i}(\vzeta^\ast, \vnu^\ast) =  \dot{\vq}^{-1} \left ( 1 - \frac{I(\cC_i = \infty)}{\wt{\pi}_\infty(\vY_{c, i}; \vnu^\ast)} \right ) \vq(\vY_{c, i}; \vzeta^\ast), \notag \\
    \dot{\vq} = \E \left [ \nabla_\vzeta \vq(\vY_{c, i}; \vzeta^\ast) \right ], \qquad \dot{\wt{\vphi}}(\vnu^\ast) = \E \left [ \nabla_\vnu \wt{\vphi}_i(\vnu^\ast) \right ]. \notag
\end{gather}
Recall that $\vzeta^\ast = \vzeta^\ast_\cc$ when the missing data model is correctly specified. Hence, the asymptotic variance-covariance matrix of $\widehat{\vtheta}_{\psppi\text{(DR)}}$ is
\begin{align}
    \Sigma_{\psppi\text{(DR)}}(\vtheta^\ast, \vgamma^\ast, \vzeta^\ast, \vnu^\ast, \Wb) &= \E \left [ \varphi_{\text{DR}, i}(\vtheta^\ast, \vgamma^\ast, \vzeta^\ast, \vnu^\ast, \Wb) \varphi^\top_{\text{DR}, i}(\vtheta^\ast, \vgamma^\ast, \vzeta^\ast, \vnu^\ast, \Wb)\right ] \notag \\
    &= \Sigma_{\vtheta} - \Sigma_{\vtheta,\vDelta}\Wb^\top - \Wb\Sigma_{\vtheta,\vDelta}^\top + \Wb\Sigma_{\vDelta}\Wb^\top \notag \\
    &\quad - \Sigma_{\vtheta, \vq} + \Wb\Sigma^\top_{\vq, \vDelta} - \Sigma^\top_{\vtheta, \vq} + \Sigma_{\vq, \vDelta} \Wb^\top + \Sigma_{\vq},
\end{align}
where
\begin{equation}
\label{eq:sigma_components_dr}
\begin{aligned}
\Sigma_{\vtheta,\vq}
&:=
\E\left[
\varphi_{\infty,i}(\vtheta^\ast,\vnu^\ast)
\varphi_{\vq,i}^{\top}(\vzeta^\ast,\vnu^\ast)
\right], \\[0.5em]
\Sigma_{\vq, \vDelta}
&:=
\E\left[
\varphi_{\vq,i}(\vzeta^\ast,\vnu^\ast)
\varphi_{\vDelta,i}^{\top}(\vgamma^\ast,\vnu^\ast)
\right], \\[0.5em]
\Sigma_{\vq}
&:=
\E\left[
\varphi_{\vq,i}(\vzeta^\ast,\vnu^\ast)
\varphi_{\vq,i}^{\top}(\vzeta^\ast,\vnu^\ast)
\right].
\end{aligned}
\end{equation}
As in Corollary~\ref{corollary:optimal_w}, we can obtain the optimal choice of $\Wb$ for $\widehat{\vtheta}_{\psppi\text{(DR)}}$: $\Wb_{\text{opt (DR)}} = \left ( \Sigma_{\vtheta, \vDelta} - \Sigma_{\vq, \vDelta} \right ) \Sigma^{-1}_\vDelta$, and the difference between $\Sigma_{\psppi\text{(DR)}}(\vtheta^\ast, \vgamma^\ast, \vzeta^\ast, \vnu^\ast, \Wb_{\text{opt(DR)}})$ and $\Sigma_{\psppi}(\vtheta^\ast, \vgamma^\ast, \vnu^\ast, \Wb_{\text{opt}})$ is:
\begin{align}
    &\Sigma_{\psppi\text{(DR)}}(\vtheta^\ast, \vgamma^\ast, \vzeta^\ast, \vnu^\ast, \Wb_{\text{opt(DR)}}) - \Sigma_{\psppi}(\vtheta^\ast, \vgamma^\ast, \vnu^\ast, \Wb_{\text{opt}}) \notag \\
    =& - \Sigma_{\vtheta,\vq} - \Sigma_{\vtheta,\vq}^\top + \Sigma_\vq + \Sigma_{\vtheta,\vDelta} \Sigma^{-1}_\vDelta \Sigma^\top_{\vq,\vDelta} + \Sigma_{\vq,\vDelta} \Sigma^{-1}_\vDelta \Sigma^\top_{\vtheta,\vDelta} - \Sigma_{\vq,\vDelta} \Sigma^{-1}_\vDelta \Sigma^\top_{\vq,\vDelta}. \label{eq:sigma_psppi_dr_opt}
\end{align}
By Equation (\ref{eq:sigma_psppi_dr_opt}), $\widehat{\vtheta}_{\psppi\text{(DR)}}$ is not guaranteed to be more efficient than either $\widehat{\vtheta}_\psppi$ or WCCA. To see this, consider a simple counterexample: let $\vY_{c, i}$ include a component $\vY_{c, i}^{\perp}$ that is independent of $\cC_i$ and the remaining components of $\vY_i$, and let $\vq$ depend only on $\vY_{c, i}^{\perp}$, while $\vpsi$, $h_1, \cdots, h_K$, and the working missing data model $\wt{\bpi}$ do not involve $\vY_{c, i}^{\perp}$. Then $\Sigma_{\vtheta, \vq}$ and $\Sigma_{\vq, \vDelta}$ are asymptotically zero, and only the positive definite $\Sigma_\vq$ remains in Equation (\ref{eq:sigma_psppi_dr_opt}). Since $\Sigma_\vq$ can be made arbitrarily large, $\widehat{\vtheta}_{\psppi\text{(DR)}}$ can be arbitrarily less efficient than PS-PPI and even WCCA. The simulation below illustrates this with an independent working model.

We conduct a simulation to verify the above statement empirically. We repeat the simulation in Section~\ref{sec:experiments}, setting $\sigma^2_{\text{pred}} = 0$ and $\lambda_{\mathrm{pred}} \in \{0, 0.05, 0.1, 0.2, 0.5, 1, 1.5, 2\}$ for $Y$, $X_1$, and $X_2$, and consider the following working models $\vq$ for PS-PPI (DR):
\begin{enumerate}
    \item \textbf{Independent $\vq$:} Let $\vY_{c, i}^{\perp}=(Q_{i1},Q_{i2},Q_{Yi})^\top$ be always-observed variables, where $Q_{i1}, Q_{i2}\overset{\mathrm{i.i.d.}}{\sim}\cN(0, 1)$ and $Q_{Yi}\sim\cN(0,\sigma^2_\vq)$ are independent of the simulated data and the missingness indicators. With $\bQ_i=(1,Q_{i1},Q_{i2})^\top$, we set $\vq(\vY_{c, i}^{\perp};\vzeta)=\bQ_i(Q_{Yi}-\bQ_i^\top\vzeta)$ and vary $\sigma_\vq\in\{1,2,5\}$. These models illustrate the efficiency counterexample described above.
    \item \textbf{Prediction-based $\vq$:} Since predictions of $Y_i$, $X_{i1}$, and $X_{i2}$ are observed for all subjects, we define $\vY_{c,i} = (\widehat{Y}_i, \widehat{X}_{i1},\widehat{X}_{i2})$. We set $\vq_i(\vY_{c,i}, \vzeta)=\widehat{\bX}_i(\widehat{Y}_i-\widehat{\bX}_i^\top\vzeta)$, where $\widehat{\bX}_i=(1,\widehat{X}_{i1},\widehat{X}_{i2})^\top$.
\end{enumerate}
WCCA and PS-PPI are included as baselines. Figure~\ref{fig:psppi_dr_working_model_efficiency} presents the simulation results. First, the 95\% confidence intervals from PS-PPI (DR) achieve the nominal 95\% coverage for all three coefficients under every choice of $\vq$, which corroborates the asymptotic validity of PS-PPI (DR) under condition (a) of Theorem~\ref{thrm:psppi-dr}. Second, for PS-PPI (DR) with independent $\vq$, the interval widths increase with both $\sigma_\vq$ and the prediction bias $\lambda_\text{pred}$; when $\sigma_\vq = 5$, PS-PPI (DR) can yield wider 95\% CIs than WCCA and PS-PPI for all three coefficients. For PS-PPI (DR) with prediction-based $\vq$, confidence interval widths increase as prediction quality deteriorates. When predictions are sufficiently biased, the confidence intervals for $\beta_0$ can be wider than those of WCCA and PS-PPI. These results show that the doubly robust correction in $\widehat{\vtheta}_{\psppi \text{(DR)}}$ does not guarantee an efficiency gain over WCCA.

\begin{figure}[!t]
    \centering
    \includegraphics[width=1.0\linewidth]{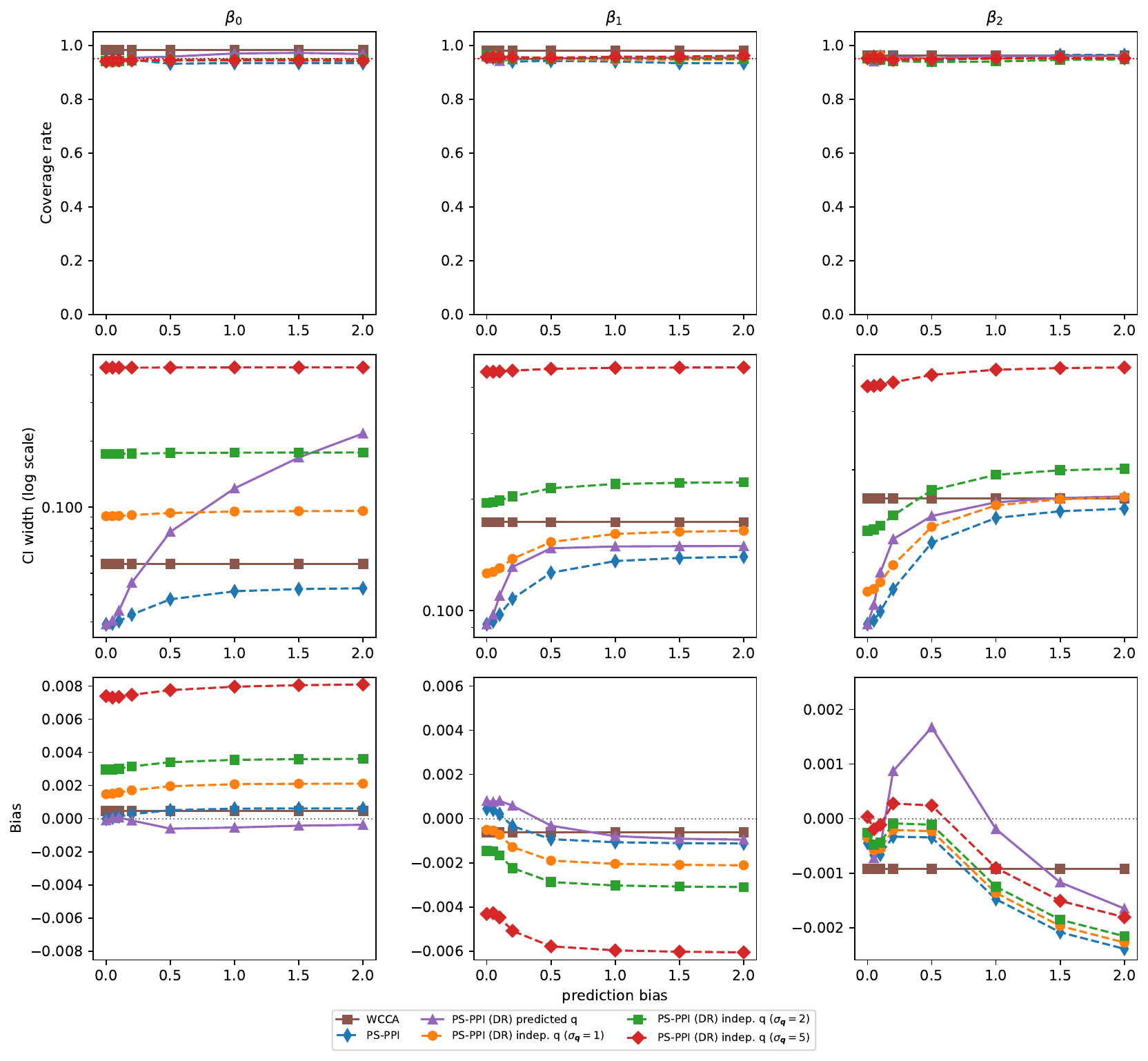}
    \caption{Comparison of WCCA, PS-PPI, and PS-PPI (DR) under different $\vq$. Rows report the empirical coverage of 95\% confidence intervals (CIs), the widths of 95\% CIs (log scale), and the empirical bias; columns correspond to $\beta_0$, $\beta_1$, and $\beta_2$. The red dotted line denotes the nominal 95\% coverage level, and the gray dotted line denotes zero bias. The X-axis represents the prediction bias level $\lambda_\text{pred}$.}
    \label{fig:psppi_dr_working_model_efficiency}
\end{figure}

\paragraph{Reduction of $\widehat{\vtheta}_{\psppi\text{(DR)}}$ to the AIPW estimator in linear regression} We now show how $\widehat{\vtheta}_{\psppi\text{(DR)}}$ reduces to the commonly used AIPW estimator in a linear regression example in which only the outcome is missing under MAR. Specifically, we show that $\widehat{\vtheta}_{\psppi\text{(DR)}}$ is asymptotically equivalent to the AIPW estimator when the missing data model is correctly specified, and $\vq$ and $h_k$ satisfy (b.1) and (b.2) in Theorem~\ref{thrm:psppi-dr}.

Consider a single missingness pattern ($K = 1$) besides the fully observed pattern, in which only the outcome $Y_{i1}$ is missing while the remaining variables $\vY_{c, i} \in \R^{d}$, $d = p - 1$, are fully observed, so that $G_1(\vY_i) = \vY_{c, i}$ and $\cC_i \in \{1, \infty\}$; we impose the MAR assumption in Equation (\ref{eq:chaudhuri-mar}). The goal is to make inference on $\vtheta^\ast \in \R^{d}$, the regression coefficients of $Y_{i1}$ on $\vY_{c, i}$. Note that we assume the intercept term is contained within $\vY_{c, i}$. As a result, the estimating function is
$$\vpsi(\vY_i; \vtheta) = \vY_{c, i} \left ( Y_{i1} - \vY_{c, i}^\top \vtheta \right ).$$
The parameter of interest $\vtheta^\ast$ solves the population-level estimating equation
$$\E_{\vY} \left [ \vpsi(\vY_i; \vtheta) \right ] = \E_{\vY} \left [ \vY_{c, i} \left ( Y_{i1} - \vY_{c, i}^\top \vtheta \right ) \right ] = \vzero.$$

In this setting, it can be derived that the correctly specified $\vq$ is $\vq(\vY_{c, i}; \vzeta) = \vY_{c, i} \{ \E [ Y_{i1} | \vY_{c, i} ] - \vY_{c, i}^\top \vzeta \}$, where $\E [ Y_{i1} | \vY_{c, i} ]$ is known as the outcome regression model in the causal inference literature~\citep{robins1994estimation, tsiatis2006semiparametric}, and the correctly specified $h$ is $h(\vY_{c, i}) = \E [ Y_{i1} | \vY_{c, i} ]$. Hence, $\vq$ can be rewritten as $\vq(\vY_{c, i}; \vzeta) =  \vY_{c, i} \{ h(\vY_{c, i}) - \vY_{c, i}^\top \vzeta \}$. We further set $\Wb_1 = \vzero$, i.e., we discard $\widehat{\vgamma}_\cc$ and $\widehat{\vgamma}_\ic$ in this example, so that the predictive function $h$ enters $\widehat{\vtheta}_{\psppi \text{(DR)}}$ only through $\vq$.

We now specialize each component of $\widehat{\vtheta}_{\psppi\text{(DR)}}$ in Equation (\ref{eq:psppi_dr}) to this setting. The base estimator $\widehat{\vtheta}_\wcc$ solves the IPW least-squares equation over the complete cases:
$$\sum^N_{i=1} \frac{I(\cC_i = \infty)}{\wt{\pi}_\infty(\vY_{c, i}; \widehat{\vnu})} \vY_{c, i} ( Y_{i1} - \vY_{c, i}^\top \vtheta ) = \vzero.$$
On the other hand, $\widehat{\vzeta}_\cc$ and $\widehat{\vzeta}$ solve the following estimating equations, respectively:
$$\sum^N_{i=1} \frac{I(\cC_i = \infty)}{\wt{\pi}_\infty(\vY_{c, i}; \widehat{\vnu})} \vY_{c, i} \left \{ h(\vY_{c, i}) - \vY_{c, i}^\top \vzeta \right \} = \vzero,$$
$$\sum^N_{i=1} \vY_{c, i} \left \{ h(\vY_{c, i}) - \vY_{c, i}^\top \vzeta \right \} = \vzero.$$
For simplicity, we ignore the uncertainty from estimating the missing data model and treat $\wt{\pi}_\infty(\vY_{c, i}; \widehat{\vnu})$ as $\wt{\pi}_\infty(\vY_{c, i}; \vnu^\ast)$ in the following discussion. As such, the $i$-th influence function of $\widehat{\vtheta}_{\psppi\text{(DR)}}$ is
\begin{align}
    \varphi_{\text{DR}, i}(\wt{\vtheta}^\ast, \vzeta^\ast_\cc, \vzeta^\ast, \vnu^\ast) &= \E \left [ a_\infty(\vY_{c, i}) \vY_{c, i} \vY_{c, i}^\top \right ]^{-1} \frac{I(\cC_i = \infty)}{\wt{\pi}_\infty(\vY_{c, i}; \vnu^\ast)} \vY_{c, i} \left \{ Y_{i1} - \vY_{c, i}^\top \wt{\vtheta}^\ast - h(\vY_{c, i}) + \vY_{c, i}^\top \vzeta^\ast_\cc \right \} \notag \\
    &\quad + \E \left [ \vY_{c, i} \vY_{c, i}^\top \right ]^{-1} \vY_{c, i} \left \{ h(\vY_{c, i}) - \vY_{c, i}^\top \vzeta^\ast \right \}. \label{eq:psppi-dr-linear-regression}
\end{align}
When the missing data model is correctly specified, and further both (b.1) and (b.2) of Theorem~\ref{thrm:psppi-dr} are satisfied, Equation (\ref{eq:psppi-dr-linear-regression}) becomes
\begin{align}
    \varphi_{\text{DR}, i}(\vtheta^\ast, \vnu^\ast) &= \E \left [ \vY_{c, i} \vY_{c, i}^\top \right ]^{-1} \vY_{c, i} \left \{  \frac{I(\cC_i = \infty)}{\pi_\infty(\vY_{c, i})} \left ( Y_{i1} - \E [Y_{i1} | \vY_{c, i}] \right ) + \left ( \E [Y_{i1} | \vY_{c, i}] - \vY_{c, i}^\top \vtheta^\ast \right ) \right \} \notag \\
    &= \E \left [ \vY_{c, i} \vY_{c, i}^\top \right ]^{-1} \vY_{c, i} \left \{  \frac{I(\cC_i = \infty)}{\pi_\infty(\vY_{c, i})} \left ( Y_{i1} - \vY_{c, i}^\top \vtheta^\ast \right ) \right. \notag \\
    &\qquad \left. - \left [ \frac{I(\cC_i = \infty)}{\pi_\infty(\vY_{c, i})} - 1 \right ] \left ( \E \left [ Y_{i1} | \vY_{c, i} \right ] - \vY_{c, i}^\top \vtheta^\ast \right ) \right \}, \label{eq:psppi-dr-linear-regression-semiparametric-efficient}
\end{align}
which is the influence function of the AIPW estimator when both the missing data model and the outcome regression model are correctly specified. Although PS-PPI (DR) requires externally trained predictive functions $h_k$ under general missingness patterns and the MAR assumption, this example is an exception. The outcome regression function $h(\vY_{c,i}) = \E\left[Y_{i1} \mid \vY_{c,i}\right]$ can be estimated from the complete cases, and the uncertainty from this internally fitted model can be accounted for through a two-step Z-estimation framework, as in the classical AIPW approach.

\subsection{Theoretical comparison between PS-PPI and multiple imputation}
\label{sec:theoretical-comparison-mi-psppi}

In this subsection, we theoretically compare the large-sample asymptotic efficiency of PS-PPI and multiple imputation (MI). We begin with a brief review of the theory of MI, following~\citet{rubin1987multiple} and~\citet{murray2018multiple}. Specifically, suppose that $\cD = \{\vY_1, \dots, \vY_N\}$ denotes the underlying full data, and let $\vY_{\obs, i} = G_{\cC_i}(\vY_i)$ denote the observed components of $\vY_i$, where $\cC_i$, as defined in Section~\ref{sec:preliminary}, is the missingness index for the $i$-th subject to indicate its observed missingness pattern. The corresponding missing component is denoted as $\vY_{\mis, i}$. MI repeatedly generates $M$ imputations of $\vY_{\mis,i}$, denoted as $\widehat{\vY}^{(1)}_{\mis, i}, \dots, \widehat{\vY}^{(M)}_{\mis, i}$, from a predictive distribution, or imputation model, $p(\vY_{\mis,i} \mid \vY_{\obs,i})$, which must be specified and fitted using the available observed data $\cD_{\obs} = \{(G_{\cC_i}(\vY_i), \cC_i)\}^N_{i=1}$. This procedure produces $M$ completed datasets. For the $m$-th imputation, we denote the completed data for the $i$-th subject by $\vYimp^{(m)}_i = \{\vY_{\obs, i}, \widehat{\vY}^{(m)}_{\mis, i} \}$. For the inferential target defined in Section~\ref{sec:preliminary}, each completed dataset yields a Z-estimator $\widehat{\vtheta}^{(1)}, \dots, \widehat{\vtheta}^{(M)}$ and a corresponding estimated sampling variance $\widehat{\Sigma}^{(1)},\dots,\widehat{\Sigma}^{(M)}$ from $M$ different imputed datasets. Specifically, $\widehat{\vtheta}^{(m)}$ and $\widehat{\Sigma}^{(m)}$ are obtained from the $m$-th completed dataset through the following estimating equations:
\begin{equation}
    \sum^N_{i=1} \vpsi \left (\vYimp^{(m)}_i; \vtheta \right ) = \vzero.
\end{equation}
The MI estimator $\widehat{\vtheta}_{\mi}$ is then obtained by averaging $\widehat{\vtheta}^{(m)}$ across the $M$ completed datasets:
\begin{equation}
    \widehat{\vtheta}_{\text{MI}} = \frac{1}{M} \sum^M_{m=1} \widehat{\vtheta}^{(m)}.
\end{equation}
According to Rubin's rule, the estimated sampling variance of $\widehat{\vtheta}_{\mi}$ is
\begin{equation}
    \label{eq:mi_rubins_rule_variance}
    \widehat{\Sigma}_M = \bar{\widehat{\Sigma}}_M + \left ( 1 + \frac{1}{M} \right ) \Bb_M,
\end{equation}
where
\begin{equation}
    \bar{\widehat{\Sigma}}_M =  \frac{1}{M}\sum^M_{m=1} \widehat{\Sigma}^{(m)}
\end{equation}
represents the within-imputation variance, and
\begin{equation}
    \Bb_M = \frac{1}{M - 1}\sum^M_{m=1}\left ( \widehat{\vtheta}^{(m)} - \widehat{\vtheta}_{\mi} \right )\left ( \widehat{\vtheta}^{(m)} - \widehat{\vtheta}_{\mi} \right )^\top
\end{equation}
represents the between-imputation variability.

\citet{wang1998large} derived the large-sample behavior of MI estimators when the full data are sampled i.i.d.\ from a correctly specified regular parametric family $\{f(\vY;\vtheta): \vtheta \in \bm{\Theta} \subset \R^d\}$, with each missing entry imputed $M$ times. Under proper imputation, the asymptotic covariance matrix of $\sqrt{N}(\widehat{\vtheta}_{\mi} - \vtheta^\ast)$, denoted by $\Sigma_{\mi, M}$, is
\begin{equation}
\label{eq:mi_asymptotic_covariance_wang_robins}
   \Sigma_{\mi, M} = \Ib^{-1}_{\obs} + \frac{1}{M} \left ( \Ib^{-1}_{\full} \Jb + \Jb^\top \Ib^{-1}_{\obs} \Jb  \right ),
\end{equation}
where $\Ib_{\full}$ and $\Ib_{\obs}$ denote the full-data and observed-data Fisher information matrices for $\vtheta$, respectively. The missing information matrix is $\Ib_\mis = \Ib_{\text{Full}} - \Ib_{\text{obs}}$, and $\Jb = \Ib_\mis \Ib_{\full}^{-1}$ is the corresponding fraction of missing information. Equation (\ref{eq:mi_asymptotic_covariance_wang_robins}) can be rewritten as
\begin{equation}
    \label{eq:mi_asymptotic_covariance_wang_robins_rewrite}
    \Sigma_{\mi, M} = \Ib^{-1}_{\obs} + \frac{1}{M} \left ( \Ib^{-1}_\obs - \Ib^{-1}_\full \right ) = \Ib^{-1}_{\full} + \left (1 + \frac{1}{M} \right ) \left ( \Ib^{-1}_\obs - \Ib^{-1}_\full \right ).
\end{equation}
To relate Equation (\ref{eq:mi_asymptotic_covariance_wang_robins_rewrite}) to Rubin's rule in Equation (\ref{eq:mi_rubins_rule_variance}), note that, as stated in \citet{rubin1987multiple} and \citet{wang1998large}, Rubin's variance estimator has asymptotic mean $\Sigma_{\mi, M}$ for finite $M$. In Rubin's rule, $N\bar{\widehat{\Sigma}}_M$ estimates the full-data asymptotic variance $\Ib^{-1}_{\full}$, whereas $N\Bb_M$ estimates $\Ib^{-1}_\obs - \Ib^{-1}_\full$, the additional uncertainty due to missingness.

By contrast, under the same parametric model, the asymptotic variance-covariance matrix of $\sqrt{N}(\widehat{\vtheta}_{\psppi} - \vtheta^\ast) $ with optimal tuning parameters can be written as
\begin{equation} \label{eq:psppi_asymptotic_covariance_parametric} \Sigma_{\psppi} = \Ib^{-1}_{\wcc} - \Sigma_{\vtheta, \vDelta} \Ib_{\vDelta} \Sigma_{\vtheta, \vDelta}^\top,
\end{equation}
where $\Ib_{\wcc} = \Sigma^{-1}_{\vtheta}$ denotes the Fisher information matrix of WCCA, and $\Ib_{\vDelta} = \Sigma^{-1}_{\vDelta}$ denotes the Fisher information matrix of the estimates of zero. Equation (\ref{eq:psppi_asymptotic_covariance_parametric}) can be written equivalently as
\begin{align}
    \Sigma_{\psppi}
    &= \Ib^{-1}_{\wcc} - \Ib_{\wcc}^{-1/2}\left(\Ib_{\vDelta}^{1/2} \Sigma_{\vtheta,\vDelta}^\top \Ib_{\wcc}^{1/2} \right)^\top \left( \Ib_{\vDelta}^{1/2} \Sigma_{\vtheta,\vDelta}^\top \Ib_{\wcc}^{1/2} \right) \Ib_{\wcc}^{-1/2} \\
    &= \Ib^{-1}_{\wcc}
    - \Ib_{\wcc}^{-1/2} \Kb_{\vtheta,\vDelta}^\top \Kb_{\vtheta,\vDelta} \Ib_{\wcc}^{-1/2}.
\label{eq:psppi_asymptotic_covariance_correlation_form}
\end{align}
When $\vtheta^\ast$ is scalar, $\Kb_{\vtheta,\vDelta} = \Ib_{\vDelta}^{1/2} \Sigma_{\vtheta,\vDelta}^\top \Ib_{\wcc}^{1/2} \in \R^{K}$ can be viewed as a measure of correlation between $\widehat{\vtheta}_\wcc$ and $\widehat{\vDelta}$. For each $k$, $\widehat{\vgamma}_{\ic, k}$ is independent of $\widehat{\vtheta}_\wcc$ when the uncertainty in $\widehat{\bpi}$ is ignored, because the two estimators are computed from disjoint subsets of the data. As a result, $\widehat{\vtheta}_\wcc$ and $\widehat{\vDelta}_k$ are strongly correlated when $\widehat{\vtheta}_\wcc$ is strongly correlated with $\widehat{\vgamma}_{\cc, k}$, as occurs when the imputations are of high quality. Therefore, better imputations tend to increase the magnitude of $\Kb_{\vtheta,\vDelta}$ and decrease $\Sigma_\psppi$.

Using Equations (\ref{eq:mi_asymptotic_covariance_wang_robins}) and (\ref{eq:psppi_asymptotic_covariance_parametric}), we characterize when PS-PPI is asymptotically more efficient than MI by examining the difference between $\Sigma_{\text{PS-PPI}}$ and $\Sigma_{\text{MI}}$:
\begin{align}
    \Sigma_{\text{PS-PPI}} - \Sigma_{\text{MI}} = \left ( \Ib^{-1}_{\wcc} - \Ib^{-1}_{\text{obs}} \right ) - \Ib_{\wcc}^{-1/2} \Kb_{\vtheta,\vDelta}^\top \Kb_{\vtheta,\vDelta} \Ib_{\wcc}^{-1/2} - M^{-1} \left ( \Ib^{-1}_\obs - \Ib^{-1}_\full \right ) \preceq  \; \vzero \notag \\
    \Leftrightarrow \underbrace{\left ( \Ib^{-1}_{\wcc} - \Ib^{-1}_{\text{obs}} \right )}_{\Tb_1} - \underbrace{ M^{-1} \left ( \Ib^{-1}_\obs - \Ib^{-1}_\full \right ) }_{\Tb_2} \preceq  \; \underbrace{\Ib_{\wcc}^{-1/2} \Kb_{\vtheta,\vDelta}^\top \Kb_{\vtheta,\vDelta} \Ib_{\wcc}^{-1/2}}_{\Tb_3}.
    \label{eq:psppi_mi_variance_comparison_condition}
\end{align}
Equation~(\ref{eq:psppi_mi_variance_comparison_condition}) characterizes the conditions under which PS-PPI can be asymptotically more efficient than MI:
\begin{enumerate}
    \item $\Tb_1$ is the gap between the asymptotic variances of $\widehat{\vtheta}_\wcc$ and $\widehat{\vtheta}_\obs$, the observed-data MLE. It quantifies the additional information about $\vtheta^\ast$ contained in the observed components of incomplete subjects beyond the complete cases. $\Tb_1$ is small when incomplete subjects contribute little additional information beyond the complete cases, which favors PS-PPI.
    \item $\Tb_2$ is the gap in asymptotic variance between the observed-data and full-data MLEs, scaled by the inverse of the number of imputations $M$. $\Tb_2$ is large, which favors PS-PPI, when the fraction of missing information is large (i.e., substantial information about $\vtheta^\ast$ is missing) and when MI uses a relatively small number of imputations $M$.
    \item $\Tb_3$ is the variance reduction achieved by PS-PPI through imputations from externally trained predictive models. As discussed above, $\Tb_3$ is large when the predictive functions generate sufficiently precise imputations such that $\widehat{\vgamma}_{\cc,k}$ is strongly correlated with $\widehat{\vtheta}_\wcc$, which favors PS-PPI.
\end{enumerate}
To further illustrate how $\Ib_\wcc$, $\Ib_\obs$, and $\Ib_\full$ affect the relative asymptotic efficiency of the two methods, we next consider a toy example of mean estimation in which these information quantities have closed forms, showing how they vary with the missingness probability and the correlation between the variables.
\begin{example}[Mean estimation under a bivariate normal model]
\label{example:bivariate_normal_mean}
\end{example}
Consider the data $\cD = \{(X_i, Y_i, \cC_i)\}^N_{i=1}$, where
\[
    (X_i,Y_i)
    \overset{\text{i.i.d.}}{\sim}
    \cN
    \left (
    \begin{pmatrix}
        \mu_X \\
        \mu_Y
    \end{pmatrix},
    \begin{pmatrix}
        \sigma_X^2 & \sigma_{XY}\\
        \sigma_{XY} & \sigma_Y^2
    \end{pmatrix}
    \right ).
\]
We denote the correlation between \(X_i\) and \(Y_i\) by $\rho = \sigma_{XY} / (\sigma_X\sigma_Y)$. The example includes three missingness patterns: 1) \(\cC_i=\infty\) indicates that the subject is fully observed; 2) \(\cC_i=1\) indicates that only \(Y_i\) is observed; and 3) \(\cC_i=2\) indicates that only \(X_i\) is observed. We assume an MCAR mechanism with \(\Pr(\cC_i=k)=\pi_k\), where \(\pi_\infty+\pi_1+\pi_2=1\).
We assume that \(\mu_Y\) is the only unknown parameter of interest and that \(\mu_X\), \(\sigma_X^2\), \(\sigma_Y^2\), and \(\sigma_{XY}\) are known. For PS-PPI, suppose a (possibly nonlinear) transformation $\vz = h^{-1}(x, y)$ provides useful predictors of $X_i$ and $Y_i$. These predictors, denoted by $\vZ_i = h^{-1}(X_i, Y_i)$, are available for every subject but are excluded from the analysis and used only to obtain predicted values $\widehat{X}_i$ and $\widehat{Y}_i$ through a predictive function $\widehat{h}(\vz)$ fitted using external data.

When the full data are available, the full-data maximum likelihood estimator (MLE) $\widehat{\mu}_\full$ maximizes the likelihood
\begin{align}
    L_{\full}(\mu_Y) &= \prod_{i=1}^N f_X(X_i) f_{Y\mid X}(Y_i \mid X_i) \notag \\
    &= \prod_{i=1}^N \left[ \frac{1}{\sqrt{2\pi\sigma_X^2}} \exp\left\{ -\frac{(X_i-\mu_X)^2}{2\sigma_X^2} \right\} \right] \notag \\
    &\quad\times \prod_{i=1}^N \frac{1}{\sqrt{2\pi\sigma_Y^2(1-\rho^2)}} \exp\left\{ -\frac{\left[Y_i - \left\{\mu_Y + \sigma_{XY}/\sigma^2_{X}(X_i - \mu_X)\right\}\right]^2}{2\sigma_Y^2(1-\rho^2)} \right\}.
\end{align}
The corresponding full-data log-likelihood is
\[
    \begin{aligned}
        \ell_{\text{Full}}(\mu_Y)
        &\propto
        -\frac{1}{2\sigma_Y^2(1-\rho^2)}\sum_{i=1}^N
        \left[
        Y_i-\left\{\mu_Y+\frac{\sigma_{XY}}{\sigma_X^2}(X_i-\mu_X)\right\}
        \right]^2 .
    \end{aligned}
\]
Differentiating the log-likelihood gives the full-data score function
\[
    \begin{aligned}
        S_{\text{Full}}(\mu_Y)
        &=
        \frac{1}{\sigma_Y^2(1-\rho^2)}
        \sum_{i=1}^N
        \left[
        Y_i-\left\{\mu_Y+\frac{\sigma_{XY}}{\sigma_X^2}(X_i-\mu_X)\right\}
        \right].
    \end{aligned}
\]
Solving \(S_{\text{Full}}(\mu_Y)=0\) yields
\begin{equation}
    \widehat\mu_{\text{Full}}
    =
    \frac{1}{N}
    \sum_{i=1}^N
    \left[
    Y_i-\frac{\sigma_{XY}}{\sigma_X^2}(X_i-\mu_X)
    \right].
\end{equation}
The per-subject Fisher information is
\[
    I_{\text{Full},1}
    =
    \frac{1}{\sigma_Y^2(1-\rho^2)}.
\]
Thus, standard MLE asymptotics give
\[
    \sqrt{N}\left(\widehat\mu_{\text{Full}}-\mu_Y^\ast\right)
    \overset{d}{\longrightarrow}
    \mathcal N(0, I_{\text{Full},1}^{-1}).
\]
The complete-case MLE uses only fully-observed subjects. Thus,
\begin{align}
    L_{\text{CCA}}(\mu_Y)
    &= \prod_{i=1}^N
    \left\{
    f_X(X_i)f_{Y\mid X}(Y_i\mid X_i)
    \right\}^{I(\cC_i = \infty)} \notag \\
    &\propto
    \prod_{i=1}^N
    \left[
    \frac{1}{\sqrt{2\pi\sigma_Y^2(1-\rho^2)}}
    \exp\left\{
    -\frac{
    \left[
    Y_i-\left\{\mu_Y+\frac{\sigma_{XY}}{\sigma_X^2}(X_i-\mu_X)\right\}
    \right]^2
    }{
    2\sigma_Y^2(1-\rho^2)
    }
    \right\}
    \right]^{I(\cC_i = \infty)} .
\end{align}
The complete-case log-likelihood is
\[
\begin{aligned}
    \ell_{\text{CCA}}(\mu_Y)
    &\propto
    -\frac{1}{2\sigma_Y^2(1-\rho^2)}
    \sum_{i=1}^N I(\cC_i = \infty)
    \left[
    Y_i-\left\{\mu_Y+\frac{\sigma_{XY}}{\sigma_X^2}(X_i-\mu_X)\right\}
    \right]^2 .
\end{aligned}
\]
The complete-case score function is
\[
    \begin{aligned}
        S_{\text{CCA}}(\mu_Y)
        &=
        \frac{1}{\sigma_Y^2(1-\rho^2)}
        \sum_{i=1}^N I(\cC_i = \infty)
        \left[
        Y_i-\left\{\mu_Y+\frac{\sigma_{XY}}{\sigma_X^2}(X_i-\mu_X)\right\}
        \right].
    \end{aligned}
\]
Solving \(S_{\text{CCA}}(\mu_Y)=0\) yields
\begin{equation}
    \widehat\mu_{\text{CCA}}
    =
    \frac{
    \sum_{i=1}^N I(\cC_i = \infty)
    \left[
    Y_i-\dfrac{\sigma_{XY}}{\sigma_X^2}(X_i-\mu_X)
    \right]
    }{
    \sum_{i=1}^N I(\cC_i = \infty)
    }.
\end{equation}
The corresponding per-subject Fisher information is
\[
    I_{\text{CCA},1}
    =
    \frac{\pi_\infty}
    {\sigma_Y^2(1-\rho^2)} = \pi_\infty I_{\full,1}.
\]
Here, $\pi_\infty$ represents the fraction of full-data information retained in the complete cases. Standard MLE asymptotics give
\[
    \sqrt{N}\left(\widehat\mu_{\text{CCA}}-\mu_Y^\ast\right)
    \overset{d}{\longrightarrow}
    \mathcal N(0,I_{\text{CCA},1}^{-1}).
\]
For the observed-data analysis, fully observed subjects contribute the joint density \(f_{X,Y}\), subjects with \(\cC_i=1\) contribute the marginal density \(f_Y\), and subjects with \(\cC_i=2\) contribute the marginal density \(f_X\). Because \(f_X\) does not depend on \(\mu_Y\), the \(X\)-only subjects provide no information about \(\mu_Y\). Ignoring the known missingness-pattern probabilities, the observed-data likelihood is
\begin{align}
    L_{\text{obs}}(\mu_Y)
    &=
    \prod_{i=1}^N
    \left\{f_X(X_i)f_{Y\mid X}(Y_i\mid X_i)\right\}^{I(\cC_i = \infty)}
    \left\{f_Y(Y_i)\right\}^{I(\cC_i = 1)}
    \left\{f_X(X_i)\right\}^{I(\cC_i = 2)} \notag \\
    &\propto
    \prod_{i=1}^N
    \left[
    \frac{1}{\sqrt{2\pi\sigma_Y^2(1-\rho^2)}}
    \exp\left\{
    -\frac{
    \left[
    Y_i-\left\{\mu_Y+\frac{\sigma_{XY}}{\sigma_X^2}(X_i-\mu_X)\right\}
    \right]^2
    }{
    2\sigma_Y^2(1-\rho^2)
    }
    \right\}
    \right]^{I(\cC_i = \infty)} \notag \\
    &\quad\times
    \prod_{i=1}^N
    \left[
    \frac{1}{\sqrt{2\pi\sigma_Y^2}}
    \exp\left\{
    -\frac{(Y_i-\mu_Y)^2}{2\sigma_Y^2}
    \right\}
    \right]^{I(\cC_i = 1)} .
\end{align}

The observed-data log-likelihood is
\[
    \begin{aligned}
        \ell_{\text{obs}}(\mu_Y)
        &\propto
        -\frac{1}{2\sigma_Y^2(1-\rho^2)}
        \sum_{i=1}^N I(\cC_i = \infty)
        \left[
        Y_i-\left\{\mu_Y+\frac{\sigma_{XY}}{\sigma_X^2}(X_i-\mu_X)\right\}
        \right]^2 \\
        &\quad
        -\frac{1}{2\sigma_Y^2}
        \sum_{i=1}^N I(\cC_i = 1)(Y_i-\mu_Y)^2 .
    \end{aligned}
\]
The observed-data score function is
\[
    \begin{aligned}
        S_{\text{obs}}(\mu_Y)
        &=
        \frac{1}{\sigma_Y^2(1-\rho^2)}
        \sum_{i=1}^N I(\cC_i = \infty)
        \left[
        Y_i-\left\{\mu_Y+\frac{\sigma_{XY}}{\sigma_X^2}(X_i-\mu_X)\right\}
        \right] \\
        &\quad
        +\frac{1}{\sigma_Y^2}
        \sum_{i=1}^N I(\cC_i = 1)(Y_i-\mu_Y).
    \end{aligned}
\]
Solving \(S_{\text{obs}}(\mu_Y)=0\) yields the observed-data MLE
\begin{equation}
    \widehat\mu_{\text{obs}}
    =
    \frac{
    \dfrac{1}{\sigma_Y^2(1-\rho^2)}
    \sum_{i=1}^N I(\cC_i = \infty)
    \left[
    Y_i-\dfrac{\sigma_{XY}}{\sigma_X^2}(X_i-\mu_X)
    \right]
    +
    \dfrac{1}{\sigma_Y^2}
    \sum_{i=1}^N I(\cC_i = 1)Y_i
    }{
    \dfrac{1}{\sigma_Y^2(1-\rho^2)}
    \sum_{i=1}^N I(\cC_i = \infty)
    +
    \dfrac{1}{\sigma_Y^2}
    \sum_{i=1}^N I(\cC_i = 1)
    }.
\end{equation}
The corresponding per-subject Fisher information is
\[
    I_{\text{obs},1}
    =
    \frac{\pi_\infty}
    {\sigma_Y^2(1-\rho^2)}
    +
    \frac{\pi_1}{\sigma_Y^2} = I_{\full, 1} \left ( \pi_\infty + \pi_1 (1 - \rho^2) \right ).
\]
Thus, $\pi_\infty + \pi_1 (1 - \rho^2)$ is the fraction of full-data information retained in the observed data. Let $\tau = \pi_\infty + \pi_1 (1 - \rho^2)$; correspondingly, $1 - \tau = \pi_2 + \pi_1\rho^2$ is the fraction of information contained in the missing components. Standard MLE asymptotics give
\[
    \sqrt{N}\left(\widehat\mu_{\text{obs}}-\mu_Y^\ast\right)
    \overset{d}{\longrightarrow}
    \mathcal N(0, I_{\text{obs},1}^{-1}).
\]
Consequently,
\begin{equation}
    \label{eq:cca_obs_diff_toy_example}
    \Tb_1 = I_{\wcc,1}^{-1} - I_{\obs,1}^{-1} = \left( \frac{1}{\pi_\infty} - \frac{1}{\pi_\infty + \pi_1(1 - \rho^2)} \right)I^{-1}_{\full,1},
\end{equation}
\begin{equation}
    \label{eq:obs_full_diff_toy_example}
    \Tb_2 = \frac{1}{M} \left ( I_{\obs,1}^{-1} - I_{\full,1}^{-1} \right ) = \frac{1 - \tau}{M\tau} I^{-1}_{\full, 1}.
\end{equation}
\begin{figure}[!h]
    \centering
    \includegraphics[width=1.0\linewidth]{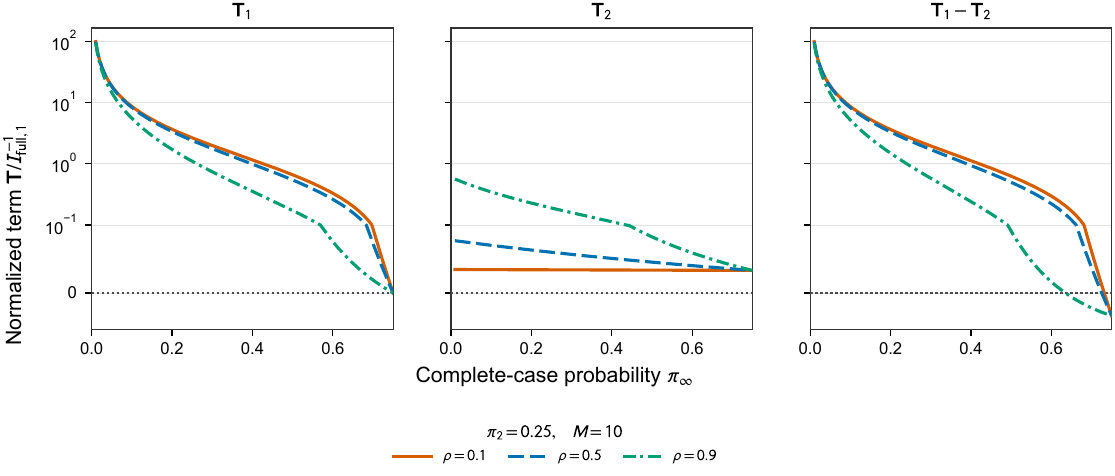}
    \caption{The left, center, and right panels plot $\Tb_1=I_{\wcc,1}^{-1}-I_{\obs,1}^{-1}$, $\Tb_2=M^{-1}(I_{\obs,1}^{-1}-I_{\full,1}^{-1})$, and $\Tb_1-\Tb_2$, respectively, against the complete-case probability $\pi_\infty$. We fix $\pi_2=0.25$ and $M=10$, set $\pi_1=1-\pi_\infty-\pi_2$, and consider $\rho$ values from $\{0.1, 0.5, 0.9\}$. By Equation~(\ref{eq:psppi_mi_variance_comparison_condition}), PS-PPI is at least as asymptotically efficient as MI whenever the variance-reduction term $\Tb_3$ meets or exceeds the threshold $\Tb_1-\Tb_2$.}
    \label{fig:t1_t2_vs_pi0}
\end{figure}
Equations (\ref{eq:cca_obs_diff_toy_example}) and (\ref{eq:obs_full_diff_toy_example}) show that $\Tb_1$ and $\Tb_2$ are rescaled versions of $\Ib^{-1}_{\full, 1}$. Figure~\ref{fig:t1_t2_vs_pi0} plots $\Tb_1$, $\Tb_2$, and their difference across settings of $(\pi_\infty, \rho)$ with $\pi_2 = 0.25$ and $M = 10$. The difference between $\Tb_1$ and $\Tb_2$ is smaller when the proportion of fully observed subjects is larger, which increases $\Ib_\wcc$, or when the correlation $\rho$ between $X$ and $Y$ is stronger, in which case $Y$-only subjects contribute less information relative to fully observed subjects.

In addition, we provide Figure~\ref{fig:inverse_information_vs_pi0} to compare the asymptotic efficiencies of PS-PPI and MI in this toy example by plotting the relative inverse information $I_{\obs,1}^{-1}/I_{\full,1}^{-1}$ and $I_{\wcc,1}^{-1}/I_{\full,1}^{-1}$ across values of $\pi_\infty$, $\pi_2$, and $\rho$. We select $\pi_2$ from $\{0.25, 0.5, 0.75\}$ and $\rho$ from $\{0.1, 0.5, 0.9\}$. Under the conditions of the toy example, the asymptotic variance of optimally weighted PS-PPI is no greater than $I_{\wcc,1}^{-1}$, whereas $I_{\obs,1}^{-1}$ is the limiting lower bound for proper MI as $M\to\infty$. For MI with one imputation, $M = 1$, we also plot $\Sigma_{\mathrm{MI}}(M=1)=2I_{\obs,1}^{-1}-I_{\full,1}^{-1}$, the largest MI variance under the same data-generating setting. The blue band between $I_{\obs,1}^{-1}$ and $\Sigma_{\mathrm{MI}}(M=1)$ therefore spans the range of MI asymptotic variances. The comparison thus reduces to the position of $\Ib^{-1}_{\wcc, 1}$, the upper bound on the asymptotic variance of PS-PPI, relative to the lower bound of this band. Figure~\ref{fig:inverse_information_vs_pi0} shows that as the proportion of fully observed subjects increases relative to that of $Y$-only subjects, $\Ib^{-1}_{\wcc, 1}$ decreases and approaches the lower bound of the blue band. This pattern is more pronounced when $X$ and $Y$ are strongly correlated or when the proportion of $X$-only subjects increases.

\begin{figure}
    \centering
    \includegraphics[width=1.0\linewidth]{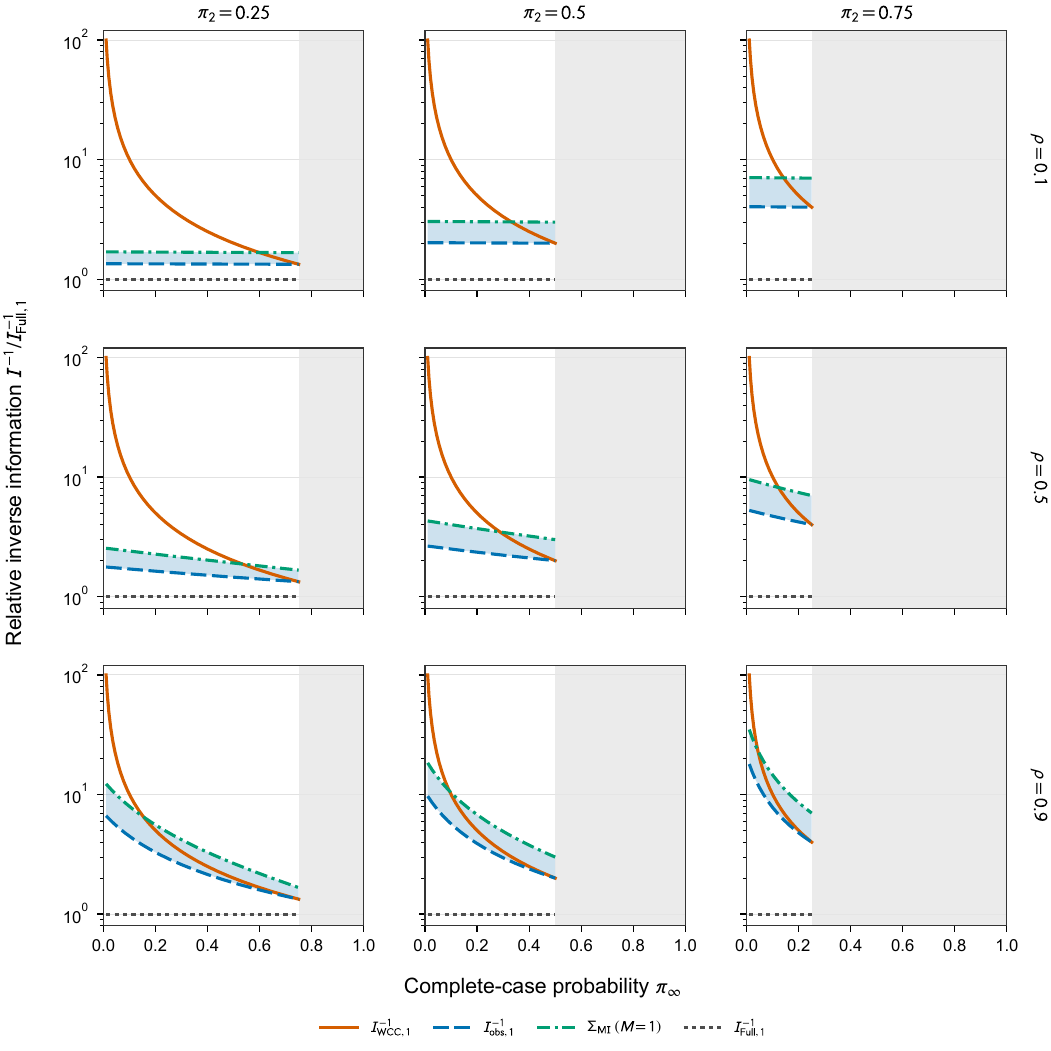}
    \caption{Relative inverse information in the toy example. Columns correspond to \(\pi_2\in\{0.25,0.5,0.75\}\), and rows correspond to \(\rho\in\{0.1,0.5,0.9\}\). In each panel, the X-axis represents the proportion of complete cases \(\pi_\infty\), and the Y-axis represents the inverse information normalized by \(I_{\full,1}^{-1}\). The blue band spans the MI variance range over \(M\geq 1\).}
    \label{fig:inverse_information_vs_pi0}
\end{figure}

\section{Implementation Details}
\label{sec:implementation-details}

\subsection{PS-PPI variance estimation}
\label{sec:variance-estimation}

The theoretical analysis in terms of the asymptotics from Theorem~\ref{thrm:general_missing_patterns} and the closed form of the asymptotic variance-covariance matrix under the globally optimal $\Wb_\opt$ in Equation (\ref{eq:optimal_asymptotic_variance}) provides us with the theoretical foundation to make valid statistical inference on $\vtheta^\ast$ using $\widehat{\vtheta}_{\text{PS-PPI}}$. For uncertainty quantification,
estimating $\Sigma_{\text{PS-PPI}}(\vtheta^\ast, \vgamma^\ast, \vomega^\ast, \Wb_\opt)$ can be decomposed into obtaining consistent estimates $\widehat{\Sigma}_{\vtheta}$, $\widehat{\Sigma}_{\vtheta, \vDelta}$, and $\widehat{\Sigma}_{\vDelta}$ through computing the empirical averages corresponding to Proposition~\ref{proposition:wcc} and Equation (\ref{eq:sigma_components}) evaluated at $(\widehat{\vtheta}_{\text{WCC}}, \widehat{\vgamma}, \widehat{\vomega})$, where $\widehat{\vgamma} = (\widehat{\vgamma}_{1}^\top, \dots, \widehat{\vgamma}_{K}^\top)^\top$, and each component  $\widehat{\vgamma}_k$ can be any consistent estimate of $\vgamma_k^\ast$ (e.g., $\widehat{\vgamma}_{\cc, k}$ or $\widehat{\vgamma}_{\ic, k}$). Specifically,
\begin{equation}
\label{eq:sigma_estimate_components}
\begin{aligned}
\widehat{\Sigma}_{\vtheta}
&:=
\frac{1}{N}\sum^N_{i=1}\left[
\varphi_{\infty,i}(\widehat{\vtheta}_{\text{WCC}}, \widehat{\vomega})
\varphi_{\infty,i}^{\top}(\widehat{\vtheta}_{\text{WCC}}, \widehat{\vomega})
\right], \\[0.5em]
\widehat{\Sigma}_{\vtheta,\vDelta}
&:=
\frac{1}{N}\sum^N_{i=1} \left[
\varphi_{\infty,i}(\widehat{\vtheta}_{\text{WCC}}, \widehat{\vomega})
\varphi_{\vDelta,i}^{\top}(\widehat{\vgamma}, \widehat{\vomega})
\right], \\[0.5em]
\widehat{\Sigma}_{\vDelta}
&:=
\frac{1}{N}\sum^N_{i=1} \left[
\varphi_{\vDelta,i}(\widehat{\vgamma}, \widehat{\vomega})
\varphi_{\vDelta,i}^{\top}(\widehat{\vgamma}, \widehat{\vomega})
\right].
\end{aligned}
\end{equation}
For the ease of implementation, one may also consider a convenient variant of the variance-covariance estimate, which ignores the uncertainty contributed by the estimated missing data model. Specifically, one may estimate $\Sigma_\vtheta$, $\Sigma_{\vtheta, \vDelta}$, and $\Sigma_\vDelta$ with the following empirical averages:
\begin{equation}
\label{eq:sigma_estimate_components_no_pi}
\begin{aligned}
\widehat{\widetilde{\Sigma}}_{\vtheta}
&:=
\frac{1}{N}\sum^N_{i=1}\left[
\widetilde{\varphi}_{\infty,i}(\widehat{\vtheta}_{\text{WCC}}, \widehat{\vomega})
\widetilde{\varphi}_{\infty,i}^{\top}(\widehat{\vtheta}_{\text{WCC}}, \widehat{\vomega})
\right], \\[0.5em]
\widehat{\widetilde{\Sigma}}_{\vtheta,\vDelta}
 &:=
\frac{1}{N}\sum^N_{i=1} \left[
\widetilde{\varphi}_{\infty,i}(\widehat{\vtheta}_{\text{WCC}}, \widehat{\vomega})
\widetilde{\varphi}_{\vDelta,i}^{\top}(\widehat{\vgamma}, \widehat{\vomega})
\right], \\[0.5em]
\widehat{\widetilde{\Sigma}}_{\vDelta}
&:=
\frac{1}{N}\sum^N_{i=1} \left[
\widetilde{\varphi}_{\vDelta,i}(\widehat{\vgamma}, \widehat{\vomega})
\widetilde{\varphi}_{\vDelta,i}^{\top}(\widehat{\vgamma}, \widehat{\vomega})
\right].
\end{aligned}
\end{equation}
Note that the difference between Equation (\ref{eq:sigma_estimate_components}) and Equation (\ref{eq:sigma_estimate_components_no_pi}) is that Equation~(\ref{eq:sigma_estimate_components_no_pi}) uses only the leading term in the influence functions, i.e.,  \(\widetilde{\varphi}_{\infty,i}\) and
\(\widetilde{\varphi}_{\vDelta,i}\), thereby ignoring the uncertainty introduced by $\widehat{\vomega}$. Such a variance-covariance estimate is computationally simpler yet more conservative~\citep{sun2018inverse}. Similarly, we can obtain a consistent estimate of the optimal tuning parameter $\widehat{\Wb}_\opt = \widehat{\Sigma}_{\vtheta, \vDelta}\widehat{\Sigma}^{-1}_{\vDelta}$.

\subsection{Computing PS-PPI with existing statistical software}
\label{sec:incorporate_software}
When there exists an off-the-shelf package of the target Z-estimation problem for weighted complete-case analysis, the estimation and uncertainty quantification of $\widehat{\vtheta}_{\text{PS-PPI}}$ can greatly benefit from such a package if the uncertainty from the estimated missing data model is ignored. An algorithm block of how the existing package can be incorporated in the PS-PPI procedure is presented in Algorithm~\ref{algo:general_missing_pattern}. Specifically, to obtain the PS-PPI point estimate, Equations (\ref{eq:general_missing_pattern_ee}), (\ref{eq:general_missing_pattern_ee_gamma_1}), and (\ref{eq:general_missing_pattern_ee_gamma_2}) can be viewed as separate weighted complete-case analyses on different data using the same Z-estimation framework. Hence, no significant additional efforts need to be made to compute $\widehat{\vtheta}_{\text{WCC}}$, $\widehat{\vgamma}_{\cc,k}$, $\widehat{\vgamma}_{\ic,k}$ once weighted estimation software is available. The estimation of $\widehat{\Wb}_\opt$ involves variance estimation, which we will introduce in the following. For uncertainty quantification, the conservative variance estimates in Equation (\ref{eq:sigma_estimate_components_no_pi}) can be computed by using existing software and resampling strategies such as Bootstrap or Jackknife. Note that unless otherwise noted, in this subsection we only talk about conservative variance estimation that ignores the uncertainty from the missing data model. In Algorithm~\ref{algo:general_missing_pattern}, we apply the delete-$1$ Jackknife as an example. Before going into details, we would like to introduce new notation for the ease of understanding. Specifically, we can rewrite $\widehat{\widetilde{\Sigma}}_{\vtheta, \vDelta} = \left ( \widehat{\widetilde{\Sigma}}_{\vtheta, \vDelta_1}, \dots, \widehat{\widetilde{\Sigma}}_{\vtheta, \vDelta_K} \right )$, where the $k$-th element
\begin{equation}
    \label{eq:sigma_theta_delta_k}
    \widehat{\widetilde{\Sigma}}_{\vtheta, \vDelta_k} = \frac{1}{N} \sum^N_{i=1} \left [ \widetilde{\varphi}_{\infty, i}(\widehat{\vtheta}_{\text{WCC}}, \widehat{\vomega}) \widetilde{\varphi}^\top_{\vDelta_k, i}(\widehat{\vgamma}_k, \widehat{\vomega}) \right ]
\end{equation}
is a conservative estimate of the asymptotic covariance matrix between $\widehat{\vtheta}_{\text{WCC}}$ and $\widehat{\vDelta}_k$. Similarly, we can rewrite $\widehat{\widetilde{\Sigma}}_{\vDelta}$ as
\begin{equation}
\label{eq:sigma_delta_block_form}
\widehat{\widetilde{\Sigma}}_\vDelta
=
\begin{pmatrix}
\widehat{\widetilde{\Sigma}}_{\vDelta_1,\vDelta_1} & \widehat{\widetilde{\Sigma}}_{\vDelta_1,\vDelta_2} & \cdots & \widehat{\widetilde{\Sigma}}_{\vDelta_1,\vDelta_K} \\
\widehat{\widetilde{\Sigma}}_{\vDelta_2,\vDelta_1} & \widehat{\widetilde{\Sigma}}_{\vDelta_2,\vDelta_2} & \cdots & \widehat{\widetilde{\Sigma}}_{\vDelta_2,\vDelta_K} \\
\vdots & \vdots & \ddots & \vdots \\
\widehat{\widetilde{\Sigma}}_{\vDelta_K,\vDelta_1} & \widehat{\widetilde{\Sigma}}_{\vDelta_K,\vDelta_2} & \cdots & \widehat{\widetilde{\Sigma}}_{\vDelta_K,\vDelta_K}
\end{pmatrix}.
\end{equation}
Here, the $(k,\ell)$-th element is given by
\begin{equation}
\label{eq:sigma_delta_kl_block}
\widehat{\widetilde{\Sigma}}_{\vDelta_k,\vDelta_\ell}
=
\frac{1}{N} \sum^N_{i=1} \left[
\widetilde{\varphi}_{\vDelta_k,i}(\widehat{\vgamma}_k,\widehat{\vomega})
\widetilde{\varphi}_{\vDelta_\ell,i}^{\top}(\widehat{\vgamma}_\ell,\widehat{\vomega})
\right],
\end{equation}
which is an estimate of the asymptotic covariance matrix between $\widehat{\vDelta}_k$ and $\widehat{\vDelta}_\ell$. Recall that $\widehat{\vDelta}_k = \widehat{\vgamma}_{\cc, k} - \widehat{\vgamma}_{\ic, k}$, which suggests that all $ \widehat{\vgamma}_{\cc, 1}, \dots, \widehat{\vgamma}_{\cc, K}$ and $\widehat{\vtheta}_{\text{WCC}}$ are estimated using the fully-observed subset, while $\widehat{\vgamma}_{\ic, k}$, $k=1,
\dots, K$ are estimated using the $\cC_i = k$ subset of data, respectively. Since each subset is mutually exclusive with others and therefore independent of other subsets, we proved in Section~\ref{sec:incorporate_software_proof} that: 1)
$\widehat{\widetilde{\Sigma}}_{\vtheta, \vDelta_k}$ is equivalent to $\widehat{\widetilde{\Sigma}}_{\vtheta, \vgamma_{\cc, k}}$, which is an estimate of the asymptotic covariance matrix between $\widehat{\vtheta}_{\text{WCC}}$ and $\widehat{\vgamma}_{\cc, k}$; 2) each non-diagonal element $\widehat{\widetilde{\Sigma}}_{\vDelta_k,\vDelta_\ell}$ in Equation (\ref{eq:sigma_delta_block_form}) is equivalent to $\widehat{\widetilde{\Sigma}}_{\vgamma_{\cc, k}, \vgamma_{\cc, \ell}}$, which is an estimate of the asymptotic covariance between $\widehat{\vgamma}_{\cc, k}$ and $\widehat{\vgamma}_{\cc, \ell}$; 3) each diagonal element $\widehat{\widetilde{\Sigma}}_{\vDelta_k,\vDelta_k}$ in Equation (\ref{eq:sigma_delta_block_form}) is equivalent to $\widehat{\widetilde{\Sigma}}_{\vgamma_{\cc, k}} + \widehat{\widetilde{\Sigma}}_{\vgamma_{\ic, k}}$, where $\widehat{\widetilde{\Sigma}}_{\vgamma_{\cc, k}}$ and $\widehat{\widetilde{\Sigma}}_{\vgamma_{\ic, k}}$ are estimates of the asymptotic variance-covariance matrices of $\widehat{\vgamma}_{\cc, k}$ and $\widehat{\vgamma}_{\ic, k}$, respectively. With existing weighted analysis statistical software that can produce variance estimates available, one can easily obtain a conservative estimate of $\Sigma_{\vtheta}$, $\Sigma_{\vgamma_{\cc, k}}$, and $\Sigma_{\vgamma_{\ic, k}}$. In terms of estimating the covariance matrices $\Sigma_{\vtheta, \vgamma_{\cc, k}}$ and  $\Sigma_{\vgamma_{\cc, k}, \vgamma_{\cc, \ell}}$, notice that both of them are covariance matrices between two estimates from the fully-observed dataset; as such, their estimates can be obtained by applying the resampling strategies on the fully-observed dataset, as detailed in Algorithm~\ref{algo:general_missing_pattern}.

\begin{algorithm}[htbp]
\caption{Computing $\widehat{\vtheta}_{\text{PS-PPI}}$ and its conservative variance estimate with existing weighted complete-case software.}
\label{algo:general_missing_pattern}
\begin{algorithmic}[1]
\Require An observed dataset $\cD_{\obs} = \{(G_{\cC_i}(\vY_i), \cC_i)\}^N_{i=1}$ with $K$ missingness patterns ($\cC_i = 1, \ldots, K$) and fully observed subjects ($\cC_i = \infty$); prediction functions $h_1, \ldots, h_K$; an estimated missing data model $\bpi(\cdot; \widehat{\vomega})$; and a weighted Z-estimation routine $f(\texttt{data}; \texttt{weights})$ that returns the parameter estimate and its estimated variance-covariance matrix.
\For{$k = 1, \ldots, K$}
    \State Impute the missing entries of the subjects with $\cC_i = k$ using $h_k$ to obtain $\cD_{\ic, k} = \{\vYimp_{k, i}: \cC_i = k\}$.
    \State $\widehat{\vgamma}_{\ic, k}, \widehat{\widetilde{\Sigma}}_{\vgamma_{\ic, k}} \gets f\left(\cD_{\ic, k}; 1/\pi_k(G_k(\vY_i); \widehat{\vomega})\right)$.
\EndFor
\State Let $\cD_{\cc} = \{\vY_i: \cC_i = \infty\}$. For each $k$, mask the entries of $\cD_{\cc}$ that are missing under the $k$-th pattern and impute them using $h_k$ to obtain $\cD_{\cc, k} = \{\vYimp_{k, i}: \cC_i = \infty\}$.
\State $\widehat{\vtheta}_{\wcc} \gets f\left(\cD_{\cc}; 1/\pi_\infty(\vY_i; \widehat{\vomega})\right)$ and $\widehat{\vgamma}_{\cc, k} \gets f\left(\cD_{\cc, k}; 1/\pi_\infty(\vY_i; \widehat{\vomega})\right)$ for $k = 1, \ldots, K$.
\For{$j = 1, \ldots, |\cD_{\cc}|$} \Comment{\texttt{Jackknife}}
    \State Remove the $j$-th fully observed subject and recompute $\widehat{\vtheta}^{(-j)}_{\wcc}$ and $\widehat{\vgamma}^{(-j)}_{\cc, k}$, $k = 1, \ldots, K$, with the same weights.
\EndFor
\State Compute the jackknife estimates $\widehat{\widetilde{\Sigma}}_{\vtheta}$, $\widehat{\widetilde{\Sigma}}_{\vtheta, \vgamma_{\cc, k}}$, and $\widehat{\widetilde{\Sigma}}_{\vgamma_{\cc, k}, \vgamma_{\cc, \ell}}$ for $k, \ell = 1, \ldots, K$, and write $\widehat{\widetilde{\Sigma}}_{\vgamma_{\cc, k}} = \widehat{\widetilde{\Sigma}}_{\vgamma_{\cc, k}, \vgamma_{\cc, k}}$.
\State Following Section~\ref{sec:incorporate_software}, set $\widehat{\widetilde{\Sigma}}_{\vtheta, \vDelta} = (\widehat{\widetilde{\Sigma}}_{\vtheta, \vgamma_{\cc, 1}}, \ldots, \widehat{\widetilde{\Sigma}}_{\vtheta, \vgamma_{\cc, K}})$, and let $\widehat{\widetilde{\Sigma}}_{\vDelta}$ have off-diagonal blocks $\widehat{\widetilde{\Sigma}}_{\vgamma_{\cc, k}, \vgamma_{\cc, \ell}}$ ($k \neq \ell$) and diagonal blocks $\widehat{\widetilde{\Sigma}}_{\vgamma_{\cc, k}} + \widehat{\widetilde{\Sigma}}_{\vgamma_{\ic, k}}$.
\State $\widehat{\Wb}_\opt = (\widehat{\Wb}_{\opt, 1}, \ldots, \widehat{\Wb}_{\opt, K}) \gets \widehat{\widetilde{\Sigma}}_{\vtheta, \vDelta} \widehat{\widetilde{\Sigma}}_{\vDelta}^{-1}$.
\State $\widehat{\vtheta}_{\text{PS-PPI}} \gets \widehat{\vtheta}_{\wcc} - \sum_{k=1}^K \widehat{\Wb}_{\opt, k} \left( \widehat{\vgamma}_{\cc, k} - \widehat{\vgamma}_{\ic, k} \right)$.
\State $\widehat{\Sigma}_{\text{PS-PPI}} \gets \widehat{\widetilde{\Sigma}}_{\vtheta} - \widehat{\widetilde{\Sigma}}_{\vtheta, \vDelta} \widehat{\widetilde{\Sigma}}_{\vDelta}^{-1} \widehat{\widetilde{\Sigma}}_{\vtheta, \vDelta}^\top$, the plug-in version of Equation~(\ref{eq:optimal_asymptotic_variance}).
\Ensure $\widehat{\vtheta}_{\text{PS-PPI}}$, $\widehat{\Sigma}_{\text{PS-PPI}}$, and the corresponding 95\% confidence intervals.
\end{algorithmic}
\end{algorithm}

\newpage

\section{Supplementary Simulation}

\subsection{Robustness against missing data model misspecification}
\label{sec:misspecification_robustness}

In this subsection, we investigate the performance of PS-PPI when the missing data model is misspecified. Specifically, we retain the simulation settings introduced in Section~\ref{sec:experimental-setup} except that we set the sample size to $N = 10{,}000$ and modify the missing data model that produces missingness, as well as the working missing data model used by PS-PPI. In particular, we define the following missing data model:
\begin{equation} \label{eq:ps_model_misspecification}
\textrm{logit}\{\pi(\cC_i| \vY_i )\} =
\begin{cases}
  -2 + 0.1X_{i2} + 0.1Z_{i1} + 0.1X_{i1}X_{i2}, & \cC_i = 1, \\
  \eta_2(Y_i, X_{i1}, Z_{i2}), & \cC_i = 2, \\
  -2 + 0.1X_{i2} + 0.2Z_{i1}, & \cC_i = 3,
\end{cases}
\end{equation}
where
\begin{equation*}
\begin{aligned}
\eta_2(y, x_1, z_2) = {} & -0.5 - 0.1y + 4x_1 - 0.5z_2 \\
& - 0.3y^2 - 4x_1^2 - 12.5z_2^2 + 0.7x_1y - 3yz_2 + 4x_1z_2,
\end{aligned}
\end{equation*}
and $\pi(\cC_i = \infty| \vY_i) = 1 - \sum^3_{k=1} \pi(\cC_i = k|\vY_i)$. The above missing data model satisfies the MAR assumption, i.e., the probability of $\cC_i = k$ depends only on variables observed in that pattern. For the working missing data model used to estimate missing data weights for PS-PPI, we include only $Z_1$ in the missing data model to make it misspecified.

In addition to PS-PPI, CCA, WCCA, PPI\texttt{++}, and MI are included as baselines. We iterate either $\sigma_{\text{pred}}$ or $\lambda_{\text{pred}}$ from $\{0, 0.05, 0.1, 0.2, 0.5, 1.0, 1.5, 2.0\}$. The sample size is $N = 10{,}000$, and the simulation is repeated 500 times. The simulation results are presented in Figures~\ref{fig:varying_prediction_noise_misspecified_ps} and~\ref{fig:varying_prediction_bias_misspecified_ps}. Note that WCCA and PS-PPI are the only two methods that are impacted by the missing data model specification. Both WCCA and PS-PPI fail to achieve the 95\% nominal coverage rates. In particular, the coverage rates of PS-PPI collapse when lower-quality predictions are used to impute the missing entries.

\begin{figure}[H]
    \centering
    \includegraphics[width=\linewidth]{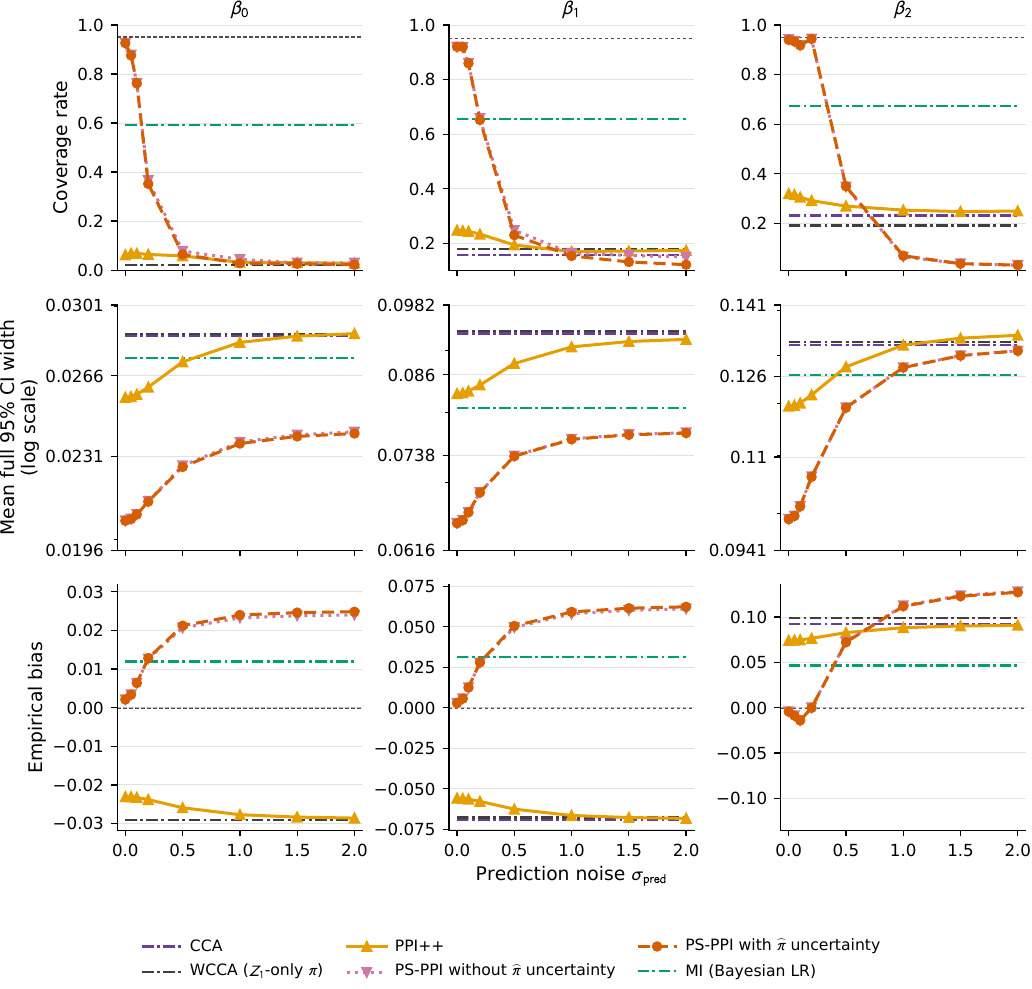}
    \caption{\small Performance comparison under the misspecified missing data model of Section~\ref{sec:misspecification_robustness}. Rows report the coverage rates of the 95\% confidence intervals (dotted lines indicate 95\%), the widths of the 95\% confidence intervals (log scale), and the empirical biases. The X-axis in all the panels represents the prediction noise level used in the simulations. The prediction bias level is fixed at $\lambda_{\text{pred}} = 0$.}
    \label{fig:varying_prediction_noise_misspecified_ps}
\end{figure}

\begin{figure}[H]
    \centering
    \includegraphics[width=\linewidth]{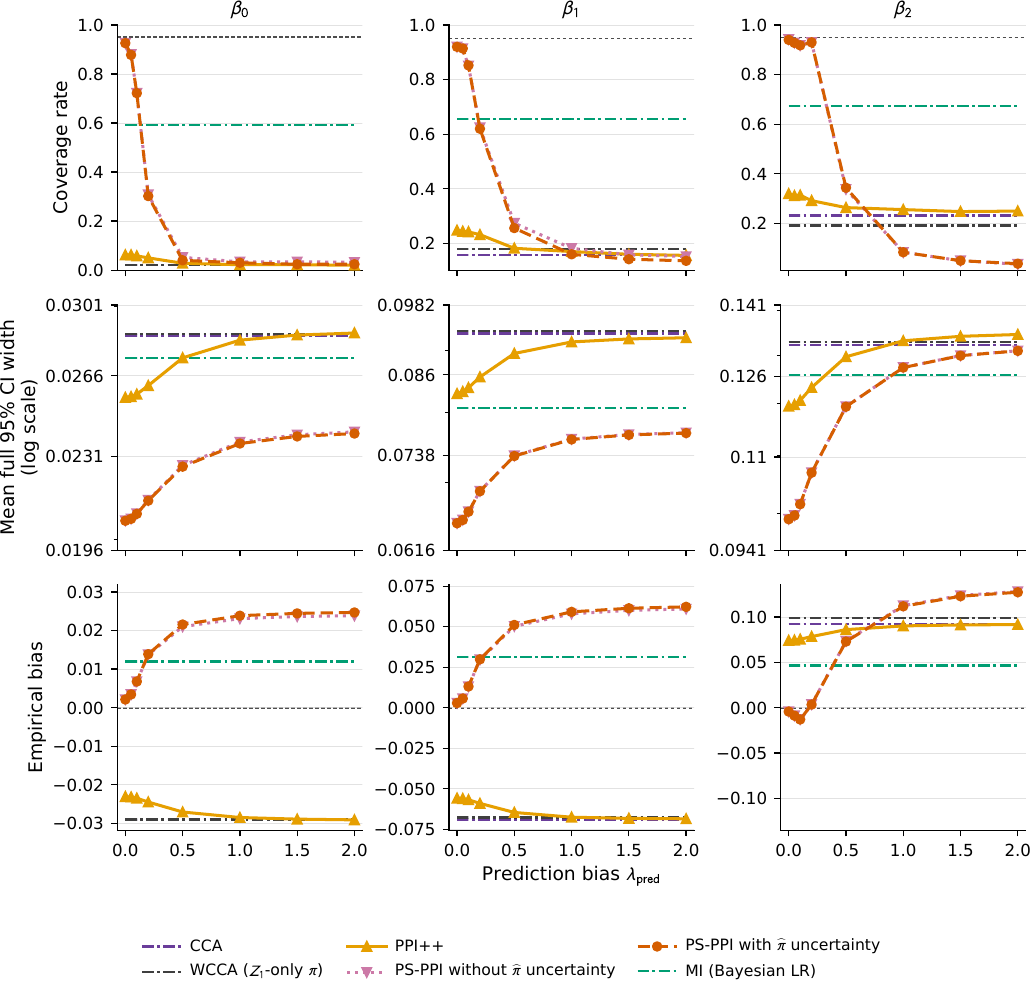}
    \caption{\small The same as Figure~\ref{fig:varying_prediction_noise_misspecified_ps}, except that the prediction bias level $\lambda_{\text{pred}}$ varies and the prediction noise level is fixed at $\sigma_{\text{pred}} = 0$.}
    \label{fig:varying_prediction_bias_misspecified_ps}
\end{figure}

\subsection{Finite-sample behavior of PS-PPI}
\label{sec:finite-sample-psppi}
In this subsection, we empirically examine the finite-sample behavior of PS-PPI. We adopt the setting of the main simulation in Section~\ref{sec:experiments}, varying only the total sample size over $N \in \{50, 100, 200, 300, 400, 500, 1000, 2000\}$.
We compare WCCA and PS-PPI with and without accounting for the uncertainty in the estimated missing data model, using the same metrics as in Section~\ref{sec:experiments}. The results are presented in Figure~\ref{fig:varying_size_N}.

At small sample sizes, all three methods exhibit undercoverage relative to the nominal 95\% level, which is more pronounced for the two PS-PPI variants. As the sample size increases, the empirical coverage rates of all methods approach the nominal level for each of the three coefficients. Specifically, WCCA achieves approximately 95\% coverage for all three coefficients when $N>200$, whereas the coverage rates of the two PS-PPI variants approach 95\% when $N>500$. These results suggest that, in our simulation setting, PS-PPI may require a larger sample size than WCCA to achieve accurate finite-sample coverage. Meanwhile, the average widths of the 95\% CIs decrease as $N$ increases, and the estimation biases of all methods approach zero. These findings are consistent with the consistency and asymptotic normality of WCCA and PS-PPI.

\begin{figure}
    \centering
    \includegraphics[width=1.0\linewidth]{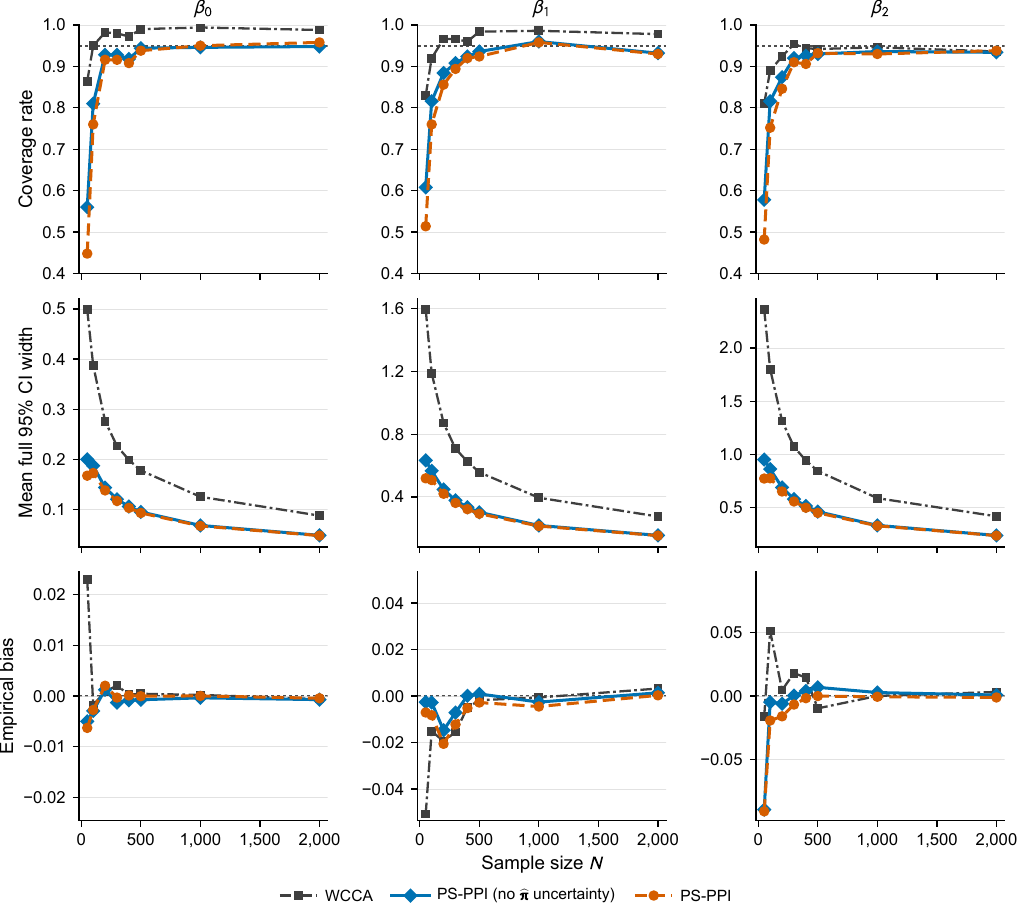}
    \caption{Performance comparison among selected methods. Rows report the coverage rates of the 95\% confidence intervals (dotted lines indicate 95\%), the widths of the 95\% confidence intervals, and the estimation biases; columns correspond to $\beta_0$, $\beta_1$, and $\beta_2$. The X-axis in all the panels represents the sample size $N$ of the simulation dataset.}
    \label{fig:varying_size_N}
\end{figure}

\subsection{PS-PPI with imputation from ML models trained on external data}
\label{sec:supp-pred-ml}

In this subsection, we investigate the performance of PS-PPI when the predictions $\widehat{Y}_i$, $\widehat{X}_{i1}$, and $\widehat{X}_{i2}$ are produced by explicitly fitted ML models rather than constructed by adding noise and bias to the true values. Specifically, we keep the simulation setting introduced in Section~\ref{sec:experiments}, except that we change how the predictions are obtained. A separate dataset of size $n_{\text{ML}} = 10{,}000$ is drawn from the same data generating process independently, and for each of $Y_i$, $X_{i1}$, and $X_{i2}$, an ML model is fitted using the training data to learn to predict it. We consider five methods in this supplementary simulation to predict the missing values: 1) linear regression; 2) random forest; 3) gradient boosting; 4) a multilayer perceptron (MLP); and 5) $k$-nearest neighbors (KNN). When fitting the predictive models, we use $0.1\exp(Z_{i1})+\sigma_{\mathrm{ML}}e_{i1}$ and $\sin(Z_{i2})+\sigma_{\mathrm{ML}}e_{i2}$ as their only predictors, where $e_{i1}$ and $e_{i2}$ are independent standard normal noises. We set $\sigma_{\mathrm{ML}}\in\{0,0.25,0.5,0.75\}$ to investigate the impact of the quality of the ML predictions on different models. For each learner and noise level, the models are fitted once and held fixed across all 500 replicates. The simulation results are presented in Figure~\ref{fig:prediction_model_z_exp_sin_noise_grouped_bar_N_5000}. It can be seen that, first, across different ML models and different noise levels, PS-PPI has coverage rates close to the 95\% nominal level across the three coefficients. This empirically verifies the asymptotic property of PS-PPI. Second, when $\sigma_{\text{ML}}$ increases, the width of the 95\% CI for $\beta_2$ increases no matter which ML model is used, corroborating the conclusion from Section~\ref{sec:experiments} that better quality of predictions leads to greater efficiency improvement. In Table~\ref{tab:pred-ml-z-exp-sin-r2-by-noise}, the predictive performance of each ML model under different $\sigma_{\text{ML}}$ is presented. It can be seen that $R^2_{X_2}$ decreases as $\sigma_{\text{ML}}$ increases. However, even when $\sigma_{\text{ML}} = 0$, $R^2_Y$ stays below 0.2 and $R^2_{X_1}$ is close to 0. This indicates that constructing ML predictions by adding noise and bias gives better control over prediction quality when investigating its impact on the performance of the proposed method.

\subsection{PS-PPI with ML-based missing data models fitted from external data}
\label{sec:supp-ps-ml}

In this subsection, we examine the performance of PS-PPI when the missing data weights are obtained from a missing data model fitted by a machine learning classifier on external data. The setting of this simulation is the same as that of the main simulation discussed in Section~\ref{sec:experiments}, with the prediction noise and bias fixed at $\sigma_{\text{pred}} = 0$ and $\lambda_{\text{pred}} = 0.1$, except for the missing data model. We replace the missing data model defined in Equations~(\ref{eq:ps_model_1}) and~(\ref{eq:ps_model_2}) with the following model, which depends only on the fully observed variables $Z_{i1}$ and $Z_{i2}$:
\begin{equation} \label{eq:ps_model_z_only}
\log \frac{\pi(\cC_i = k \mid \vY_i)}{\pi(\cC_i = \infty \mid \vY_i)} =
\begin{cases}
  0.2\exp(Z_{i1}) + 0.5Z_{i2}, & k = 1, \\
  -0.4 - 0.4Z_{i1}, & k = 2, \\
  0.5\exp(Z_{i1}) + 0.8Z_{i1}\sin(Z_{i2}), & k = 3.
\end{cases}
\end{equation}
The missing data weights for WCCA and PS-PPI are then estimated by an ML model using $Z_{i1}$ and $Z_{i2}$. To fit the ML-based missing data model using external data, we sample, before the simulation, an external training set of size $n_{\text{PS}} \in \{500, 1000, 2000, 5000\}$ from the same data generating process, and then use it to fit a classifier that learns to predict the missingness index $\cC_i$. Once the model has been fitted, it is reused across the 500 replicates in the simulation. We consider the following different ML models in this supplementary simulation: 1) random forest; 2) gradient boosting; 3) multi-layer perceptron (MLP); 4) $k$-nearest neighbors (KNN). We compare the performance of WCCA and PS-PPI using weights from each of the above methods. Note that we ignore the uncertainty from the missing data model for PS-PPI in this supplementary simulation. The results are presented in Table~\ref{tab:ps-ml}.

It can be seen that, first, the 95\% CIs produced by both PS-PPI and WCCA attain the 95\% nominal coverage rate across most of the settings, except that in some settings with random forest, gradient boosting, and KNN, we observe coverage rates around 0.92--0.93, slightly below the nominal 95\% level. Second, the size of the external training set for the missing data model will impact the efficiency of PS-PPI. Specifically, when the size of the training set increases, the widths of the 95\% CIs produced by PS-PPI remain mostly unchanged when random forest, MLP, or KNN is used; when gradient boosting is used to fit the missing data model, a larger training set improves the efficiency of the PS-PPI method. Finally, across all the settings, PS-PPI outperforms WCCA in terms of efficiency; namely, it achieves narrower confidence intervals than WCCA.

Overall, the supplementary simulation above provides the following insights. First, the results reveal a caveat of applying ML models fitted on external data to obtain missing data weights. Although fitting an ML model using external data can be a convenient choice for obtaining estimated missing data weights, such an approach does not provide a theoretical guarantee for the validity of the resulting statistical inference; as described in Theorem~\ref{thrm:general_missing_patterns}, estimated missing data weights should be obtained from a correctly specified and correctly estimated missing data model fitted on the analysis data itself. As a result, a poorly fitted or misspecified ML-based missing data model may lead to inefficient or biased inference in the downstream analysis. On the other hand, the results also empirically indicate that the efficiency gain over WCCA persists when the weights come from an ML-based missing data model fitted on external data. This indicates that fitting ML-based missing data models still allows users to gain efficiency improvement by using PS-PPI.

\newpage

\section{Descriptive Analysis of the \textit{All of Us} Data}
\label{sec:descriptive_aou}

\begin{table}
\centering
\caption{Summary statistics for different subsets of the \textit{All of Us} analytic sample used in this study: 1) complete cases (CC); 2) CC + missing outcome only data; 3) CC + data with missingness patterns of size $\geq$ 10K; 4) all data. For C-reactive protein (CRP), mean, standard deviation, and range are reported. For other variables, the values represent percentages; counts are shown in parentheses. “Missing” rows show counts only. Missingness in smoking status and education levels is due to responses of ``Skip'', ``Don’t know'', or ``Prefer not to answer''; missingness in race is due to responses of ``Skip'', ``None of these'', or ``Prefer not to answer''. Note that age is treated as a continuous variable in the regression model.}
\label{tab:summary_aou}
\resizebox{1.0\textwidth}{!}{
\begin{tabular}{lcccc}
\toprule
 & \textbf{CC} & \textbf{CC+missing CRP} & \textbf{CC+patterns of size $\geq$ 10K} & \textbf{Total} \\
\midrule
\textbf{Sample Size} & 47{,}317 & 220{,}002 & 280{,}976 & 294{,}347 \\
\addlinespace[0.5em]
\textbf{C-reactive protein (CRP; mg/L)} & & & & \\
Mean (SD) & 3.9 (5.0) & 3.9 (5.0) & 4.1 (5.1) & 4.1 (5.1) \\
Range & 0.0--26.1 & 0.0--26.1 & 0.0--26.1 & 0.0--26.1 \\
\addlinespace[0.5em]
\textbf{College graduate or higher} & 50.1 (23{,}729) & 51.6 (113{,}617) & 46.5 (130{,}688) & 46.3 (133{,}523) \\
Missing & (0) & (0) & (0) & (6{,}072) \\
\addlinespace[0.5em]
\textbf{Age group} & & & & \\
18--24 & 0.4 (170) & 0.7 (1{,}464) & 0.8 (2{,}312) & 0.8 (2{,}489) \\
25--34 & 4.9 (2{,}332) & 9.4 (20{,}575) & 10.9 (30{,}497) & 10.9 (31{,}998) \\
35--44 & 9.7 (4{,}583) & 13.8 (30{,}386) & 14.9 (41{,}829) & 14.9 (43{,}882) \\
45--54 & 12.1 (5{,}734) & 13.3 (29{,}193) & 14.3 (40{,}168) & 14.3 (42{,}150) \\
55--64 & 20.0 (9{,}468) & 19.4 (42{,}609) & 19.4 (54{,}384) & 19.5 (57{,}277) \\
65+ & 52.9 (25{,}030) & 43.5 (95{,}775) & 39.8 (111{,}786) & 39.6 (116{,}551) \\
\addlinespace[0.5em]
\textbf{Sex at birth} & & & & \\
Female & 61.5 (29{,}102) & 59.6 (131{,}070) & 60.9 (171{,}010) & 60.6 (178{,}463) \\
Male & 38.5 (18{,}215) & 40.4 (88{,}932) & 39.1 (109{,}966) & 39.4 (115{,}884) \\
\addlinespace[0.5em]
\textbf{Race} & & & & \\
Asian & 2.3 (1{,}083) & 3.9 (8{,}666) & 3.9 (8{,}666) & 3.9 (8{,}995) \\
White & 78.8 (37{,}273) & 72.7 (159{,}848) & 72.7 (159{,}848) & 71.8 (164{,}883) \\
Black & 16.6 (7{,}862) & 20.9 (45{,}993) & 20.9 (45{,}993) & 21.7 (49{,}708) \\
Others & 2.3 (1{,}099) & 2.5 (5{,}495) & 2.5 (5{,}495) & 2.6 (5{,}968) \\
Missing & (0) & (0) & (60{,}974) & (64{,}793) \\
\addlinespace[0.5em]
\textbf{Have ever smoked $\geq$100 cigarettes} & 45.4 (21{,}463) & 42.1 (92{,}627) & 40.3 (113{,}103) & 40.3 (115{,}628) \\
Missing & (0) & (0) & (0) & (7{,}731) \\
\addlinespace[0.5em]
\textbf{Depression} & 42.1 (19{,}932) & 26.5 (58{,}217) & 26.2 (73{,}746) & 26.2 (77{,}196) \\
\addlinespace[0.5em]
\textbf{Hypertension} & 63.2 (29{,}909) & 42.4 (93{,}384) & 41.3 (115{,}996) & 41.2 (121{,}409) \\
\addlinespace[0.5em]
\textbf{Diabetes} & 30.6 (14{,}467) & 18.5 (40{,}739) & 19.5 (54{,}732) & 19.5 (57{,}469) \\
\addlinespace[0.5em]
\textbf{Anxiety} & 45.6 (21{,}596) & 29.0 (63{,}872) & 28.6 (80{,}450) & 28.6 (84{,}152) \\
\addlinespace[0.5em]
\textbf{Obesity} & 45.0 (21{,}289) & 28.8 (63{,}268) & 29.5 (82{,}849) & 29.4 (86{,}515) \\
\bottomrule
\end{tabular}}
\end{table}

We present the summary statistics of included variables and several key conditions in the \textit{All of Us} data in Table~\ref{tab:summary_aou}. Specifically, we compute the summary statistics for 1) the complete-case subset; 2) the subset consisting of complete cases and outcome-missing-only cases; 3) the subset of complete cases and all missingness patterns with more than 10,000 samples; and 4) the full dataset. The characteristics of several variables, including both the outcome and the covariate of interest, differ across these subsets. Specifically, in the subset with only complete cases and outcome-missing cases, the sample mean of the CRP level (in mg/L) is 3.9, compared with 4.1 in the full data. Similarly, the proportion of participants with a college degree or higher is 50.1\% and 51.6\% for the complete-case and complete-case plus outcome-missing subsets, respectively, both higher than that in the full dataset (46.3\%). These discrepancies indicate that CCA and outcome-based prediction-based inference methods rely on subsets with distributions shifted from the overall population, thereby making the result less generalizable to the target population. Even if a method that uses the cases with complete or outcome-missing-only data can be modified to account for such distributional differences using propensity scores for outcome missingness, participants with other missingness patterns and the associated information are not used. In contrast, for the subset including complete cases and all patterns with more than 10,000 samples, the sample mean of CRP (4.1) matches that of the full data, and the proportion of college graduates (46.5\%) closely approximates the full-sample value (46.3\%). This demonstrates that the proposed PS-PPI method enables inference based on data subsets that are more representative of the full data distribution, thereby improving the generalizability of the inferential results under data missingness.

\begin{table}[htbp]
\centering
\caption{Results for WCCA and PS-PPI with missing data weights estimated by ML models fitted to external data, by ML model and training size. Coverage is the empirical coverage rate of the 95\% confidence intervals; width is their mean width.}
\label{tab:ps-ml}
\resizebox{\textwidth}{!}{%
\begin{tabular}{lcccccccccc}
\toprule
Method & training size & \multicolumn{3}{c}{$\beta_0$} & \multicolumn{3}{c}{$\beta_1$} & \multicolumn{3}{c}{$\beta_2$} \\
\cmidrule(lr){3-5}\cmidrule(lr){6-8}\cmidrule(lr){9-11}
 & & Coverage & Width & Bias & Coverage & Width & Bias & Coverage & Width & Bias \\
\midrule
\multicolumn{11}{c}{\textbf{Random forest}} \\
\midrule
WCCA & 500 & 0.946 & 0.064 & 0.001 & 0.948 & 0.206 & 0.000 & 0.934 & 0.316 & -0.006 \\
WCCA & 1000 & 0.948 & 0.064 & 0.001 & 0.952 & 0.202 & 0.001 & 0.942 & 0.308 & -0.007 \\
WCCA & 2000 & 0.954 & 0.064 & 0.001 & 0.958 & 0.201 & 0.000 & 0.948 & 0.307 & -0.008 \\
WCCA & 5000 & 0.952 & 0.064 & 0.001 & 0.956 & 0.202 & 0.000 & 0.954 & 0.311 & -0.008 \\
\addlinespace
PS-PPI & 500 & 0.934 & 0.032 & 0.000 & 0.944 & 0.104 & 0.002 & 0.940 & 0.156 & -0.002 \\
PS-PPI & 1000 & 0.940 & 0.031 & 0.000 & 0.954 & 0.103 & 0.002 & 0.950 & 0.154 & -0.003 \\
PS-PPI & 2000 & 0.932 & 0.031 & 0.000 & 0.954 & 0.102 & 0.002 & 0.946 & 0.154 & -0.003 \\
PS-PPI & 5000 & 0.932 & 0.031 & 0.001 & 0.954 & 0.102 & 0.001 & 0.950 & 0.154 & -0.003 \\
\midrule
\multicolumn{11}{c}{\textbf{Gradient boosting}} \\
\midrule
WCCA & 500 & 0.930 & 0.093 & 0.001 & 0.938 & 0.305 & -0.004 & 0.942 & 0.434 & -0.004 \\
WCCA & 1000 & 0.942 & 0.106 & 0.001 & 0.956 & 0.342 & 0.007 & 0.930 & 0.488 & 0.001 \\
WCCA & 2000 & 0.948 & 0.092 & 0.002 & 0.942 & 0.287 & 0.001 & 0.946 & 0.462 & -0.006 \\
WCCA & 5000 & 0.948 & 0.096 & 0.001 & 0.938 & 0.311 & 0.001 & 0.936 & 0.507 & -0.005 \\
\addlinespace
PS-PPI & 500 & 0.928 & 0.052 & 0.000 & 0.932 & 0.175 & 0.000 & 0.934 & 0.237 & 0.001 \\
PS-PPI & 1000 & 0.936 & 0.045 & 0.001 & 0.940 & 0.147 & 0.004 & 0.940 & 0.211 & 0.000 \\
PS-PPI & 2000 & 0.948 & 0.040 & 0.000 & 0.950 & 0.131 & 0.002 & 0.960 & 0.198 & -0.002 \\
PS-PPI & 5000 & 0.930 & 0.040 & 0.001 & 0.950 & 0.132 & 0.001 & 0.928 & 0.199 & -0.001 \\
\midrule
\multicolumn{11}{c}{\textbf{MLP}} \\
\midrule
WCCA & 500 & 0.946 & 0.063 & 0.000 & 0.958 & 0.200 & 0.000 & 0.940 & 0.307 & -0.008 \\
WCCA & 1000 & 0.944 & 0.063 & 0.000 & 0.962 & 0.199 & 0.000 & 0.946 & 0.305 & -0.008 \\
WCCA & 2000 & 0.944 & 0.063 & 0.000 & 0.962 & 0.198 & 0.000 & 0.950 & 0.303 & -0.008 \\
WCCA & 5000 & 0.948 & 0.063 & 0.000 & 0.956 & 0.198 & 0.000 & 0.944 & 0.303 & -0.009 \\
\addlinespace
PS-PPI & 500 & 0.942 & 0.031 & 0.001 & 0.954 & 0.102 & 0.002 & 0.954 & 0.154 & -0.001 \\
PS-PPI & 1000 & 0.940 & 0.031 & 0.000 & 0.954 & 0.102 & 0.002 & 0.950 & 0.153 & -0.001 \\
PS-PPI & 2000 & 0.942 & 0.031 & 0.001 & 0.950 & 0.102 & 0.002 & 0.942 & 0.153 & -0.001 \\
PS-PPI & 5000 & 0.938 & 0.031 & 0.001 & 0.950 & 0.101 & 0.001 & 0.944 & 0.153 & -0.002 \\
\midrule
\multicolumn{11}{c}{\textbf{KNN}} \\
\midrule
WCCA & 500 & 0.940 & 0.065 & 0.001 & 0.956 & 0.209 & 0.000 & 0.932 & 0.320 & -0.007 \\
WCCA & 1000 & 0.936 & 0.067 & 0.001 & 0.960 & 0.214 & 0.000 & 0.940 & 0.322 & -0.006 \\
WCCA & 2000 & 0.940 & 0.069 & 0.001 & 0.956 & 0.219 & -0.001 & 0.944 & 0.327 & -0.005 \\
WCCA & 5000 & 0.936 & 0.071 & 0.001 & 0.948 & 0.226 & -0.001 & 0.952 & 0.348 & -0.007 \\
\addlinespace
PS-PPI & 500 & 0.934 & 0.033 & 0.000 & 0.942 & 0.107 & 0.002 & 0.946 & 0.160 & -0.002 \\
PS-PPI & 1000 & 0.924 & 0.033 & 0.000 & 0.954 & 0.108 & 0.002 & 0.954 & 0.160 & -0.002 \\
PS-PPI & 2000 & 0.928 & 0.034 & 0.000 & 0.952 & 0.110 & 0.002 & 0.940 & 0.162 & -0.002 \\
PS-PPI & 5000 & 0.932 & 0.034 & 0.000 & 0.954 & 0.112 & 0.001 & 0.954 & 0.167 & -0.003 \\
\bottomrule
\end{tabular}%
}
\end{table}

\begin{figure}
    \centering
    \includegraphics[width=1.0\linewidth]{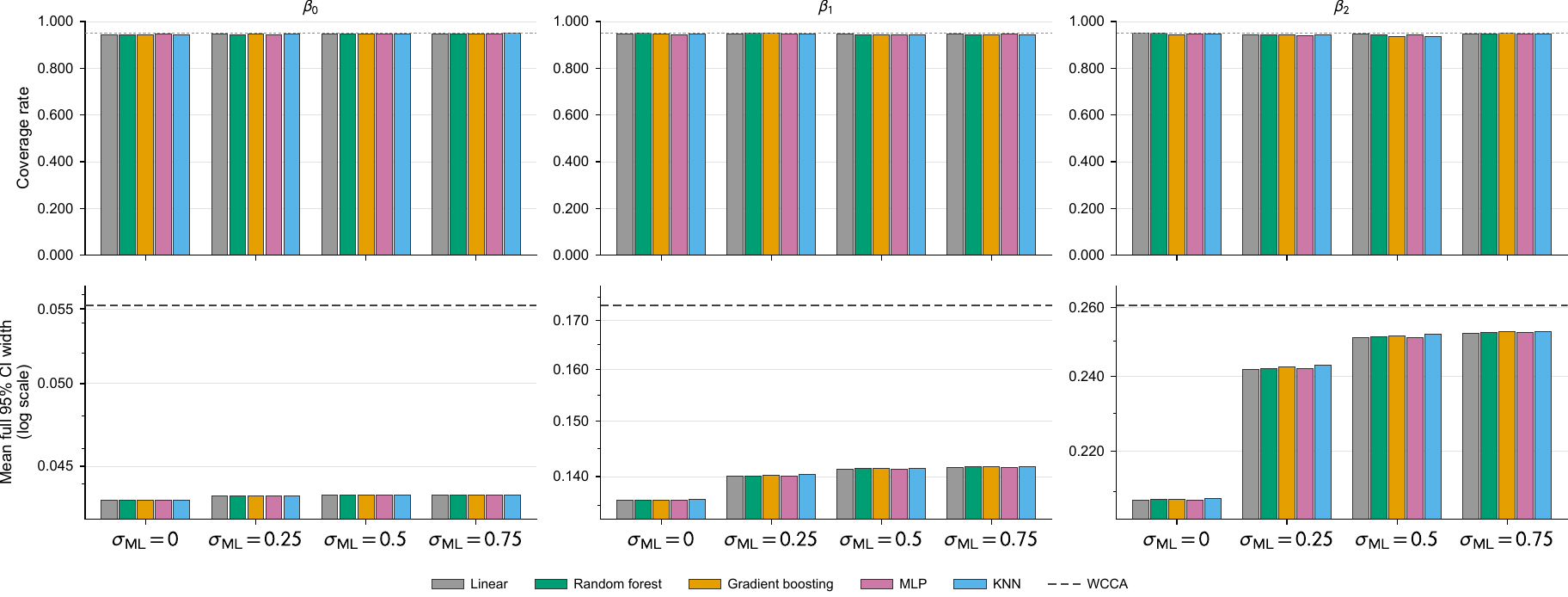}
    \caption{Simulation results of PS-PPI with predictions generated from ML models.}
    \label{fig:prediction_model_z_exp_sin_noise_grouped_bar_N_5000}
\end{figure}

\begin{table}[!t]
\centering
\caption{Predictive performance, measured by $R^2$, of different ML models under different $\sigma_{\text{ML}}$ levels. $R^2_Y$, $R^2_{X_1}$, and $R^2_{X_2}$ denote the $R^2$ for predicting $Y$, $X_1$, and $X_2$, respectively.}
\label{tab:pred-ml-z-exp-sin-r2-by-noise}
\begin{tabular}{clccc}
\toprule
$\sigma_{\text{ML}}$ & ML model & $R^2_Y$ & $R^2_{X_1}$ & $R^2_{X_2}$ \\
\midrule
0 & Linear & 0.173 & 0.004 & 0.990 \\
 & Random forest & 0.166 & 0.001 & 0.989 \\
 & Gradient boosting & 0.154 & -0.017 & 0.989 \\
 & MLP & 0.173 & 0.002 & 0.987 \\
 & KNN & 0.141 & -0.033 & 0.989 \\
\addlinespace
0.25 & Linear & 0.066 & 0.000 & 0.377 \\
 & Random forest & 0.059 & -0.003 & 0.372 \\
 & Gradient boosting & 0.044 & -0.023 & 0.363 \\
 & MLP & 0.063 & -0.001 & 0.375 \\
 & KNN & 0.031 & -0.042 & 0.353 \\
\addlinespace
0.5 & Linear & 0.023 & 0.000 & 0.132 \\
 & Random forest & 0.018 & -0.003 & 0.125 \\
 & Gradient boosting & 0.002 & -0.023 & 0.111 \\
 & MLP & 0.021 & -0.002 & 0.115 \\
 & KNN & -0.014 & -0.040 & 0.097 \\
\addlinespace
0.75 & Linear & 0.011 & 0.000 & 0.063 \\
 & Random forest & 0.006 & -0.004 & 0.057 \\
 & Gradient boosting & -0.011 & -0.023 & 0.039 \\
 & MLP & 0.009 & -0.003 & 0.057 \\
 & KNN & -0.029 & -0.039 & 0.025 \\
\bottomrule
\end{tabular}
\end{table}

\section{Supplementary Figures}

\begin{figure}[H]
    \centering
    \includegraphics[width=1.0\linewidth]{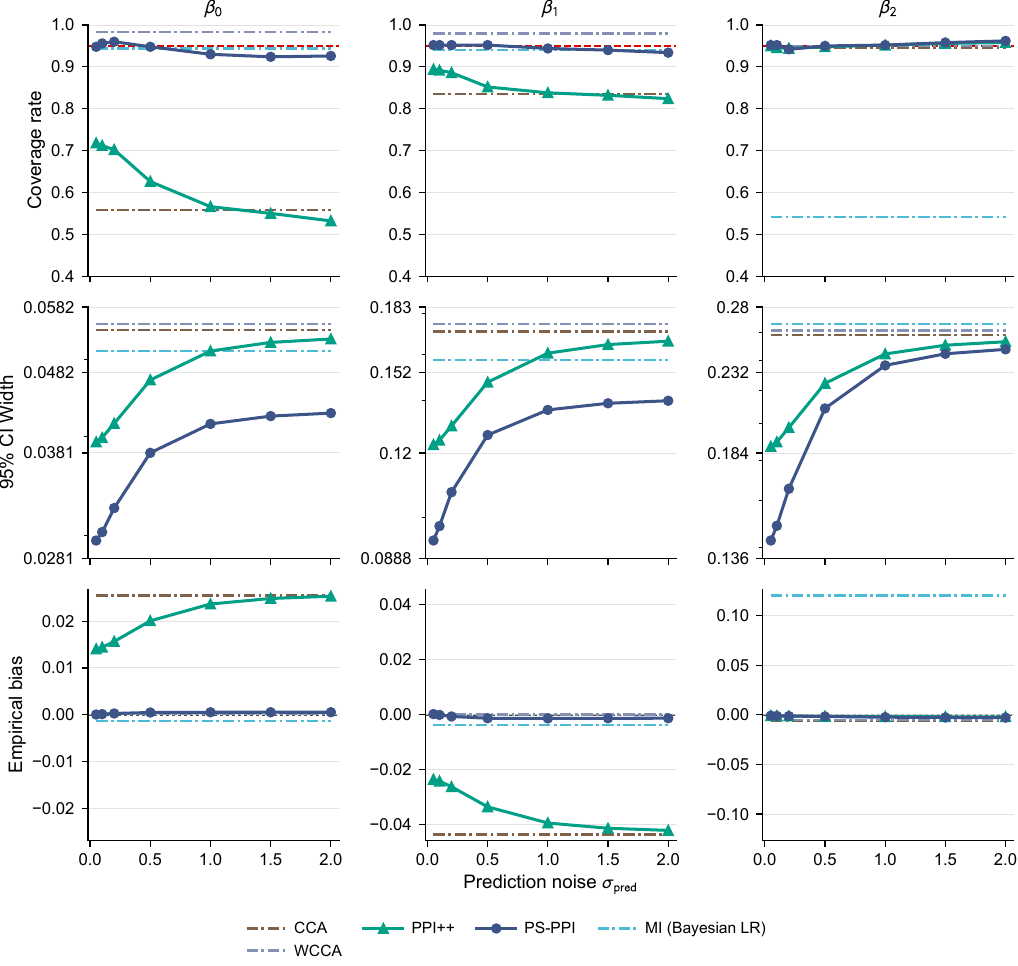}
    \caption{Performance comparison among selected methods on: (a) coverage rates of the 95\% confidence intervals (red dashed lines indicate 95\%); (b) widths of 95\% confidence intervals ($\log_{10}$ transformed);  (c) estimation biases. The X-axis in all the panels represents the prediction noise level $\sigma_{\text{pred}}$ used in the simulations. The prediction bias level $\lambda_\text{pred}$ is fixed at $\lambda_{\text{pred}} = 0$. }
    \label{fig:ml_imputation_quality_simulation_noise_sweep}
\end{figure}

\begin{figure}[!htbp]
    \centering
    \includegraphics[width=1.0\linewidth]{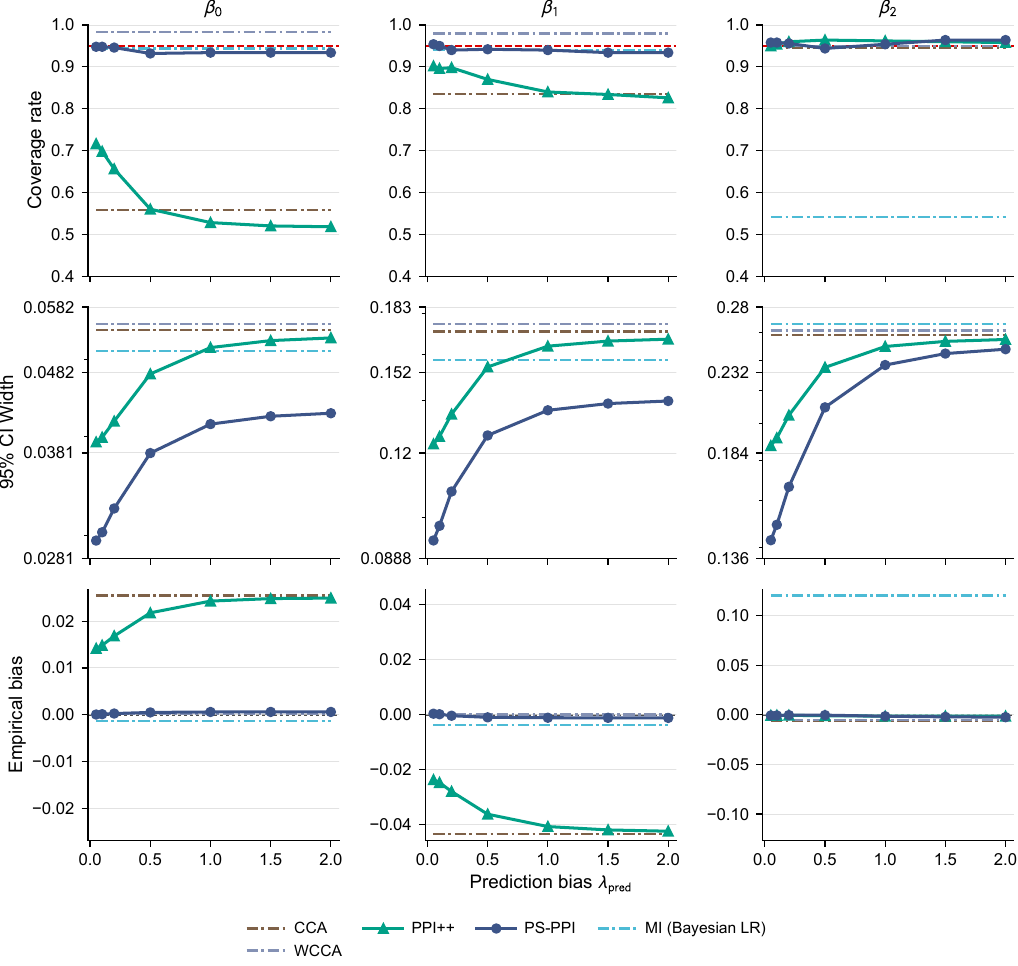}
    \caption{Performance comparison among selected methods on: (a) coverage rates of the 95\% confidence intervals (red dashed lines indicate 95\%); (b) widths of 95\% confidence intervals ($\log_{10}$ transformed);  (c) estimation biases. The X-axis in all the panels represents the prediction bias level $\lambda_{\text{pred}}$ used in the simulations. The prediction noise level $\sigma_\text{pred}$ is fixed at $\sigma_{\text{pred}} = 0$.}
    \label{fig:ml_imputation_quality_simulation_bias_sweep}
\end{figure}

\begin{figure}[!htbp]
    \centering
    \includegraphics[width=1.0\linewidth]{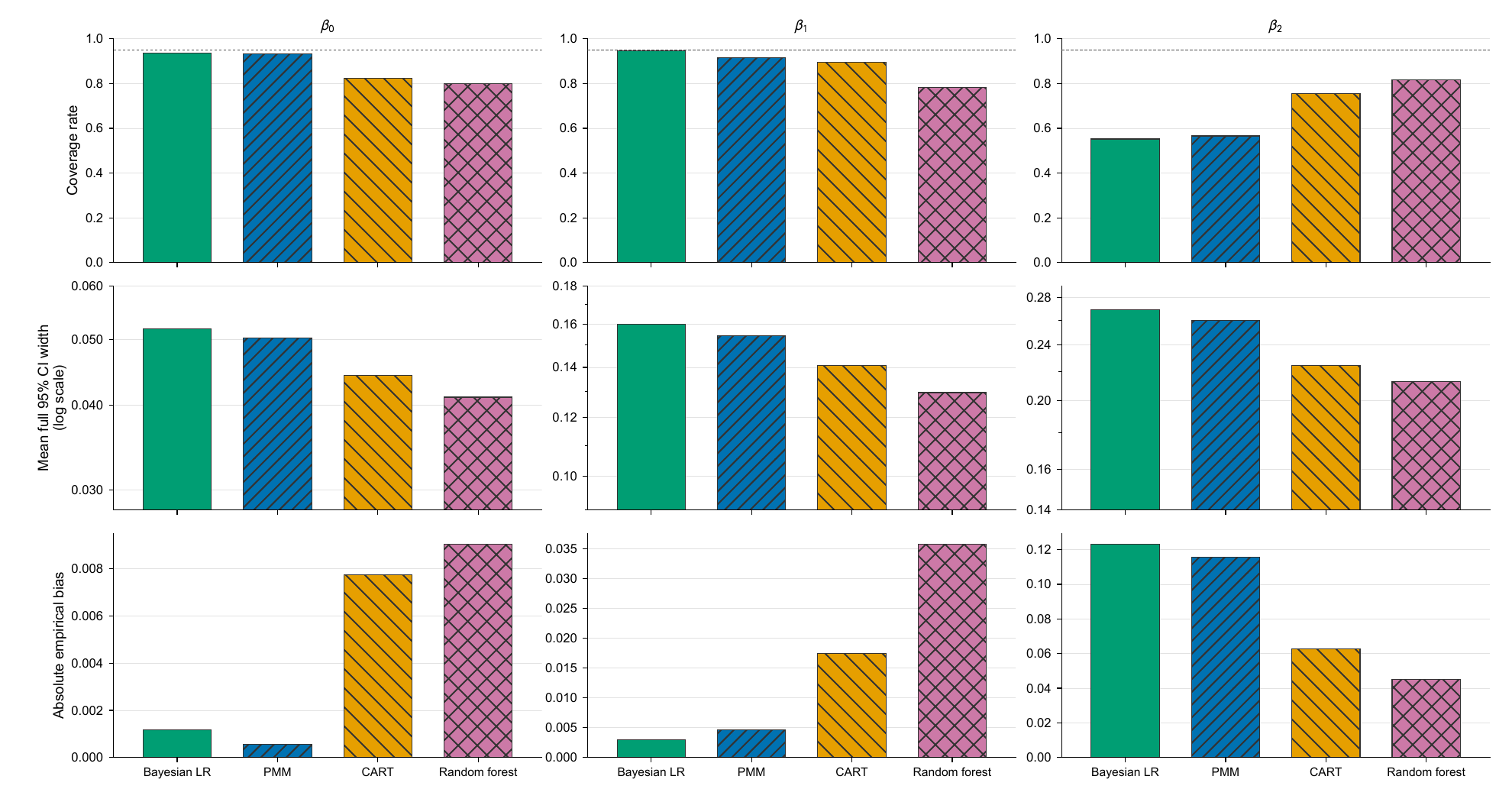}
    \caption{Performance of different imputation models of multiple imputation available in the \texttt{mice} package. The following imputation models are considered: 1) Bayesian linear regression (Bayesian LR); 2) predictive mean matching (PMM); 3) classification and regression trees (CART); and 4) random forests. Each column corresponds to results on $\beta_0$, $\beta_1$, and $\beta_2$, respectively. Each row reports 1) empirical coverage of nominal 95\% confidence intervals; 2) average 95\% confidence interval width (log scale); and 3) absolute empirical bias. The dashed horizontal line indicates the nominal 95\% coverage level. Results are based on 500 replicates with $N=5000$ and five repeated imputations.}
    \label{fig:mi_method_comparison}
\end{figure}

\end{document}